\documentclass[11pt]{article}
\usepackage{amsmath,amssymb}
\usepackage{amsthm}
\usepackage{bbm}
\usepackage{bm}
\usepackage{centernot}
\usepackage{enumitem}
\usepackage{framed}
\usepackage{graphicx}
\usepackage{indentfirst}
\usepackage{interval}
\usepackage{mathtools}
\usepackage[round,longnamesfirst]{natbib}
\usepackage{rotating}
\usepackage{setspace}
\usepackage{subcaption}
\usepackage{threeparttable}
\usepackage{titling}
\usepackage{authblk}%put after ``titling"
\newtheorem{thm}{Theorem}

\newtheorem{asm}{Assumption}

\renewenvironment{proof}[1][\proofname]{{\bfseries #1. }}{\qed}
\newtheorem{lemappB}{Lemma}

\newtheorem{lemappC}{Lemma}

\theoremstyle{definition}
\newtheorem{rem}{Remark}
\title{Testing for Coefficient Randomness in Local-to-Unity Autoregressions}
\author[1]{Mikihito Nishi\footnote{I am greatly indebted to Eiji Kurozumi, my advisor, for discussions and his help, support and encouragement. I am also grateful to Mototsugu Shintani and Yohei Yamamoto for their helpful comments. All errors are mine. Address correspondence to: Graduate School of Economics, Hitotsubashi University, 2-1 Naka, Kunitachi, Tokyo 186-8601, Japan; e-mail: ed225007@g.hit-u.ac.jp}}

\affil[1]{Graduate School of Economics, Hitotsubashi University}
\date{\today}

\begin{document}
\baselineskip= 6mm
	
    \begin{titlingpage}
        \maketitle
        
        \begin{abstract}
            In this study, we propose a test for the coefficient randomness in autoregressive models where the autoregressive coefficient is local to unity, which is empirically relevant given the results of earlier studies. Under this specification, we theoretically analyze the effect of the correlation between the random coefficient and disturbance on tests' properties, which remains largely unexplored in the literature. Our analysis reveals that the correlation crucially affects the power of tests for coefficient randomness and that tests proposed by earlier studies can perform poorly when the degree of the correlation is moderate to large. The test we propose in this paper is designed to have a power function robust to the correlation. Because the asymptotic null distribution of our test statistic depends on the correlation $\psi$ between the disturbance and its square as earlier tests do, we also propose a modified version of the test statistic such that its asymptotic null distribution is free from the nuisance parameter $\psi$. The modified test is shown to have better power properties than existing ones in large and finite samples.
        \end{abstract}

\medskip


\medskip


\end{titlingpage}

\section{Introduction}

In this paper, we consider the random-coefficient autoregressive (RCA) model,
\begin{align}
	y_t = (\rho+ \omega v_t)y_{t-1} + \varepsilon_t, \ \ t=1,2,\cdots,T, \label{model:rca}
\end{align}
where $(\varepsilon_t,v_t)'$ is a random vector with $\mathbb{E}[\varepsilon_t] = \mathbb{E}[v_t] = 0$ and $\mathbb{V}[v_t]=1$. Model \eqref{model:rca} is a generalization of the usual AR(1) model, in that the autoregressive coefficient fluctuates over time around its mean $\rho$ with variance $\omega^2$, instead of being constant at $\rho$. Since \citet{nicholls1982RandomCoefficient}, much attention has been paid to the estimation and inference theory of model \eqref{model:rca}; see, for example, \citet{hwang2005ExplosiveRandomCoefficientb}, \citet{aue2011QuasiLikelihoodEstimation}, \citet{horvath2019Testingrandomnessa} and references therein.

One aspect of the inferential theory for \eqref{model:rca} that has attracted econometricians and statisticians is how to test the hypothesis of $\omega^2=0$, that is, how to test the null hypothesis of the usual autoregressive specification against the RCA alternative. \citet{nicholls1982RandomCoefficient} and \citet{lee1998Coefficientconstancya} proposed test statistics for testing $H_0: \omega^2=0$ under the stationarity condition $\rho^2+\omega^2<1$. \citet{nagakura2009Testingcoefficienta} showed the \citet{lee1998Coefficientconstancya} test statistic has the same asymptotic null distribution when $\rho=1$ as when $|\rho|<1$.\footnote{\citet{nagakura2009Asymptotictheory} also showed the \citet{lee1998Coefficientconstancya} test is consistent when $\rho^2+\omega^2>1$ and some other conditions hold. However, he did not show the test statistic has the same limiting null distribution when $\rho>1$ as when $|\rho|\leq1$.} The condition, $\rho^2+\omega^2\leq1$, which these tests rest on, however, is restrictive and limits their applicability in practice, in view of the empirical analyses by  \citet{hill2014UnifiedIntervalb} and \citet{horvath2019Testingrandomnessa}. They applied model \eqref{model:rca} to several macroeconomic variables and estimated $\rho$ to be near 1 for most of the variables, with $\rho$ estimated to be greater than 1 for some. For instance, the $\rho$ estimates obtained by \citet{horvath2019Testingrandomnessa} took values between 0.9884 and 1.0021. Their results indicate that nonstationary RCA models where $\rho^2+\omega^2$ and $\rho$ may be greater than 1 with $\rho$ near unity are empirically relevant, and thus testing procedures are required that are valid under these conditions. \citet{horvath2019Testingrandomnessa} proposed a test statistic for $H_0: \omega^2=0$ valid under the nonstationarity conditions (as well as under the stationarity conditions).

Another insight earlier studies provide is that the variation of the autoregressive root is smaller (if it exists) than assumed in the extant literature. For example, according to the empirical analysis conducted by \citet{horvath2019Testingrandomnessa}, 4 variables out of 7 were estimated to have positive variance ($\omega^2>0$) in their coefficient, which is between $7.2\times10^{-5}$ and $8.2\times10^{-3}$. However, \citet{horvath2019Testingrandomnessa} studied the finite-sample power of their test only for $\omega^2\geq2.5\times10^{-1}$, which are of larger magnitude than observed in practice (as long as macroeconomic variables are concerned). \footnote{\citet{horvath2022ChangepointDetection} applied the RCA model to a cryptocurrency index and estimated $\omega^2$ to be around $10^{-1}$.} Therefore, it is unclear whether their test performs well when $\omega^2$ is of small magnitude that seems to be typical in empirical studies. In fact, their test has almost no distinguishing power when $\omega^2$ is relatively small, as our simulation studies will reveal.

Given these observations, we consider model \eqref{model:rca} with $\rho$ near unity and $\omega^2$ near zero and employ the following near unit root RCA model:
\begin{align}
	y_t=(\rho_T + \omega_T v_t)y_{t-1} + \varepsilon_t, \label{model:near_unity_rca}
\end{align}
where $\rho_T=1+a/T$ and $\omega_T=c/T^{3/4}$. When $\omega_T^2=0$, \eqref{model:near_unity_rca} is the popular (conventional) local-to-unity AR(1) model. This specification has been useful in analyzing the power of unit root tests \citep{elliott1996Efficienttestsa} and in developing estimation and inference theory for models involving persistent variables \citep[see, \textit{inter alia},][]{phillips1987Unifiedasymptotica,stock1991Confidenceintervals,campbell2006Efficienttests,phillips2014ConfidenceIntervals}. The local-to-zero variance $\omega_T^2=c^2/T^{3/2}$ has been employed under $\rho_T=1$ by \citet{mccabe1998Powertestsa} and \cite{nishi2022StochasticLocal} to derive and compare the local asymptotic power functions of unit root tests of $\omega^2=0$ against $\omega^2>0$. Introducing the local-to-zero variance provides us with a convenient framework in which we can evaluate power functions of several tests against the alternatives that are close to the null of $\omega^2=0$ and seem relevant for empirical applications.\footnote{For example, when $T=200$ and $0<c^2\leq50$, the variance $\omega_T^2$ takes values between 0 and $1.8\times10^{-2}$.} Model \eqref{model:near_unity_rca} integrates these two local-to-unit-root parametrizations, thereby rendering itself an empirically relevant random-coefficient model.

In the literature, testing procedures for $H_0: \omega^2=0$ have often been studied in the special case with $\rho_T=1$. Such a model is called stochastic unit root (STUR) model. Test statistics for $H_0: \omega^2=0$ in the STUR model have been proposed by earlier studies, including \citet{mccabe1995Testingtimea}, \citet{leybourne1996Caneconomica} and \citet{distaso2008Testingunit}. Moreover, \citet{mccabe1998Powertestsa} and \citet{nishi2022StochasticLocal} analyzed the power properties of several tests for $H_0: \omega^2=0$ under the STUR modelling. On the one hand, the STUR model is a generalization of the pure unit root model, which is a main reason authors have paid much attention to it. But on the other hand, the assumption of $\rho_T$ exactly being unity is a restrictive condition that is unlikely to hold in empirical analysis. Our model, \eqref{model:near_unity_rca}, generalizes the STUR model by allowing $\rho_T$ to take a value different from unity (but near it), which is a more realistic assumption given the empirical analyses conducted by earlier studies.

\citet{nishi2022StochasticLocal} revealed that, under the STUR ($\rho_T=1$) modelling with $\mathrm{Corr}(\varepsilon_t,v_t)=0$, the test for $H_0: \omega^2=0$ proposed by \citet{lee1998Coefficientconstancya} and \cite{nagakura2009Testingcoefficienta} (hereafter, LN) has a higher local asymptotic power function than other tests. We will demonstrate that this is also the case when $\rho_T = 1+a/T$ and $\mathrm{Corr}(\varepsilon_t,v_t)=0$. In this paper, however, it will be shown that the LN test can perform poorly when $\mathrm{Corr}(\varepsilon_t,v_t)\neq0$, for both the cases $\rho_T=1$ and $\rho_T=1+a/T$. We will therefore propose several tests for $H_0: \omega^2=0$ whose power properties are robust to the value of $\mathrm{Corr}(\varepsilon_t,v_t)$. One of those tests turns out to be more powerful than the LN test for moderate to large values of $\mathrm{Corr}(\varepsilon_t,v_t)$. To the best of our knowledge, this study is the first to investigate how the value of $\mathrm{Corr}(\varepsilon_t,v_t)$ affects the power properties of tests for $H_0: \omega^2=0$, with the only exception of \citet{su2012Examiningpowera}, who analyzed through simulation this effect in finite samples under $\rho=1$.

The other issue that this paper tackles is how to remove the effect of nuisance parameters from the null distributions of the test statistics mentioned above. As pointed out by \cite{nagakura2009Testingcoefficienta}, the limiting null distribution of the LN test statistic depends on the correlation between $\varepsilon_t$ and $\varepsilon_t^2-\sigma_{\varepsilon}^2$, $\mathrm{Corr}(\varepsilon_t,\varepsilon_t^2-\sigma_{\varepsilon}^2)$, and so do the test statistics proposed in this article. If the true value of $\rho_T$ is known, the nuisance parameter can be removed straightforwardly by a similar way to the modification proposed by \citet{nagakura2009Testingcoefficienta}, but this is not the case when the true $\rho_T$ is unknown. This problem is caused by the fact that the localizing parameter $a$ cannot be consistently estimated, as will be made clear later. As a solution for this problem, we propose a testing procedure based on the Bonferroni approach as \citet{campbell2006Efficienttests} and \citet{phillips2014ConfidenceIntervals} did in the context of predictive regressions.

The remainder of this paper is organized as follows. In Section 2, we consider the local-to-unity model \eqref{model:near_unity_rca} with the true value of $\rho_T$ (or equivalently $a$) known, to uncover the effect $\mathrm{Corr}(\varepsilon_t,v_t)$ has on power properties of tests for $H_0: \omega^2=0$ and propose new tests that have power functions robust to this effect. We also propose a modification to make these tests independent of $\mathrm{Corr}(\varepsilon_t,\varepsilon_t^2-\sigma_{\varepsilon}^2)$, a nuisance parameter, under the null. In Section 3, we consider model \eqref{model:near_unity_rca} with unknown $\rho_T$. Because the modification proposed in the preceding section cannot be directly applied in this case due to the fact that $a$ is not consistently estimable, we base our tests on the Bonferroni approach by constructing a confidence interval for $a$. Section 4 analyzes the finite-sample properties of our tests and compares them with those of existing tests. In Section 5, we apply our tests to real data. Section 6 concludes.

\section{The Case of Known $\rho_T$}
\subsection{The effect of $\mathrm{Corr}(\varepsilon_t,v_t)$}
In this section, we begin our analysis under the assumption that the true value of $\rho_T$ is known. This assumption will be relaxed in the next section. Our analysis for model \eqref{model:near_unity_rca} in this and the next section is conducted under the following assumption.
\begin{asm}
	$(\varepsilon_t, v_t)' \sim \mathrm{i.i.d} \ (0,\Omega)$, where 
	\begin{align*}
		\Omega \coloneqq \begin{pmatrix}
			\sigma_\varepsilon^2 & \sigma_{\varepsilon v} \\
			\sigma_{\varepsilon v} & 1
		\end{pmatrix}.
	\end{align*}
	 Also, $\mathbb{E}[\varepsilon_t^4] < \infty$ and $\mathbb{E}[v_t^8] < \infty$. Moreover, $y_0 = o_p(T^{1/2})$. \label{asm:local}
\end{asm}
Define $\eta_t \coloneqq \varepsilon_t^2 - \sigma_{\varepsilon}^2$ and $\sigma_\eta^2 \coloneqq \mathbb{E}[\eta_t^2]$. Define also the partial sum process $(W_{\varepsilon,T}, W_{\eta,T})'$ on $[0,1]$ by $W_{\varepsilon,T}(r) \coloneqq T^{-1/2}\sigma_{\varepsilon}^{-1}\sum_{t=1}^{\lfloor Tr\rfloor}\varepsilon_t$ and $W_{\eta,T}(r) \coloneqq T^{-1/2}\sigma_\eta^{-1}\sum_{t=1}^{\lfloor Tr\rfloor}\eta_t$. Then, it follows from the functional central limit theorem (FCLT) that under Assumption \ref{asm:local}, as $T \to \infty$
\begin{equation*}
	\begin{pmatrix}
		W_{\varepsilon,T} \\ W_{\eta,T}
	\end{pmatrix}
	\Rightarrow \begin{pmatrix}
		W_{\varepsilon} \\ W_{\eta}
	\end{pmatrix},
\end{equation*}
in the Skorokhod space $D[0,1]$, where $(W_{\varepsilon},W_{\eta})'$ is a vector Brownian motion with the covariance coefficient
\begin{align*}
	\begin{pmatrix}
		1 & \psi \\
		\psi & 1
	\end{pmatrix}
\end{align*}
with $\psi \coloneqq \mathbb{E}[\varepsilon_t\eta_t]/(\sigma_{\varepsilon}\sigma_\eta)$. Note that $W_\varepsilon$ and $W_\eta$ are not necessarily independent because of the covariance $\psi$, and the Brownian motion $W_\eta$ satisfies the following equality in distribution:
\begin{align}
	W_\eta \stackrel{d}{=} \psi W_\varepsilon + \sqrt{1-\psi^2}W_1, \label{eqn:w_eta_varepsilon_1}
\end{align}
where $W_1$ is a standard Brownian motion independent of $W_\varepsilon$. Following the argument by \citet{nishi2022StochasticLocal}, we can show that under Assumption \ref{asm:local}, the standardized process $T^{-1/2}y_{\lfloor T\cdot \rfloor}$ on $[0,1]$ constructed from \eqref{model:near_unity_rca} weakly converges to the Ornstein-Uhlenbeck (OU) process $\sigma_{\varepsilon}J_a(\cdot)$, where $J_a$ solves $dJ_a(r) = aJ_a(r)dr + dW_\varepsilon(r)$. We can also construct consistent estimators of variances, namely, $\sigma_{\varepsilon}^2$ and $\sigma_\eta^2$. Define $z_t(\rho_T) \coloneqq y_t - \rho_T y_{t-1}(=\omega_Tv_ty_{t-1}+\varepsilon_t)$. Then, the estimators $\hat{\sigma}_{\varepsilon,T}^2(\rho_T) \coloneqq T^{-1} \sum_{t=1}^{T}z_t^2(\rho_T)$ and $\hat{\sigma}^2_{\eta,T}(\rho_T) \coloneqq T^{-1}\sum_{t=1}^{T}\{z_t^2(\rho_T) - \hat{\sigma}_{\varepsilon,T}^2(\rho_T)\}^2$ are consistent for $\sigma_{\varepsilon}^2$ and $\sigma_\eta^2$, respectively, which is proven in Appendix B.

Several tests of $H_0: \omega_T^2=0$ for model \eqref{model:near_unity_rca} have been proposed by earlier work such as LN. The LN test statistic is defined by
\begin{align*}
	\mathrm{LN}_T(\rho_T) \coloneqq \frac{\sum_{t=1}^{T}(y_{t-1}^2-T^{-1}\sum_{t=1}^{T}y_{t-1}^2)z_t^2(\rho_T)}{\hat{\sigma}_{\eta,T}(\rho_T)\{\sum_{t=1}^{T}(y_{t-1}^2-T^{-1}\sum_{t=1}^{T}y_{t-1}^2)^2\}^{1/2}}. 
\end{align*}
\citet{nishi2022StochasticLocal} found that the LN test has a high power function 
under the assumption that $\sigma_{\varepsilon v}=0$ and $\rho_T=1$.

As pointed out by \citet{nishi2022StochasticLocal}, the LN test, which was originally derived as a locally best invariant test, can also be obtained as the t-test for $H_0: \omega_T^2=0$ under $\sigma_{\varepsilon v}=0$. To see this, note that from model \eqref{model:near_unity_rca}, a simple calculation gives
	\begin{align}
		z_t^2(\rho_T) = \sigma_{\varepsilon}^2 + \omega_T^2y_{t-1}^2 + \xi_t, \label{model:linearized}
	\end{align}
where $\xi_t \coloneqq \omega_T^2y_{t-1}^2(v_t^2-1) + 2\omega_Ty_{t-1}\varepsilon_tv_t + (\varepsilon_t^2-\sigma_{\varepsilon}^2)$. Because $\mathbb{E}[\xi_t]=0$ and $\mathbb{E}[y_{t-1}^2\xi_t]=0$ under Assumption \ref{asm:local} with $\sigma_{\varepsilon v}=0$, model \eqref{model:linearized} may be viewed as the linear regression model with $\xi_t$ playing the role of the disturbance, and the LN test statistic is obtained as the t-test for $H_0: \omega_T^2=0$ (with the variance estimator under the null used).

In view of this observation, the LN test seems to be a natural test for $H_0: \omega_T^2=0$. However, this may not be the case when $\sigma_{\varepsilon v}\neq0$, because it results in  $\mathbb{E}[y_{t-1}^2\xi_t]=2\omega_T\sigma_{\varepsilon v}\mathbb{E}[y_{t-1}^3]\neq0$, that is, endogeneity. In fact, as we will shortly show, the power of the LN test is crucially affected by the value of $\sigma_{\varepsilon v}$, and the greater the value of $\sigma_{\varepsilon v}$ (in absolute value), the more poorly the LN test performs. Because we consider the localized model \eqref{model:near_unity_rca}, the setup for analyzing the influence $\sigma_{\varepsilon v}$ has is also a localized one. To derive relevant local asymptotic distributions, we localize the correlation (rather than the covariance $\sigma_{\varepsilon v}$) between $\varepsilon_t$ and $v_t$ in the following way:
\begin{asm}
	$\mathrm{Corr}(\varepsilon_t,v_t) = \sigma_{\varepsilon v}/\sigma_\varepsilon = q/T^{1/4}$. \label{asm:localized_corr}
\end{asm} 
Here, $q$ is interpreted as the correlation coefficient in the limit as $T\to \infty$. Under this localization, the LN test statistic has the following asymptotic distribution:
\begin{thm}
	Consider model \eqref{model:near_unity_rca}. Under Assumptions \ref{asm:local} and \ref{asm:localized_corr}, we have
	\begin{align}
		\mathrm{LN}_T(\rho_T) \Rightarrow \frac{\int_{0}^{1}\widetilde{J_{a,2}}(r)dW_\eta(r)}{\bigl\{\int_{0}^{1}\bigl(\widetilde{J_{a,2}}\bigr)^2(r)dr\bigr\}^{1/2}} +  \frac{\sigma_\varepsilon^2}{\sigma_\eta}\Biggl[\frac{c^2\int_{0}^{1}\bigl(\widetilde{J_{a,2}}\bigr)^2(r)dr + 2cq\int_{0}^{1}\widetilde{J_{a,2}}(r)\widetilde{J_{a,1}}(r)dr}{\bigl\{\int_{0}^{1}\bigl(\widetilde{J_{a,2}}\bigr)^2(r)dr\bigr\}^{1/2}}\Biggr], \label{dist:ln}
	\end{align}
	where $\widetilde{J_{a,1}}(r) \coloneqq J_a(r) - \int_{0}^{1}J_a(s)ds$ and
	$\widetilde{J_{a,2}}(r) \coloneqq J_a^2(r) - \int_{0}^{1}J_a^2(s)ds$. \label{thm:ln_dist}
\end{thm} 

Note that when $q=0$ and $\rho_T=1$ (or $a=0$), the limit distribution reduces to the one \citet{nishi2022StochasticLocal} derived (their equation (19)). \citet{nishi2022StochasticLocal} found that the LN test performs better than other tests when $q=0$. To see how the value of $q$ alters the LN test's power properties, we simulate the asymptotic distribution in \eqref{dist:ln} by 100,000 replications. The replications are based on $\varepsilon_t \sim \mathrm{i.i.d} \ N(0,1)$, so that $\sigma_\varepsilon^2=1$, $\sigma_\eta^2=2$ and $\psi=\mathbb{E}[\varepsilon_t\eta_t]/(\sigma_\varepsilon\sigma_\eta)=0$. Note that the effect of $q$ on the limit distribution is symmetric when $\psi=0$; that is, $q=\bar{q}$ and $q=-\bar{q}$ produce the identical distribution. This is because $(W_\varepsilon,W_\eta)' \stackrel{d}{=}(-W_\varepsilon,W_\eta)'$ when $\psi=0$, and hence $(J_a,W_\eta)' \stackrel{d}{=} (-J_a,W_\eta)'$ in this case. Thus, in this simulation, and in the replications conducted later where $\psi=0$ holds, we only consider positive values of $q$. Figure \ref{fig:power_asym_Lee} shows the power functions of the LN test for $q=0,1,2,3$ under $a=0$. One noticeable feature is that as $q$ gets larger, the power function gets lower and flatter for $c^2\geq 5$, while the power becomes greater for $c^2\leq5$. Since the power gains over $c^2\leq5$ are obviously outweighed by the power losses over $c^2\geq5$, we conclude that the LN test performs poorly when $q$ is large (in absolute value).\footnote{In our unreported simulations, we also found a similar tendency in the power function of \citet{distaso2008Testingunit}'s LM test, which is based on the assumption of $\varepsilon_t$ and $v_t$ being independent. The results are omitted to save space.}

The LN test's poor performance for large $q$ could be attributed to the endogeneity $\mathbb{E}[y_{t-1}^2\xi_t]=2c\sigma_\varepsilon qT^{-1}\mathbb{E}[y_{t-1}^3] \neq 0$ in \eqref{model:linearized}. One solution for this endogeneity is to augment model \eqref{model:linearized} by adding $y_{t-1}$ as a regressor,
\begin{align}
		z_t^2(\rho_T) = \sigma_{\varepsilon}^2 + \delta_Ty_{t-1} + \omega_T^2y_{t-1}^2 + \xi_t^*, \label{model:linearized_augmented}
\end{align}
where $\delta_T \coloneqq 2c\sigma_\varepsilon qT^{-1}$ and $\xi^*_t \coloneqq \omega_T^2y_{t-1}^2(v_t^2-1) + 2\omega_Ty_{t-1}(\varepsilon_tv_t - \sigma_{\varepsilon v}) + (\varepsilon_t^2-\sigma_{\varepsilon}^2)$. Because this augmented model is free from the endogeneity, that is, $\mathbb{E}[\xi_t^*] = \mathbb{E}[y_{t-1}\xi_t^*] = \mathbb{E}[y_{t-1}^2\xi_t^*]=0$ under Assumption \ref{asm:local}, we propose using the t-test for $H_0: \omega_T^2=0$ in model \eqref{model:linearized_augmented}. We also propose using the Wald test for $(\delta_T,\omega_T^2)=(0,0)$, because $\omega_T^2=0$ if and only if $(\delta_T,\omega_T^2)=(0,0)$. To express the t and Wald test statistics, first regress $z_t^2(\rho_T)$, $y_{t-1}$ and $y_{t-1}^2$ on a constant:
\begin{align}
		\widetilde{z_t^2}(\rho_T) = \delta_T\widetilde{y_{t-1}} + \omega_T^2\widetilde{y_{t-1}^2} + \widetilde{\xi_t^*}, \label{model:linearized_augmented_demeaned}
\end{align}
where $\widetilde{z_t^2}(\rho_T) \coloneqq z_t^2(\rho_T) - T^{-1}\sum_{t=1}^{T}z_t^2(\rho_T)$, and $\widetilde{y_{t-1}}$, $\widetilde{y_{t-1}^2}$ and $\widetilde{\xi_t^*}$ are defined similarly. Let $\widetilde{Z_2}(\rho_T) \coloneqq (\widetilde{z_1^2}(\rho_T),\widetilde{z_2^2}(\rho_T),\ldots,\widetilde{z_T^2}(\rho_T))'$, $\widetilde{X_1} \coloneqq (\widetilde{y_{0}},\widetilde{y_{1}},\ldots,\widetilde{y_{T-1}})'$, $\widetilde{X_2} \coloneqq (\widetilde{y_{0}^2},\widetilde{y_{1}^2},\ldots,\widetilde{y_{T-1}^2})'$, and $\widetilde{\Xi^*} \coloneqq (\widetilde{\xi_1^*},\widetilde{\xi_2^*},\ldots,\widetilde{\xi_T^*})'$. Then, model \eqref{model:linearized_augmented_demeaned} can be expressed in matrix notation as
\begin{align}
	\widetilde{Z_2}(\rho_T) = \widetilde{X}\theta_T+\widetilde{\Xi^*}, \label{model:linearized_augmented_matrix}
\end{align}
where $\widetilde{X} \coloneqq (\widetilde{X_1},\widetilde{X_2})$ and $\theta_T \coloneqq (\delta_T,\omega_T^2)'$. The t and Wald test statistics are then given by
\begin{align*}
	t_{\hat{\omega}_T^2}(\rho_T) \coloneqq \hat{\sigma}_{\widetilde{\xi^*}}^{-1}(\rho_T)(\widetilde{X_2}'M_1\widetilde{X_2})^{-1/2}\widetilde{X_2}'M_1\widetilde{Z_2}(\rho_T) \\
	\intertext{and}
	W_T(\rho_T) \coloneqq \frac{\hat{\theta}_T(\rho_T)'(\widetilde{X}'\widetilde{X})\hat{\theta}_T(\rho_T)}{\hat{\sigma}_{\widetilde{\xi^*}}^2(\rho_T)},
\end{align*}
where $M_1 \coloneqq I_T - \widetilde{X_1}(\widetilde{X_1}'\widetilde{X_1})^{-1}\widetilde{X_1}'$, and $\hat{\sigma}_{\widetilde{\xi^*}}^2(\rho_T)$ and $\hat{\theta}_T(\rho_T)$ are the OLS variance and coefficient estimators of \eqref{model:linearized_augmented_matrix}, respectively. We shall call the t and Wald tests in model \eqref{model:linearized_augmented} augmented t and Wald tests, respectively. The limiting distributions of these augmented test statistics are collected in the following theorem.

\begin{thm}
	Consider model \eqref{model:near_unity_rca}. Under Assumptions \ref{asm:local} and \ref{asm:localized_corr}, we have
	\begin{align}
		t_{\hat{\omega}_T^2}(\rho_T) &\Rightarrow \frac{\int_{0}^{1}Q_a(r)dW_\eta(r)}{\bigl[\int_{0}^{1}Q_a^2(r)dr\bigr]^{1/2}} + \frac{c^2\sigma_\varepsilon^2}{\sigma_\eta}\Biggl[\int_{0}^{1}Q_a^2(r)dr\Biggr]^{1/2}, \label{dist:t}
	\end{align}
	where
	\begin{align*}
		Q_a(r) \coloneqq \widetilde{J_{a,2}}(r) -  \frac{\int_{0}^{1}\widetilde{J_{a,1}}(s)\widetilde{J_{a,2}}(s)ds}{\int_{0}^{1}\bigl(\widetilde{J_{a,1}}\bigr)^2(s)ds}\widetilde{J_{a,1}}(r),
	\end{align*}
	and
	\begin{align}
		W_T(\rho_T) \Rightarrow& \Bigg\{\begin{pmatrix}
			\int_{0}^{1}\widetilde{J_{a,1}}(r)dW_\eta(r) \\ \int_{0}^{1}\widetilde{J_{a,2}}(r)dW_\eta(r) \end{pmatrix} + \frac{\sigma_\varepsilon^2}{\sigma_\eta} \begin{pmatrix} c^2\int_{0}^{1}\widetilde{J_{a,1}}(r)\widetilde{J_{a,2}}(r)dr+2cq\int_{0}^{1}\bigl(\widetilde{J_{a,1}}\bigr)^2(r)dr \\  c^2\int_{0}^{1}\bigl(\widetilde{J_{a,2}}\bigr)^2(r)dr+2cq\int_{0}^{1}\widetilde{J_{a,1}}(r)\widetilde{J_{a,2}}(r)dr
		\end{pmatrix}\Biggr\}' \notag \\ &\times \begin{pmatrix}
			\int_{0}^{1}\bigl(\widetilde{J_{a,1}}\bigr)^2(r)dr & \int_{0}^{1}\widetilde{J_{a,1}}(r)\widetilde{J_{a,2}}(r)dr \\ \int_{0}^{1}\widetilde{J_{a,2}}(r)\widetilde{J_{a,1}}(r)dr & \int_{0}^{1}\bigl(\widetilde{J_{a,2}}\bigr)^2(r)dr
		\end{pmatrix}^{-1} \notag \\ &\times \Bigg\{\begin{pmatrix}
			\int_{0}^{1}\widetilde{J_{a,1}}(r)dW_\eta(r) \\ \int_{0}^{1}\widetilde{J_{a,2}}(r)dW_\eta(r) \end{pmatrix} + \frac{\sigma_\varepsilon^2}{\sigma_\eta} \begin{pmatrix} c^2\int_{0}^{1}\widetilde{J_{a,1}}(r)\widetilde{J_{a,2}}(r)dr+2cq\int_{0}^{1}\bigl(\widetilde{J_{a,1}}\bigr)^2(r)dr \\  c^2\int_{0}^{1}\bigl(\widetilde{J_{a,2}}\bigr)^2(r)dr+2cq\int_{0}^{1}\widetilde{J_{a,1}}(r)\widetilde{J_{a,2}}(r)dr
		\end{pmatrix}\Biggr\}. \label{dist:wald}
	\end{align}
	In particular, the asymptotic null distributions ($c^2=0$) are
		\begin{align*}
		t_{\hat{\omega}_T^2}(\rho_T) &\Rightarrow \frac{\int_{0}^{1}Q_a(r)dW_\eta(r)}{\bigl[\int_{0}^{1}Q_a^2(r)dr\bigr]^{1/2}}
	\end{align*}
	and
	\begin{align*}
		W_T(\rho_T) \Rightarrow \begin{pmatrix}
			\int_{0}^{1}\widetilde{J_{a,1}}(r)dW_\eta(r) \\ \int_{0}^{1}\widetilde{J_{a,2}}(r)dW_\eta(r) \end{pmatrix}'  \begin{pmatrix}
			\int_{0}^{1}\bigl(\widetilde{J_{a,1}}\bigr)^2(r)dr & \int_{0}^{1}\widetilde{J_{a,1}}(r)\widetilde{J_{a,2}}(r)dr \\ \int_{0}^{1}\widetilde{J_{a,2}}(r)\widetilde{J_{a,1}}(r)dr & \int_{0}^{1}\bigl(\widetilde{J_{a,2}}\bigr)^2(r)dr
		\end{pmatrix}^{-1} \begin{pmatrix}
			\int_{0}^{1}\widetilde{J_{a,1}}(r)dW_\eta(r) \\ \int_{0}^{1}\widetilde{J_{a,2}}(r)dW_\eta(r) \end{pmatrix}, 
	\end{align*}
	which are standard normal and chi square with 2 degrees of freedom, respectively, when $\psi=0$.
 \label{thm:t_wald_dist}
\end{thm}
There are two points worth mentioning. First, the augmented t test statistic in  model \eqref{model:linearized_augmented} has the asymptotic null distribution independent of $q$. This is a direct result of the augmentation, which is intended to remove the endogeneity from the linearized model \eqref{model:linearized}. When performing the augmented t test, we regress the endogenous regressor $\widetilde{y_{t-1}^2}$ (and $\widetilde{z_t^2}(\rho_T)$) on $\widetilde{y_{t-1}}$. In the limit, this projection amounts to replacing $\widetilde{J_{a,2}}$ in the limiting distribution of $\mathrm{LN}_T$ with $Q_a$, the residual from the linear projection of $\widetilde{J_{a,2}}$ on $\widetilde{J_{a,1}}$ in the Hilbert space (see, for example, \citet{phillips1990AsymptoticProperties}).

The second point is that the limiting null distributions of the augmented t and Wald test statistics are standard normal and chi square with 2 degrees of freedom, respectively, when $\psi=\mathbb{E}[\varepsilon_t\eta_t]/(\sigma_\varepsilon \sigma_\eta)=0$. This will be seen immediately upon noting $J_a$ and $W_\eta$ are independent when $\psi=0$ (due to the independence between $W_\varepsilon$ and $W_\eta$), and the limit distributions under $c^2=0$
conditional on $W_\varepsilon$ are standard normal and chi square with 2 degrees of freedom (and so are they unconditionally). Along the lines of this argument, the asymptotic null distribution of the LN test statistic is seen to be standard normal when $\psi=0$.

\subsection{Removing $\psi$ from the limiting null distributions}
Unless $\psi=0$, the test statistics discussed thus far have asymptotic null distributions dependent on $\psi$. Actually, this dependence stems from the strong persistence of the regressors $y_{t-1}$ and $y_{t-1}^2$ and the long run endogeneity present in the linearized models \eqref{model:linearized} and \eqref{model:linearized_augmented}. To illustrate this, consider model \eqref{model:linearized} and the LN test. Under the null of $\omega_T^2=0$, the model reduces to
\begin{align}
	z_t^2(\rho_T) = \sigma_{\varepsilon}^2 + \omega_T^2y_{t-1}^2 + \eta_t, \label{model:linearized_null}
\end{align}
where $\eta_t=\varepsilon_t^2-\sigma_{\varepsilon}^2$ and $y_{t} = \rho_T y_{t-1} + \varepsilon_t$. This model is free from endogeneity because $\mathbb{E}[\eta_t] = \mathbb{E}[y_{t-1}^2\eta_t]=0$. However, the regressor $y_{t-1}^2$ is, in a sense, ``endogenous" in the limit as $T\to\infty$, because of the correlation $\psi$ between its innovation $\varepsilon_t$ and the disturbance $\eta_t=\varepsilon_t^2-\sigma_{\varepsilon}^2$. Indeed, the LN test statistic becomes
\begin{align*}
	\mathrm{LN}_T(\rho_T) = \frac{\sum_{t=1}^{T}\widetilde{y_{t-1}^2}\eta_t}{\hat{\sigma}_{\eta,T}(\rho_T)\{\sum_{t=1}^{T}(\widetilde{y_{t-1}^2})^2\}^{1/2}} \Rightarrow \frac{\int_{0}^{1}\widetilde{J_{a,2}}(r)dW_\eta(r)}{\bigl\{\int_{0}^{1}\bigl(\widetilde{J_{a,2}}\bigr)^2(r)dr\bigr\}^{1/2}},
\end{align*}
where $\widetilde{J_{a,2}}(r)$ ($\approx T^{-1}\widetilde{y_{\lfloor Tr \rfloor}^2}$) is correlated with the differential $dW_\eta(r)$ ($\approx T^{-1/2}\eta_{\lfloor Tr \rfloor}$). This correlation originates from that between $W_\varepsilon$ and $W_\eta$, which is denoted by $\psi$. Therefore, the correlation between the regressor's innovation $\varepsilon_t$ and the regression disturbance $\eta_t$ affects the test statistic's behavior in the limit. 

Interestingly, the situation we are in is analogous to the one that has been considered in the literature on predictive regressions. In predictive regressions, a main aim is typically to investigate whether stock returns ($r_t$) can be predicted by another economic or financial variable ($x_t$). To test the predictability of $r_t$, predictive regressions involve regressing $r_t$ on a constant and the lag of $x_t$:
\begin{align*}
	r_t=\mu + \beta x_{t-1} + u_{r,t}.
\end{align*}
Here, $\beta$ represents the predictability of $r_t$; the stock return is not predictable by $x_t$ if $\beta=0$. In the literature on predictive regressions, it has been well known that the usual t test for the hypothesis $\beta=0$ can be misleading when the regressor $x_t$ is persistent, or has a generating mechanism of the form $x_t = \rho_{x,T}x_{t-1} + u_{x,t}$ with $\rho_{x,T}= 1+a_x/T$, and its innovation $u_{x,t}$ is correlated with the regression disturbance $u_{r,t}$. The problem arising in such a case is that the asymptotic null distribution of the t statistic depends on $\mathrm{Corr}(u_{x,t},u_{r,t})$ and is not standard normal unless $\mathrm{Corr}(u_{x,t},u_{r,t})=0$ \citep[see, for example,][]{campbell2006Efficienttests}. Our situation here is essentially the same: the regressor $y_{t-1}^2$ in \eqref{model:linearized_null} is persistent with the mechanism $y_t=\rho_Ty_{t-1}+\varepsilon_t$ and its innovation $\varepsilon_t$ is correlated with the regression disturbance $\eta_t$, which results in the test statistic having the limiting null distribution dependent on the correlation $\psi=\mathrm{Corr}(\varepsilon_t,\eta_t)$. 

For the predictive regression model, there is an extensive literature on this problem, and numerous solutions have been proposed \citep{campbell2006Efficienttests,phillips2013Predictiveregression,phillips2014ConfidenceIntervals,kostakis2015RobustEconometric}. For our case, fortunately, one of those solutions can be applied. Specifically, following \citet{campbell2006Efficienttests}, we modify the test statistics (LN, augmented t and augmented Wald) so that their asymptotic null distributions are standard normal and chi square with 2 degrees of freedom, irrespective of the value of $\psi$.  To explain the idea, take the LN test as an example. Note that under the null of $c^2=0$, $z_t^2(\rho_T)$ in the numerator of the LN test statistic may be asymptotically expressed as
\begin{align*}
	z_{\lfloor Tr \rfloor}^2(\rho_T) = \sigma_\varepsilon^2 + (\varepsilon_{\lfloor Tr \rfloor}^2 - \sigma_\varepsilon^2) &\approx \sigma_\varepsilon^2 +  T^{1/2}\times \sigma_\eta dW_\eta(r) \\
	&\stackrel{d}{=}\sigma_\varepsilon^2 + T^{1/2}\times \sigma_\eta\psi dW_\varepsilon(r) + T^{1/2}\times \sigma_\eta\sqrt{1-\psi^2}dW_1(r),
\end{align*}
given the distributional equivalence \eqref{eqn:w_eta_varepsilon_1}. This observation leads us to propose the following modification to remove the effect of $\psi$:
\begin{align}
	z_t^2(\rho_T) \to z_t^{2*}(\rho_T) \coloneqq \frac{z_t^2(\rho_T) - \hat{\sigma}_{\eta,T}(\rho_T)\hat{\psi}_T(\rho_T)(z_t(\rho_T)/\hat{\sigma}_{\varepsilon,T}(\rho_T))}{\sqrt{1-\hat{\psi}^2_T(\rho_T)}}, \label{modification}
\end{align}
where $\hat{\psi}_T(\rho_T) \coloneqq T^{-1}\sum_{t=1}^{T}z_t(\rho_T)\{z_t^2(\rho_T) - \hat{\sigma}_{\varepsilon,T}^2(\rho_T)\}/(\hat{\sigma}_{\varepsilon,T}(\rho_T)\hat{\sigma}_{\eta,T}(\rho_T))$ is an estimator of $\psi$.\footnote{In fact, \cite{campbell2006Efficienttests} proposed the modification based on the optimality argument.} In Appendix B, we show that $\hat{\psi}_T(\rho_T)$ is consistent. Note that $z_t(\rho_T)$ is used as a proxy for $T^{1/2}\times \sigma_{\varepsilon}dW_\varepsilon$.  With the replacement given in \eqref{modification}, the modified LN test statistic is defined by
\begin{align*}
	\mathrm{LN}_T^*(\rho_T) \coloneqq \frac{\sum_{t=1}^{T}(y_{t-1}^2-T^{-1}\sum_{t=1}^{T}y_{t-1}^2)z_t^{2*}(\rho_T)}{\hat{\sigma}_{\eta,T}(\rho_T)\{\sum_{t=1}^{T}(y_{t-1}^2-T^{-1}\sum_{t=1}^{T}y_{t-1}^2)^2\}^{1/2}}
\end{align*}
The modified augmented t and Wald tests are based on the following regression model:
\begin{align}
	\widetilde{Z_2^*}(\rho_T) = \widetilde{X}\theta_T+\widetilde{\Xi^{**}}, \label{model:linearized_augmented_matrix_modified}
\end{align}
where $\widetilde{Z_2^*}(\rho_T) \coloneqq (\widetilde{z_1^{2*}}(\rho_T),\widetilde{z_2^{2*}}(\rho_T),\ldots,\widetilde{z_T^{2*}}(\rho_T))'$ with $\widetilde{z_t^{2*}}(\rho_T) \coloneqq z_t^{2*}(\rho_T) - T^{-1}\sum_{t=1}^{T}z_t^{2*}(\rho_T)$. The modified augmented test statistics are
\begin{align*}
	t_{\hat{\omega}_T^2}^*(\rho_T) \coloneqq \hat{\sigma}_{\widetilde{\xi^{**}}}^{-1}(\rho_T)(\widetilde{X_2}'M_1\widetilde{X_2})^{-1/2}\widetilde{X_2}'M_1\widetilde{Z_2^*}(\rho_T) \\
	\intertext{and}
	W_T^*(\rho_T) \coloneqq \frac{\hat{\theta}_T^{*'}(\rho_T)(\widetilde{X}'\widetilde{X})\hat{\theta}_T^*(\rho_T)}{\hat{\sigma}_{\widetilde{\xi^{**}}}^2(\rho_T)},
\end{align*}
where $\hat{\sigma}_{\widetilde{\xi^{**}}}^2(\rho_T)$ and $\hat{\theta}_T^{*}(\rho_T)$ are the OLS variance and coefficient estimators of \eqref{model:linearized_augmented_matrix_modified}, respectively.

\begin{thm}
	Consider model \eqref{model:near_unity_rca}. Under Assumptions \ref{asm:local} and \ref{asm:localized_corr}, we have
	\begin{align}
		\mathrm{LN}_T^*(\rho_T) &\Rightarrow \frac{\int_{0}^{1}\widetilde{J_{a,2}}(r)dW_1(r)}{\bigl\{\int_{0}^{1}\bigl(\widetilde{J_{a,2}}\bigr)^2(r)dr\bigr\}^{1/2}} + \frac{\sigma_\varepsilon^2}{\sigma_\eta\sqrt{1-\psi^2}}  \Biggl[\frac{c^2\int_{0}^{1}\bigl(\widetilde{J_{a,2}}\bigr)^2(r)dr + 2cq\int_{0}^{1}\widetilde{J_{a,2}}(r)\widetilde{J_{a,1}}(r)dr}{\bigl\{\int_{0}^{1}\bigl(\widetilde{J_{a,2}}\bigr)^2(r)dr\bigr\}^{1/2}}\Biggr], \label{dist:mod_ln} \\
		t_{\hat{\omega}_T^2}^*(\rho_T) &\Rightarrow \frac{\int_{0}^{1}Q_a(r)dW_1(r)}{\bigl[\int_{0}^{1}Q_a^2(r)dr\bigr]^{1/2}} + \frac{c^2\sigma_\varepsilon^2}{\sigma_\eta\sqrt{1-\psi^2}}\Biggl[\int_{0}^{1}Q_a^2(r)dr\Biggr]^{1/2}, \label{dist:mod_t}
	\end{align}
	and
	\begin{align}
			W_T^*(\rho_T) \Rightarrow& \Bigg\{\begin{pmatrix}
			\int_{0}^{1}\widetilde{J_{a,1}}(r)dW_1(r) \\ \int_{0}^{1}\widetilde{J_{a,2}}(r)dW_1(r) \end{pmatrix} + \frac{\sigma_\varepsilon^2}{\sigma_\eta\sqrt{1-\psi^2}}		\begin{pmatrix} c^2\int_{0}^{1}\widetilde{J_{a,1}}(r)\widetilde{J_{a,2}}(r)dr+2cq\int_{0}^{1}\bigl(\widetilde{J_{a,1}}\bigr)^2(r)dr \\  c^2\int_{0}^{1}\bigl(\widetilde{J_{a,2}}\bigr)^2(r)dr+2cq\int_{0}^{1}\widetilde{J_{a,1}}(r)\widetilde{J_{a,2}}(r)dr
		\end{pmatrix}\Biggr\}' \notag \\ &\times \begin{pmatrix}
			\int_{0}^{1}\bigl(\widetilde{J_{a,1}}\bigr)^2(r)dr & \int_{0}^{1}\widetilde{J_{a,1}}(r)\widetilde{J_{a,2}}(r)dr \\ \int_{0}^{1}\widetilde{J_{a,2}}(r)\widetilde{J_{a,1}}(r)dr & \int_{0}^{1}\bigl(\widetilde{J_{a,2}}\bigr)^2(r)dr
		\end{pmatrix}^{-1} \notag \\ &\times \Bigg\{\begin{pmatrix}
			\int_{0}^{1}\widetilde{J_{a,1}}(r)dW_1(r) \\ \int_{0}^{1}\widetilde{J_{a,2}}(r)dW_1(r) \end{pmatrix} + \frac{\sigma_\varepsilon^2}{\sigma_\eta\sqrt{1-\psi^2}} \begin{pmatrix} c^2\int_{0}^{1}\widetilde{J_{a,1}}(r)\widetilde{J_{a,2}}(r)dr+2cq\int_{0}^{1}\bigl(\widetilde{J_{a,1}}\bigr)^2(r)dr \\  c^2\int_{0}^{1}\bigl(\widetilde{J_{a,2}}\bigr)^2(r)dr+2cq\int_{0}^{1}\widetilde{J_{a,1}}(r)\widetilde{J_{a,2}}(r)dr
		\end{pmatrix}\Biggr\} \label{dist:mod_wald},
	\end{align}
	where $W_1$ is a standard Brownian motion independent of $W_\varepsilon$. \label{thm:mod_ln_t_Wald_dist}
\end{thm}
It should be noticed that the modified test statistics have pivotal asymptotic null distributions (standard normal and chi square with 2 degrees of freedom), thanks to the independence between $J_a$ and $W_1$. Also note that the limiting distributions are unaffected by our modification when $\psi=0$ (cf. Theorems \ref{thm:ln_dist} and \ref{thm:t_wald_dist}).

Figure \ref{fig:power_ln_t_wald} compares the asymptotic power functions of the modified LN, augmented t and augmented Wald tests under $q=0,1,2,3$ and $a=0$ with $\varepsilon_t \sim \mathrm{i.i.d} \ N(0,1)$. When $q=0,1$, the LN test performs best, and the Wald test's power function is slightly below the LN test's. As for the comparison between the augmented t and Wald tests, the Wald test has better power properties than the t test. When $q=2,3$, the Wald test outperforms the LN test (and the augmented t test) with the greater dominance by the Wald test for larger $q$. To translate these results into finite-sample ones, consider for example the case of $T=200$, in which case $-3.76<q<3.76$. Based on the results from the local asymptotic case, it is expected that the augmented Wald test performs more poorly than the LN test if $|q|\leq 1$, or $|\mathrm{Corr}(\varepsilon_t,v_t)|\leq0.266$ in this case, and outperforms the LN test otherwise. To see whether this reasoning gives a good approximation of finite-sample results, we present the size-adjusted power functions for $T=200$ in Figure \ref{fig:power_ln_t_wald_200}. The calculation of these power functions are based on 20,000 replications where $\varepsilon_t \sim \mathrm{i.i.d} \ N(0,1)$, $v_t \sim \mathrm{i.i.d} \ N(0,1)$, $\mathrm{Corr}(\varepsilon_t,v_t)\in\{0,0.25,0.5,0.75\}$ and $\omega_T^2\leq0.0177$ so that $c^2\leq50$ under $T=200$. Figure \ref{fig:power_ln_t_wald_200} shows the local asymptotic analysis can well predict the finite-sample results: the LN test performs slightly better than the augmented Wald test when $|\mathrm{Corr}(\varepsilon_t,v_t)|\leq0.25$, but the latter test outperforms the former otherwise.

We also compute the asymptotic power functions under $a\in\{-5,-10\}$, to investigate the effect of the $a$ value on the power properties. The computed power functions are displayed in Figures \ref{fig:power_ln_t_wald_general_5} and \ref{fig:power_ln_t_wald_general_10}. When $a<0$, the general pattern of the power properties is the same as when $a=0$: the LN test performs best for small $q$, and the Wald test performs best for moderate to large $q$. However, the powers of all the tests we consider get lower as $a$ deviates from 0 (as long as $a<0$).\footnote{According to the results of our simulation studies given later, deviations from 0 by positive $a$ seem to lead to higher powers.} A similar tendency of the power properties of tests for $H_0: \omega_T^2=0$ has been observed through simulation by earlier work such as \citet{nagakura2009Testingcoefficienta} and \citet{horvath2019Testingrandomnessa}. Nonetheless, even when $a<0$, the Wald test's power function is increasing in $q$ while the LN test's is decreasing in $q$. This renders the Wald test preferable in empirical applications, where in general the degree of the correlation between the random coefficient and disturbance is unknown to practitioners. 

Although all the above results are based on the assumption that the true $\rho_T$ is known so that the tests discussed so far are infeasible, they suggest the potential of the augmentation to enhance the ability to detect the nonzero variance of the autoregressive root.

\section{The Case of Unknown $\rho_T$}
In this section, we consider model \eqref{model:near_unity_rca} with the true $\rho_T$ unknown. To deal with the uncertainty about $\rho_T$, we use the OLS estimator $\hat{\rho}_T$ of $\rho_T$, which is defined by $\hat{\rho}_T\coloneqq \sum_{t=1}^{T}y_{t-1}y_t/\sum_{t=1}^{T}y_{t-1}^2$. In Appendix C, we show that $\hat{\rho}_T$ is $T-$consistent, and that other estimators defined in the preceding section such as $\hat{\sigma}_{\varepsilon,T}^2(\rho_T)$ remain consistent if they are computed using $\hat{\rho}_T$ instead of the true $\rho_T$. Given these consistency results, one may expect that the asymptotic behaviors of $\mathrm{LN}_T^*(\hat{\rho}_T)$, $t_{\hat{\omega}_T^2}^*(\hat{\rho}_T)$ and $W_T^*(\hat{\rho}_T)$ are the same as those of $\mathrm{LN}_T^*(\rho_T)$, $t_{\hat{\omega}_T^2}^*(\rho_T)$ and $W_T^*(\rho_T)$. Unfortunately, however, this is not the case. Indeed, it can be shown that the limiting null distributions of the test statistics, if calculated using $\hat{\rho}_T$, are no longer normal or chi square. This is essentially because $a$ is not consistently estimable, which has been well known in the literature \citep{campbell2006Efficienttests}. 

To deal with this problem, following \citet{cavanagh1995InferenceModels} and \citet{campbell2006Efficienttests}, we base our tests on the Bonferroni approach by using a confidence interval for $\rho_T$ instead of its point estimate $\hat{\rho}_T$. The Bonferroni-based test consists of two steps: first, construct a confidence interval for $\rho_T$, and then repeat either of the tests considered above using all the hypothetical $\rho_T$ values belonging to the confidence interval. Specifically, letting $S_T(\rho_T)$ denote either the modified LN, augmented t or augmented Wald test statistic (calculated using $\rho_T$), the testing procedure based on the Bonferroni approach is described as follows:

\begin{itemize} 
	
	\item[Step 1.] Given the data $\{y_t\}_{t=0}^T$, calculate for each hypothetical $\bar{\rho}_T = 1+\bar{a}/T$ (on some grid) the t statistic $t_{\bar{\rho_T}} \coloneqq (\hat{\rho}_T-\bar{\rho}_T)(\hat{\sigma}^{-2}_T(\hat{\rho}_T)\sum_{t=1}^{T}y_{t-1}^2)^{1/2}$, to construct the (equal-tailed) $100(1-\alpha_1)\%$ confidence interval for $\rho_T$, denoted by $\mathrm{CI}(\alpha_1)$. That is, $\mathrm{CI}(\alpha_1)$ is the collection of all the $\bar{\rho}_T = 1+\bar{a}/T$ that satisfies $cv(\bar{a})_{\alpha_1/2} \leq t_{\bar{\rho_T}} \leq cv(\bar{a})_{1-\alpha_1/2}$ with $cv(\bar{a})_{\alpha_1/2}$ and $cv(\bar{a})_{1-\alpha_1/2}$ denoting the lower and upper $\alpha_1/2$ quantiles of the asymptotic distribution of $t_{\bar{\rho_T}}$ derived under the assumption that $\rho_T=\bar{\rho}_T$.
		
	\item[Step 2.] For each $\bar{\rho}_T \in \mathrm{CI}(\alpha_1)$, calculate $S_T(\bar{\rho}_T)$ and compare its value with the critical value $cv_{\alpha_2}$ with significance level $\alpha_2$ (based on the standard normal or chi square distributions). Reject the null only if all the calculated $S_T(\bar{\rho}_T)$ ($\bar{\rho}_T \in \mathrm{CI}(\alpha_1)$) exceed the critical value.
\end{itemize}
We shall call tests following these two steps Bonferroni tests.

\begin{rem} 
	\begin{itemize}
		\item[ ]
		
		\item[1.] The asymptotic distribution of $t_{\bar{\rho_T}}$ under $\rho_T=\bar{\rho}_T$ is $\int_{0}^{1}J_{\bar{a}}(r)dW_\varepsilon(r)/\{\int_{0}^{1}J_{\bar{a}}^2(r)dr\}^{1/2}$. The quantiles $cv(\bar{a})_{\alpha_1/2}$ and $cv(\bar{a})_{1-\alpha_1/2}$ can be found through simulation.
		
		\item[2.] By the Bonferroni inequality, the resulted test is (asymptotically) a test with significance level $\alpha=\alpha_1+\alpha_2$: under $H_0$,
		\begin{align*}
			P_{a,\psi}\Bigl(\bigcap_{\bar{\rho}_T\in \mathrm{CI}(\alpha_1)}\Bigl\{S_T(\bar{\rho}_T)>cv_{\alpha_2}\Bigr\}\Bigr) &= P_{a,\psi}\Bigl(\bigcap_{\bar{\rho}_T\in \mathrm{CI}(\alpha_1)}\Bigl\{S_T(\bar{\rho}_T)>cv_{\alpha_2}\Bigr\}|a\in \mathrm{CI}(\alpha_1)\Bigr)P_{a,\psi}\Bigl(a\in \mathrm{CI}(\alpha_1)\Bigr) \\
			&+ P_{a,\psi}\Bigl(\bigcap_{\bar{\rho}_T\in \mathrm{CI}(\alpha_1)}\Bigl\{S_T(\bar{\rho}_T)>cv_{\alpha_2}\Bigr\}|a\notin \mathrm{CI}(\alpha_1)\Bigr)P_{a,\psi}\Bigl(a\notin \mathrm{CI}(\alpha_1)\Bigr) \\
			&\leq \alpha_2(1-\alpha_1) + \alpha_1 \leq \alpha,
		\end{align*}
		
		\item[3.] \citet{campbell2006Efficienttests} constructed the confidence interval for $\rho_T$ by inverting the Dickey-Fuller type t statistic, which is calculated by centering $\hat{\rho}_T$ by unity rather than $\bar{\rho}_T$. However, it has been known in the literature on predictive regressions that the use of the Dickey-Fuller type t ratio leads to severe size distortion \citep{phillips2014ConfidenceIntervals}. We found in our unreported simulations that our Bonferroni tests also suffer from the same problem when it is based on the Dickey-Fuller t ratio. Following the theoretical analysis by \citet{phillips2014ConfidenceIntervals}, we here construct the confidence interval for $\rho_T$ using the t statistic centered by hypothetical $\bar{\rho}_T$'s.
	\end{itemize}
\end{rem}

Although the Bonferroni test defined above is a valid test with significance level $\alpha=\alpha_1+\alpha_2$, the test's type 1 errors (dependent on $(a,\psi)$) can be quite smaller than $\alpha$, as pointed out by \citet{cavanagh1995InferenceModels} and \citet{campbell2006Efficienttests}. To make the type 1 errors close to the desired nominal level, $\tilde{\alpha}$ (say), they proposed using a pair $(\alpha_1,\alpha_2)$ with which the Bonferroni test's type 1 errors are close to the given $\tilde{\alpha}$. Theoretically, we would find such numerous pairs $(\alpha_1,\alpha_2)$ by changing both the values of $\alpha_1$ and $\alpha_2$. To simplify the searching process, we fix $\alpha_2$ and set  $\alpha_2=\tilde{\alpha}$, following \citet{campbell2006Efficienttests}. In this study, we consider the case $\tilde{\alpha}=\alpha_2=0.05$, giving the Bonferroni test with significance level 0.05. Then, for each $\psi \in (-1,1)$, we numerically find the value of $\alpha_1(\psi)$ such that under the null $P_{a,\psi}\Bigl(\bigcap_{\bar{\rho}_T\in \mathrm{CI}(\alpha_1(\psi))}\Bigl\{S_T(\bar{\rho}_T)>cv_{\alpha_2}\Bigr\}\Bigr)\leq\tilde{\alpha}=0.05$, and this probability is as close to 0.05 as possible, for all $a$ on some grid. The simulation process to find the $\alpha_1(\psi)$ values is described in Appendix A.

Table \ref{tab:significance_ci} displays the significance level $\alpha_1(\psi)$ of the confidence interval for $\rho_T$ along with corresponding intervals of $|\psi|$ values.\footnote{We assign $\alpha_1$ values to each interval of $\psi$ instead of each single $\psi$ (on some grid) for computational ease of the Bonferroni-test.} Given the good asymptotic performance by the modified Wald test, we report in Table \ref{tab:significance_ci} $\alpha_1$ values for the case of $S_T(\rho_T)$ being the modified augmented Wald test statistic. When performing the Bonferroni-Wald test, first estimate $\psi$ by its consistent estimator $\hat{\psi}(\hat{\rho}_T)$, and then select the value of $\alpha_1$ based on Table \ref{tab:significance_ci}. For example, if $|\hat{\psi}(\hat{\rho}_T)|=0.17$, $\alpha_1=0.31$ is selected. With the selected $\alpha_1$ value, one can perform the Bonferroni-Wald test following the two-step testing procedure outlined before. The finite-sample performances of the Bonferroni-Wald test are investigated through simulation in the next section.

\section{Finite-Sample Performance}
As described in Appendix A, the simulation exercise to determine the values of $\alpha_1$ is based on the asymptotic procedure. Thus, we need to verify whether the Bonferroni-Wald test we have proposed performs well in finite samples.

\subsection{Empirical size}
Following \citet{nagakura2009Testingcoefficienta}, we employ three data generating mechanism to evaluate empirical sizes: (i) $\varepsilon_t\sim \mathrm{i.i.d} \ N(0,1)$, so that $\psi=0$, (ii) $\varepsilon_t \sim \mathrm{i.i.d} \ (\chi^2(10)-10)/\sqrt{20}$, so that $\psi=0.5$, and (iii) $\varepsilon_t \sim \mathrm{i.i.d} \ (\chi^2(1)-1)/\sqrt{2}$, so that $\psi=0.756$. For each case, we set $y_0=0$. The sample size we use is $T\in\{200,500,1000\}$, and the number of replications is 5,000. The $\rho_T=\rho$ values are fixed across $T$, and we consider $\rho\in [0.7,1.01]$. The simulation results are collected in Figure \ref{fig:sizes_fs}. The general pattern is that the empirical rejection rates under $H_0$ are relatively small when $\rho<1$ and tend to be greater than the nominal level 0.05 when $\rho$ is near unity. As for the normal case (Figure \ref{fig:sizes_fs}(a)), rejection rates are stable around 0.05 over $\rho\in[0.7,1.01]$, with them approaching 0.05 as $T$ increases. As for the chi square cases (Figure \ref{fig:sizes_fs}(b) and (c)), rejection rates stay around the nominal level, but they get farther away from 0.05 when $|\psi|$ is larger, $T$ is smaller and $\rho$ is near 1. In particular, when $\varepsilon_t \sim \mathrm{i.i.d} \ (\chi^2(1)-1)/\sqrt{2}$ and $T=200$, the rejection rates can be as large as 0.09 (around $\rho=1$), although they approach 0.05 as $T$ increases. A similar tendency can be observed in the finite-sample behavior of modified LN tests proposed by \citet{nagakura2009Testingcoefficienta} according to his simulation results. One possible cause of this phenomenon would be the finite-sample bias involved in the estimation of $\psi$ when the $\psi$ value is large. One may be advised to use a more conservative Bonferroni-Wald test (with smaller $\alpha_1$ values) when the $|\psi|$ estimate is large and the sample size is not so large.

\subsection{Finite-sample power comparison}
Next, we investigate tests' ability to detect the nonzero variance in the autoregressive root. The simulation design is as follows: $T=200$, $\varepsilon_t\sim \mathrm{i.i.d} \ N(0,1)$ and $v_t \sim \mathrm{i.i.d} \ N(0,1)$ with $\mathrm{Corr}(\varepsilon_t,v_t)\in \{0,0.25,0.5,0.75\}$, $\rho\in\{1.01,1,0.98,0.95\}$, and $\omega^2 \in (0,0.01]$ (corresponding to $c^2 \in (0,28.28]$). The tests we consider here are the Bonferroni-Wald test, the infeasible modified Wald test (calculated using the true $\rho$), and one of the modified LN tests proposed by \citet{nagakura2009Testingcoefficienta}, which is denoted by $\widetilde{G}_{T,1}$ in his notation. The infeasible Wald test is taken as a benchmark, and thereby we can evaluate the power loss originating from using the confidence interval for $\rho$ to perform the Bonferroni-Wald test. \citet{nagakura2009Testingcoefficienta}'s modified LN test statistic is designed to converge in distribution under $H_0$ to the standard normal irrespective of the $\psi$ value, under the data generating mechanism with $\rho \in (-1,1]$ fixed (i.e., independent of $T$). Because \citet{nagakura2009Testingcoefficienta} did not show its asymptotic null distribution remains standard normal when $\rho>1$, we do not perform it for the case $\rho=1.01$. We also consider the test proposed recently by \citet{horvath2019Testingrandomnessa} (hereafter HT). The HT test is a randomized one, and their test statistic $\Theta_{T,R}$ (in their notation) converges in distribution to the chi square with one degree of freedom, for almost all realizations.\footnote{The HT test needs a tuning parameter, $x \in(0,0.5)$, to be performed, and they stated their test's performance is insensitive to the choice of the $x$ value and set $x=0.1$. However, in our simulation, their test's performance is somehow affected by the choice of $x$. Thus, we set $x=0.38$, to obtain simulation results similar to those of HT.} The HT test is not originally proposed under the local-to-unity specification, but it will be informative to practitioners to reveal the performance of the test under this setting.

The simulation results are shown in Figures \ref{fig:powers_fs_101} through \ref{fig:powers_fs_095}. For the case $\rho=1.01$, where all but the LN test are performed, the infeasible and Bonferroni-Wald tests have good power, and the discrepancy between their power functions is small. The latter result is due to the refinement on the construction of the confidence interval for $\rho$. In contrast, the power function of the HT test stays around the nominal level 0.05 over $\omega^2 \in (0,0.01]$, which implies it has almost no distinguishing power for these alternatives. Indeed, HT conjectured (based on some theoretical analysis) that their test would have no power when $1-\rho_T=O(T^{-1})$ and $\omega_T^2 = O(T^{-1})$, which is of larger magnitude than our local-to-zero variance $\omega_T^2=O(T^{-3/2})$.\footnote{\citet{nishi2022StochasticLocal} showed that under the STUR specification, the LN and some other tests are consistent when $\omega_T^2=O(T^{-1})$, from which they concluded that this case should be regarded as capturing stochastic ``moderate" departures from a unit root rather than local departures.} Our simulation results corroborate their statement.

Turning to the cases $\rho \leq 1$, it is noticeable that for each $\rho$, the LN test performs well for small values of $\mathrm{Corr}(\varepsilon_t,v_t)$, but Wald type tests outperform the LN test otherwise. In particular, the greater value of $\mathrm{Corr}(\varepsilon_t,v_t)$ leads to the greater dominance by the Wald type tests. The power function of HT test, again, stays around the nominal level for these cases. Overall, the Bonferroni-Wald test performs better than the LN and HT tests for moderate to large values of $\mathrm{Corr}(\varepsilon_t,v_t)$ (irrespective of the $\rho$ value) and performs almost as efficiently as its ideal, infeasible counterpart.
 
 \section{Empirical Application}
In this section, we apply the Bonferroni-Wald test along with the LN and HT tests\footnote{Based on our simulation results, we set the tuning parameter $x=0.38$ for the HT test.} to several U.S. macroeconomic and financial time series, following \citet{hill2014UnifiedIntervalb} and HT. The dataset includes CPI, real GDP, industrial production, M2, S\&P500, the 3 month Treasury bill rate and the unemployment rate. We take the logarithm of the first 5 series before detrending all the series. All data have been extracted from Federal Reserve Economic Data. The data description and testing results are displayed in Tables \ref{tab:empirical_description} and \ref{tab:empirical_results}, respectively.

Before discussing main results, it should be recalled from our simulation results that the Bonferroni-Wald test can be oversized when $T$ is not large, $\psi$ is large and $\rho$ is near unity (see Section 4). Thus, we need to check whether all the three conditions hold for our series or not. According to Table \ref{tab:empirical_results}, all the $\rho$ estimates are near unity, as expected. We also have obtained moderate estimates of $\psi$ for the CPI and S\&P500 series, but the sample sizes for them are large enough that it is unlikely the Bonferroni-Wald test is oversized for these series. As for GDP, the sample size is $T=292$ but the $\psi$ estimate is near zero, and hence we expect the Bonferroni-Wald test is not severely oversized for this series. Overall, all the three conditions for the potential oversize problem do not jointly hold. 

For GDP and T-bill rate, we have obtained consistent results from all the three tests: the null of $H_0:\omega^2=0$ is not rejected for these series. This coincidence could be viewed as evidence for nonrandomness of the autoregressive root for these series. As for CPI, industrial production and S\&P500, it is observed that the Bonferroni-Wald and LN tests reject the null while the HT test does not. This result will be attributed to the fact that the former two tests are much more powerful than the latter one. Finally, as for M2 and unemployment rate, only the Bonferroni-Wald test rejects the null. This outcome will be due to the fact that the Bonferroni-Wald test tends to be the most powerful among the three tests. 

In Table \ref{tab:empirical_results}, we also present the $\omega^2$ estimates proposed by \citet{horvath2019Testingrandomnessa} (denoted by $\hat{\omega}^2_\mathrm{HT}$). They showed this estimator is consistent under the non-local RCA models (with $\rho$ and $\omega^2$ fixed). The values of $\hat{\omega}^2_\mathrm{HT}$ seem to support our testing results. For instance, we have a negative $\hat{\omega}^2_\mathrm{HT}$ for T-bill rate, for which none of the tests reject the null. Moreover, for the other series, $\hat{\omega}^2_\mathrm{HT}$ takes values around $10^{-4}$ to $10^{-3}$, magnitudes of the coefficient randomness with which the HT test tends to be powerless under the sample size given in Table \ref{tab:empirical_results}. For example, for M2, the sample size is $T=2044$, and $\hat{\omega}^2_\mathrm{HT}$ is $3.12\times10^{-4}$, translating into $\hat{c}^2=\hat{\omega}^2_\mathrm{HT}\times T^{3/2}=28.83$ estimate of the localizing coefficient $c^2$. According to our simulation results, the HT test has almost no power against the alternative of this magnitude, hence the testing result given in Table \ref{tab:empirical_results}.

Given these findings, our empirical application illustrates the merit in choosing our Bonferroni-Wald test over existing ones.

\section{Conclusion and Discussion}
Given the results of empirical analyses conducted by earlier studies, the local-to-unity RCA models, which extend the STUR modelling, are empirically relevant. Under this setting, we can analyze the effect of the correlation between the random coefficient and disturbance on the power properties of
tests for coefficient randomness. Theoretical and simulation analyses reveal that tests proposed by earlier studies can perform poorly when the degree of the correlation is moderate to large and the coefficient randomness is local to zero, while the augmented-Wald test we have proposed performs well even in such cases. Our test is also independent of the nuisance parameter $\psi$, the correlation between the disturbance and its square, so that it is implementable without the knowledge about the value of $\psi$. To deal with the uncertainty about the mean $\rho$ of the autoregressive root, we have proposed using a confidence interval for $\rho$, leading to the Bonferroni-Wald test, where the significance level for the confidence interval is selected according to the value of the $\psi$ estimate. Embedding this selection process into the Bonferroni-Wald test helps stabilize the test's size and improve the test's power.

Several directions for future research are possible. First, from the similarity in the construction of test statistics between our model and predictive regressions, it is expected that the theory developed by numerous studies on predictive regressions can be applied to testing for coefficient randomness in local-to-unity autoregressions. For example \citet{phillips2013Predictiveregression,phillips2016Robusteconometric} proposed the use of the so-called IVX procedure for predictability testing \citep[see also][]{kostakis2015RobustEconometric}. This procedure leads to good size and power properties and also requires less computational burden for implementation than the Bonferroni approach employed by \citet{campbell2006Efficienttests}. The use of the IVX approach may facilitate testing for coefficient randomness in local-to-unity autoregressions. The analysis of tests' performance when $\rho$ is distant from unity will also be of interest. In such a case, more preferable tests might be available than the Bonferroni-Wald test, which is based on the local-to-unity asymptotics.

\bibliographystyle{standard}
\bibliography{stur_end}

@article{aue2011QuasiLikelihoodEstimation,
  title = {Quasi-{{Likelihood Estimation}} in {{Stationary}} and {{Nonstationary Autoregressive Models}} with {{Random Coefficients}}},
  author = {Aue, Alexander and Horv{\'a}th, Lajos},
  year = {2011},
  journal = {Statistica Sinica},
  volume = {21},
  number = {3},
  pages = {973--999},
  publisher = {{Institute of Statistical Science, Academia Sinica}},
  issn = {1017-0405}
}

@article{campbell2006Efficienttests,
  title = {Efficient Tests of Stock Return Predictability},
  author = {Campbell, John Y. and Yogo, Motohiro},
  year = {2006},
  month = jul,
  journal = {Journal of Financial Economics},
  volume = {81},
  number = {1},
  pages = {27--60},
  issn = {0304-405X},
  langid = {english}
}

@article{cavanagh1995InferenceModels,
  title = {Inference in {{Models}} with {{Nearly Integrated Regressors}}},
  author = {Cavanagh, Christopher L. and Elliott, Graham and Stock, James H.},
  year = {1995},
  journal = {Econometric Theory},
  volume = {11},
  number = {5},
  pages = {1131--1147},
  publisher = {{Cambridge University Press}},
  issn = {0266-4666}
}

@article{distaso2008Testingunit,
  title = {Testing for Unit Root Processes in Random Coefficient Autoregressive Models},
  author = {Distaso, Walter},
  year = {2008},
  month = jan,
  journal = {Journal of Econometrics},
  volume = {142},
  number = {1},
  pages = {581--609},
  issn = {03044076},
  langid = {english}
}

@article{elliott1996Efficienttestsa,
  title = {Efficient Tests for an Autoregressive Unit Root},
  author = {Elliott, Graham and Rothenberg, Thomas J. and Stock, James H.},
  year = {1996},
  journal = {Econometrica},
  volume = {64},
  number = {4},
  pages = {813--836},
  publisher = {{[Wiley, Econometric Society]}}
}

@article{hansen1992Convergencestochastica,
  title = {Convergence to Stochastic Integrals for Dependent Heterogeneous Processes},
  author = {Hansen, Bruce E.},
  year = {1992},
  journal = {Econometric Theory},
  volume = {8},
  number = {4},
  pages = {489--500},
  publisher = {{Cambridge University Press}},
  langid = {english}
}

@article{hill2014UnifiedIntervalb,
  title = {Unified {{Interval Estimation}} for {{Random Coefficient Autoregressive Models}}},
  author = {Hill, Jonathan and Peng, Liang},
  year = {2014},
  journal = {Journal of Time Series Analysis},
  volume = {35},
  number = {3},
  pages = {282--297},
  issn = {1467-9892},
  langid = {english}
}

@article{horvath2019Testingrandomnessa,
  title = {Testing for Randomness in a Random Coefficient Autoregression Model},
  author = {Horv{\'a}th, Lajos and Trapani, Lorenzo},
  year = {2019},
  journal = {Journal of Econometrics},
  volume = {209},
  number = {2},
  pages = {338--352},
  langid = {english}
}

@article{horvath2022ChangepointDetection,
  title = {Changepoint {{Detection}} in {{Heteroscedastic Random Coefficient Autoregressive Models}}},
  author = {Horv{\'a}th, Lajos and Trapani, Lorenzo},
  year = {2022},
  month = sep,
  journal = {Journal of Business \& Economic Statistics (forthcoming)},
  issn = {0735-0015}
}

@article{hwang2005ExplosiveRandomCoefficientb,
  title = {Explosive {{Random-Coefficient AR}}(1) {{Processes}} and {{Related Asymptotics}} for {{Least-Squares Estimation}}},
  author = {Hwang, S. Y. and Basawa, I. V.},
  year = {2005},
  journal = {Journal of Time Series Analysis},
  volume = {26},
  number = {6},
  pages = {807--824},
  issn = {1467-9892},
  langid = {english}
}

@article{kostakis2015RobustEconometric,
  title = {Robust {{Econometric Inference}} for {{Stock Return Predictability}}},
  author = {Kostakis, Alexandros and Magdalinos, Tassos and Stamatogiannis, Michalis P.},
  year = {2015},
  journal = {The Review of Financial Studies},
  volume = {28},
  number = {5},
  pages = {1506--1553},
  publisher = {{[Oxford University Press, The Society for Financial Studies]}},
  issn = {0893-9454}
}

@article{lee1998Coefficientconstancya,
  title = {Coefficient Constancy Test in a Random Coefficient Autoregressive Model},
  author = {Lee, Sangyeol},
  year = {1998},
  journal = {Journal of Statistical Planning and Inference},
  volume = {74},
  number = {1},
  pages = {93--101},
  publisher = {{Elsevier Science}}
}

@article{leybourne1996Caneconomica,
  title = {Can Economic Time Series Be Differenced to Stationarity?},
  author = {Leybourne, S. J. and McCabe, B. P. M. and Tremayne, A. R.},
  year = {1996},
  journal = {Journal of business \& economic statistics},
  volume = {14},
  number = {4},
  pages = {435}
}

@article{mccabe1995Testingtimea,
  title = {Testing a Time Series for Difference Stationarity},
  author = {McCabe, B. P. M. and Tremayne, A. R.},
  year = {1995},
  journal = {The Annals of Statistics},
  volume = {23},
  number = {3},
  pages = {1015--1028},
  publisher = {{Institute of Mathematical Statistics}}
}

@article{mccabe1998Powertestsa,
  title = {The Power of Some Tests for Difference Stationarity under Local Heteroscedastic Integration},
  author = {McCabe, B. P. M. and Smith, R. J.},
  year = {1998},
  journal = {Journal of the American Statistical Association},
  volume = {93},
  number = {442},
  pages = {751--761},
  publisher = {{Amer Statistical Assoc}},
  address = {{Alexandria}},
  langid = {english}
}

@article{nagakura2009Asymptotictheory,
  title = {Asymptotic Theory for Explosive Random Coefficient Autoregressive Models and Inconsistency of a Unit Root Test against a Stochastic Unit Root Process},
  author = {Nagakura, Daisuke},
  year = {2009},
  month = dec,
  journal = {Statistics \& Probability Letters},
  volume = {79},
  number = {24},
  pages = {2476--2483},
  publisher = {{Elsevier}},
  address = {{Amsterdam}},
  issn = {0167-7152},
  langid = {english}
}

@article{nagakura2009Testingcoefficienta,
  title = {Testing for Coefficient Stability of {{AR}}(1) Model When the Null Is an Integrated or a Stationary Process},
  author = {Nagakura, Daisuke},
  year = {2009},
  journal = {Journal of Statistical Planning and Inference},
  volume = {139},
  number = {8},
  pages = {2731--2745},
  langid = {english}
}

@book{nicholls1982RandomCoefficient,
  title = {Random {{Coefficient Autoregressive Models}}: {{An Introduction}}},
  shorttitle = {Random {{Coefficient Autoregressive Models}}},
  author = {Nicholls, Des F. and Quinn, Barry G.},
  editor = {Brillinger, D. and Fienberg, S. and Gani, J. and Hartigan, J. and Krickeberg, K.},
  year = {1982},
  series = {Lecture {{Notes}} in {{Statistics}}},
  volume = {11},
  publisher = {{Springer US}},
  address = {{New York, NY}}
}

@techreport{nishi2022StochasticLocal,
  title = {Stochastic {{Local}} and {{Moderate Departures}} from a {{Unit Root}} and {{Its Application}} to {{Unit Root Testing}}},
  author = {Nishi, Mikihito and Kurozumi, Eiji},
  year = {2022},
  month = aug,
  journal = {Discussion Papers},
  number = {2022-02},
  institution = {{Graduate School of Economics, Hitotsubashi University}},
  langid = {english}
}

@article{phillips1987Unifiedasymptotica,
  title = {Towards a Unified Asymptotic Theory for Autoregression},
  author = {Phillips, Peter C. B.},
  year = {1987},
  journal = {Biometrika},
  volume = {74},
  number = {3},
  pages = {535--547},
  publisher = {{[Oxford University Press, Biometrika Trust]}}
}

@article{phillips1990AsymptoticProperties,
  title = {Asymptotic {{Properties}} of {{Residual Based Tests}} for {{Cointegration}}},
  author = {Phillips, Peter C. B. and Ouliaris, S.},
  year = {1990},
  journal = {Econometrica},
  volume = {58},
  number = {1},
  pages = {165--193},
  publisher = {{[Wiley, Econometric Society]}},
  issn = {0012-9682}
}

@article{phillips2013Predictiveregression,
  title = {Predictive Regression under Various Degrees of Persistence and Robust Long-Horizon Regression},
  author = {Phillips, Peter C. B. and Lee, Ji Hyung},
  year = {2013},
  month = dec,
  journal = {Journal of Econometrics},
  series = {Dynamic {{Econometric Modeling}} and {{Forecasting}}},
  volume = {177},
  number = {2},
  pages = {250--264},
  issn = {0304-4076},
  langid = {english}
}

@article{phillips2014ConfidenceIntervals,
  title = {On {{Confidence Intervals}} for {{Autoregressive Roots}} and {{Predictive Regression}}},
  author = {Phillips, Peter C. B.},
  year = {2014},
  journal = {Econometrica},
  volume = {82},
  number = {3},
  pages = {1177--1195},
  publisher = {{The Econometric Society}},
  issn = {0012-9682}
}

@article{phillips2016Robusteconometric,
  title = {Robust Econometric Inference with Mixed Integrated and Mildly Explosive Regressors},
  author = {Phillips, Peter C. B. and Lee, Ji Hyung},
  year = {2016},
  month = jun,
  journal = {Journal of Econometrics},
  series = {Innovations in {{Multiple Time Series Analysis}}},
  volume = {192},
  number = {2},
  pages = {433--450},
  issn = {0304-4076},
  langid = {english}
}

@article{stock1991Confidenceintervals,
  title = {Confidence Intervals for the Largest Autoregressive Root in {{U}}.{{S}}. Macroeconomic Time Series},
  author = {Stock, James H.},
  year = {1991},
  month = dec,
  journal = {Journal of Monetary Economics},
  volume = {28},
  number = {3},
  pages = {435--459},
  issn = {0304-3932},
  langid = {english}
}

@article{su2012Examiningpowera,
  title = {Examining the Power of Stochastic Unit Root Tests without Assuming Independence in the Error Processes of the Underlying Time Series},
  author = {Su, Jen-Je and Roca, Eduardo},
  year = {2012},
  journal = {Applied Economics Letters},
  volume = {19},
  number = {4},
  pages = {373--377},
  langid = {english}
}

\section*{Appendix A: Procedure to Determine $\alpha_1$ Values for the Bonferroni-Wald Test}
\setcounter{equation}{0}
\renewcommand{\theequation}{A.\arabic{equation}}

In this appendix, we describe the procedure to determine the $\alpha_1$ values for the confidence interval for $\rho_T$ explained in Section 3. The procedure is based on simulating asymptotic distributions with 5,000 replications. In each replication, we first generate $\{y_t\}_{t=1}^T$ by the mechanism
\begin{align}
y_t=(1+a/T)y_{t-1}+\varepsilon_t, \quad t=1,2,\ldots,T \label{eqnapp:mechanism}
\end{align}
with $y_0=0$, $T=2000$, $a\in [-300,10]$ and $\varepsilon_t \sim \mathrm{i.i.d} \ N(0,1)$. Then, with given $\alpha_1$, we conduct the two-step procedure for the Bonferroni-Wald test explained in Section 3 and calculate the frequency of $H_0$ being rejected for the $\alpha_1$.

Note that $\varepsilon_t$ used in \eqref{eqnapp:mechanism} satisfies $\psi=\mathrm{Corr}(\varepsilon_t,\varepsilon_t^2-\sigma_{\varepsilon}^2)=0$. To calculate the false rejection frequencies for other $\psi$ values, we artificially produce an environment where the Wald test statistic depends on $\psi\neq0$. Noting that under the null, 
\begin{align*}
	z_t^2(\bar{\rho}_T)=1+(\varepsilon_t^2-1)+(\bar{a}/T-a/T)^2y_{t-1}^2-2(\bar{a}/T-a/T)y_{t-1}\varepsilon_t,
\end{align*}
and $\eta_t=\varepsilon_t^2-1$ (combined with $\varepsilon_t$) determines the value of $\psi$ on which the asymptotic distribution of $W_T^*(\bar{\rho}_T)$ depends, we artificially replace $\eta_t=\varepsilon_t^2-1$ with $\eta_t^{rep} \coloneqq \sqrt{1-\psi^2}\eta_t+\psi\sqrt{2}\varepsilon_t$, obtaining
\begin{align*}
	z_t^2(\bar{\rho}_T)^{rep} &\coloneqq 1+\eta_t^{rep}+(\bar{a}/T-a/T)^2y_{t-1}^2-2(\bar{a}/T-a/T)y_{t-1}\varepsilon_t,
\end{align*}
in view of the distributional equivalence $W_\eta \stackrel{d}{=}\sqrt{1-\psi^2}W_1+\psi W_\varepsilon$ and the fact that $\mathbb{V}[\eta_t]= \sigma_\eta^2=2$ and $\mathrm{Corr}(\eta_t,\varepsilon_t)=0$. In this replacement, the newly crafted variable $\eta_t^{rep}$ takes over the role of $\eta_t=\varepsilon_t^2-1$, satisfying $\mathbb{E}[\eta_t^{rep}]=0$, $\mathbb{V}[\eta_t^{rep}]=2=\mathbb{V}[\eta_t]$ and $\mathrm{Corr}(\eta_t^{rep},\varepsilon_t)=\psi$ by construction. This replacement can be justified by the fact that under the null, the Wald test statistic asymptotically depends only on $\psi$ (and $a$) and is not dependent on any other moment. It follows that the Bonferroni-Wald test statistic using $z_t^2(\bar{\rho}_T)^{rep}$ in place of $z_t^2(\bar{\rho}_T)$ asymptotically depends on $\psi$, so that we can calculate the false rejection frequencies for any value of $\psi$.

For each $\psi$ on some grid, we calculate the false rejection rates of the Bonferroni-Wald test with $a$ moving over a grid on $[-300,10]$ and determine the value of $\alpha_1$ for the given $\psi$ such that the false rejection rate is less than or equal to 0.05 for all $a$.

\section*{Appendix B: Proofs of Results in Section 2}
\setcounter{equation}{0}
\renewcommand{\theequation}{B.\arabic{equation}}
In this appendix, we prove the theorems stated in Section 2.

\begin{lemappB}
	Consider model \eqref{model:near_unity_rca} under Assumptions \ref{asm:local} and \ref{asm:localized_corr}. Define the stochastic process $Y_T$ on $[0,1]$ by $Y_T(r)\coloneqq T^{-1/2}y_{\lfloor Tr \rfloor}, \ 0\leq r\leq 1$. Then, $Y_T\Rightarrow \sigma_{\varepsilon} J_a$ in the Skorokhod space $D[0,1]$, where $J_a$ solves $dJ_a(r) = aJ_a(r)dr + dW_\varepsilon(r)$. \label{lemappB:wc_Y}
\end{lemappB}
\noindent \begin{proof}
	The proof is essentially the same as that of Lemma 1(a) of \citet{nishi2022StochasticLocal} and hence is omitted.
\end{proof}

\begin{lemappB}
	Consider model \eqref{model:near_unity_rca} under Assumptions \ref{asm:local} and \ref{asm:localized_corr}. Then, we have
	\begin{itemize}
		\item[(a)] $\hat{\sigma}_{\varepsilon,T}^2(\rho_T) \stackrel{p}{\to}\sigma_{\varepsilon}^2$,
		\item[(b)] $\hat{\sigma}_{\eta,T}^2(\rho_T) \stackrel{p}{\to} \sigma_\eta^2$,
		\item[(c)] $\hat{\psi}_T(\rho_T)\stackrel{p}{\to}\psi$.
		\end{itemize} \label{lemappB:consistency}
\end{lemappB}
\begin{proof}
	The proofs of parts (a) and (b) are identical to those of Lemma 1(b) and (c) of \citet{nishi2022StochasticLocal} and hence are omitted. To prove part (c), it suffices to show
	\begin{align*}
		T^{-1}\sum_{t=1}^{T}z_t(\rho_T)\Big\{z_t^2(\rho_T)-\hat{\sigma}_{\varepsilon,T}^2(\rho_T)\Bigr\} \stackrel{p}{\to} \mathbb{E}[\varepsilon_t^3].
	\end{align*}
	A simple calculation gives
	\begin{align*}
		&T^{-1}\sum_{t=1}^{T}z_t(\rho_T)\{z_t^2(\rho_T)-\hat{\sigma}_{\varepsilon,T}^2(\rho_T)\} \\ &= T^{-1}\sum_{t=1}^{T}(cT^{-3/4}y_{t-1}v_t + \varepsilon_t)\{c^2T^{-3/2}y_{t-1}^2v_t^2 + 2cT^{-3/4}y_{t-1}\varepsilon_tv_t + \varepsilon_t^2-\hat{\sigma}_{\varepsilon,T}^2(\rho_T)\} \\
		&=T^{-1}\sum_{t=1}^{T}\varepsilon_t^3 + A_T,
		\end{align*}
	where
	\begin{align*}
		A_T \coloneqq c^3T^{-13/4}\sum_{t=1}^{T}y_{t-1}^3v_t^3 &+ 3c^2T^{-5/2}\sum_{t=1}^{T}y_{t-1}^2\varepsilon_tv_t^2 + 3cT^{-7/4}\sum_{t=1}^{T}y_{t-1}\varepsilon_t^2v_t\\ &-\hat{\sigma}_{\varepsilon,T}^2(\rho_T)\times\Bigl\{cT^{-7/4}\sum_{t=1}^{T}y_{t-1}v_t+ T^{-1}\sum_{t=1}^{T}\varepsilon_t\Bigr\}.
	\end{align*}
	The first term of $A_T$ satisfies
	\begin{align*}
		c^3T^{-13/4}\sum_{t=1}^{T}y_{t-1}^3v_t^3 &= c^3\mathbb{E}[v_t^3]T^{-13/4}\sum_{t=1}^{T}y_{t-1}^3 + c^3T^{-13/4}\sum_{t=1}^{T}y_{t-1}^3(v_t^3- \mathbb{E}[v_t^3]) \\
		&=c^3\mathbb{E}[v_t^3]T^{-3/4}\int_{0}^{1}Y_T^3(r)dr + c^3T^{-5/4}\int_{0}^{1}Y_T^3(r)dW_{v^3-\mathbb{E}[v^3],T}(r) = O_p(T^{-3/4}),
	\end{align*}
	say, because $\{v_t^3 - \mathbb{E}[v_t^3]\}$ is i.i.d with zero mean and finite variance. Similarly, we can prove that the other terms of $A_T$ is $O_p(T^{-1/4})$. Thus
	\begin{align*}
		T^{-1}\sum_{t=1}^{T}z_t(\rho_T)\{z_t^2(\rho_T)-\hat{\sigma}_{\varepsilon,T}^2(\rho_T)\} = T^{-1}\sum_{t=1}^{T}\varepsilon_t^3 + o_p(1) \stackrel{p}{\to}\mathbb{E}[\varepsilon_t^3].
	\end{align*}
\end{proof}

For later reference, we give several results on the weak convergence of components of test statistics.

\begin{lemappB}
	Consider model \eqref{model:near_unity_rca} under Assumptions \ref{asm:local} and \ref{asm:localized_corr}. Then, we have
	\begin{itemize}
		\item[(a)] 
		\begin{align*}
			T^{-3/2}\sum_{t=1}^{T}\widetilde{y_{t-1}^2}z_t^2(\rho_T) \Rightarrow \sigma_\eta\sigma_{\varepsilon}^2\int_{0}^{1}\widetilde{J_{a,2}}(r)dW_\eta(r) + c^2\sigma_{\varepsilon}^4\int_{0}^{1}(\widetilde{J_{a,2}})^2(r)dr + 2c\sigma_{\varepsilon}^4q\int_{0}^{1}\widetilde{J_{a,1}}(r)\widetilde{J_{a,2}}(r)dr,
		\end{align*}
		
		\item[(b)] \begin{align*}
			T^{-1}\sum_{t=1}^{T}\widetilde{y_{t-1}}z_t^2(\rho_T) \Rightarrow \sigma_\eta\sigma_{\varepsilon}\int_{0}^{1}\widetilde{J_{a,1}}(r)dW_\eta(r) + c^2\sigma_{\varepsilon}^3\int_{0}^{1}\widetilde{J_{a,1}}(r)\widetilde{J_{a,2}}(r)dr + 2c\sigma_{\varepsilon}^3q\int_{0}^{1}(\widetilde{J_{a,1}})^2(r)dr.
		\end{align*}
	
		\item[(c)] $\hat{\sigma}_{\widetilde{\xi^*}}^2(\rho_T)\stackrel{p}{\to}\sigma_\eta^2$,
		
		\item[(d)] $\hat{\sigma}_{\widetilde{\xi^{**}}}^2(\rho_T)\stackrel{p}{\to}\sigma_\eta^2$,
	\end{itemize}
	where $\hat{\sigma}_{\widetilde{\xi^*}}^2(\rho_T)$ and $\hat{\sigma}_{\widetilde{\xi^{**}}}^2(\rho_T)$ are the OLS variance estimators of \eqref{model:linearized_augmented_matrix} and \eqref{model:linearized_augmented_matrix_modified}, respectively. \label{lemappB:wc_components}
\end{lemappB}
\noindent \begin{proof}
	For part (a), a straightforward calculation gives
	\begin{align*}
		T^{-3/2}\sum_{t=1}^{T}\widetilde{y_{t-1}^2}z_t^2(\rho_T) = B_{1,T}+B_{2,T}+B_{3,T},
	\end{align*}
	where
	\begin{align*}
		B_{1,T} &\coloneqq T^{-3/2}\sum_{t=1}^{T}\Bigl(y_{t-1}^2-T^{-1}\sum_{t=1}^{T}y_{t-1}^2\Bigr)(\varepsilon_t^2-\sigma_{\varepsilon}^2), \\
		B_{2,T} &\coloneqq c^2T^{-3}\sum_{t=1}^{T}\Bigl(y_{t-1}^2-T^{-1}\sum_{t=1}^{T}y_{t-1}^2\Bigr)y_{t-1}^2v_t^2, \\
		\intertext{and}
		B_{3,T} &\coloneqq 2cT^{-9/4}\sum_{t=1}^{T}\Bigl(y_{t-1}^2-T^{-1}\sum_{t=1}^{T}y_{t-1}^2\Bigr)y_{t-1}\varepsilon_tv_t.
	\end{align*}
	It is straightforward to show $B_{1,T} \Rightarrow \sigma_\eta \sigma_{\varepsilon}^2\int_{0}^{1}\widetilde{J_{a,2}}(r)dW_\eta(r)$, using Lemma \ref{lemappB:wc_Y} and Theorem 2.1 of \citet{hansen1992Convergencestochastica}. As for $B_{2,T}$, we have
	\begin{align*}
		B_{2,T}&= c^2T^{-3}\sum_{t=1}^{T}\Bigl(y_{t-1}^2 - T^{-1}\sum_{t=1}^{T}y_{t-1}^2\Bigr)^2 + c^2T^{-3}\sum_{t=1}^{T}\Bigl(y_{t-1}^2-T^{-1}\sum_{t=1}^{T}y_{t-1}^2\Bigr)y_{t-1}^2(v_t^2-1) \\
		&=c^2\int_{0}^{1}\Bigl(Y_T^2(r)-\int_{0}^{1}Y_T^2(s)ds\Bigr)^2dr + c^2T^{-1/2}\int_{0}^{1}\Bigl(Y_T^2(r)-\int_{0}^{1}Y_T^2(s)ds\Bigr)Y_T^2(r)dW_{v^2-1,T} \\ &\Rightarrow c^2\sigma_{\varepsilon}^4\int_{0}^{1}(\widetilde{J_{a,2}})^2(r)dr,
	\end{align*}
	where the last convergence follows from Lemma \ref{lemappB:wc_Y} and the continuous mapping theorem (CMT). By a similar argument, we obtain
	\begin{align*}
		B_{3,T}&=2c\sigma_{\varepsilon}qT^{-5/2}\sum_{t=1}^{T}\Bigl(y_{t-1}^2 -T^{-1}\sum_{t=1}^{T}y_{t-1}^2\Bigr)\Bigl(y_{t-1}-T^{-1}\sum_{t=1}^{T}y_{t-1}\Bigr) \\ &+2cT^{-9/4}\sum_{t=1}^{T}\Bigl(y_{t-1}^2 -T^{-1}\sum_{t=1}^{T}y_{t-1}^2\Bigr)y_{t-1}(\varepsilon_t v_t-\sigma_{\varepsilon v}) \\
		& \Rightarrow 2c\sigma_{\varepsilon}^4q\int_{0}^{1}\widetilde{J_{a,2}}(r)\widetilde{J_{a,1}}(r)dr.
	\end{align*}
	Therefore, we arrive at
	\begin{align*}
		T^{-3/2}\sum_{t=1}^{T}\widetilde{y_{t-1}^2}z_t^2(\rho_T) \Rightarrow \sigma_\eta \sigma_{\varepsilon}^2\int_{0}^{1}\widetilde{J_{a,2}}(r)dW_\eta(r)+c^2\sigma_{\varepsilon}^4\int_{0}^{1}(\widetilde{J_{a,2}})^2(r)dr+2c\sigma_{\varepsilon}^4q\int_{0}^{1}\widetilde{J_{a,2}}(r)\widetilde{J_{a,1}}(r)dr,
	\end{align*}
	as desired.
	
	Part (b) can be proven in a similar fashion. Write $T^{-1}\sum_{t=1}^{T}\widetilde{y_{t-1}}z_t^2(\rho_T)$ as
	\begin{align*}
		T^{-1}\sum_{t=1}^{T}\widetilde{y_{t-1}}z_t^2(\rho_T)=C_{1,T}+C_{2,T}+C_{3,T},
	\end{align*}
	where $C_{1,T} \coloneqq T^{-1}\sum_{t=1}^{T}\widetilde{y_{t-1}}(\varepsilon_t^2-\sigma_{\varepsilon}^2)$, $C_{2,T} \coloneqq c^2T^{-\frac{5}{2}}\sum_{t=1}^{T}\widetilde{y_{t-1}}y_{t-1}^2v_t^2$, and $C_{3,T} \coloneqq 2cT^{-\frac{7}{4}}\sum_{t=1}^{T}\widetilde{y_{t-1}}y_{t-1}\varepsilon_t v_t$. Then, it is straightforward to show
	\begin{align*}
		C_{1,T} &\Rightarrow \sigma_\eta \sigma_{\varepsilon} \int_{0}^{1}\widetilde{J_{a,1}}(r)dW_\eta(r), \\
		C_{2,T} &\Rightarrow c^2\sigma_{\varepsilon}^3\int_{0}^{1}\widetilde{J_{a,1}}(r)\widetilde{J_{a,2}}(r)dr, \\
		\intertext{and}
		C_{3,T} &\Rightarrow2c\sigma_{\varepsilon}^3q\int_{0}^{1}(\widetilde{J_{a,1}})^2(r)dr.
	\end{align*}
	Combining the above results completes the proof of part (b).
	
	To prove part (c), we define $M\coloneqq I_T-\widetilde{X}(\widetilde{X}'\widetilde{X})^{-1}\widetilde{X}$ and write $\hat{\sigma}_{\widetilde{\xi^*}}^2(\rho_T)$ as 
	\begin{align}
		\hat{\sigma}_{\widetilde{\xi^*}}^2(\rho_T) &= T^{-1}\widetilde{\Xi^*}'M\widetilde{\Xi^*} \notag \\ &=T^{-1}\sum_{t=1}^{T}(\widetilde{\xi^*_t})^2 - T^{-1}\begin{pmatrix} \sum_{t=1}^{T}\widetilde{y_{t-1}}\xi_t^* & \sum_{t=1}^{T}\widetilde{y_{t-1}^2}\xi_t^* \\
		\end{pmatrix} \notag \\
		&\times \begin{pmatrix}
			\sum_{t=1}^{T}(\widetilde{y_{t-1}})^2 & \sum_{t=1}^{T}\widetilde{y_{t-1}}\widetilde{y_{t-1}^2} \\
			\sum_{t=1}^{T}\widetilde{y_{t-1}^2}\widetilde{y_{t-1}} & \sum_{t=1}^{T}(\widetilde{y_{t-1}^2})^2 \\ 
		\end{pmatrix}^{-1} \begin{pmatrix} \sum_{t=1}^{T}\widetilde{y_{t-1}}\xi_t^* \\ \sum_{t=1}^{T}\widetilde{y_{t-1}^2}\xi_t^* \\
	\end{pmatrix}. \label{eqn:OLSvariance_expand}
	\end{align} 
	The first term of \eqref{eqn:OLSvariance_expand} becomes
	\begin{align*}
		T^{-1}\sum_{t=1}^{T}(\widetilde{\xi_t^*})^2 &= T^{-1}\sum_{t=1}^{T}(\xi_t^*)^2 - \Bigl(T^{-1}\sum_{t=1}^{T}\xi_t^*\Bigr)^2 \\ &=T^{-1}\sum_{t=1}^{T}\{c^2T^{-3/2}y_{t-1}^2(v_t^2-1)+2cT^{-3/4}y_{t-1}(\varepsilon_tv_t-\sigma_{\varepsilon v})+(\varepsilon_t^2-\sigma_{\varepsilon}^2)\}^2 \\
		&-\Bigl[T^{-1}\sum_{t=1}^{T}\{c^2T^{-3/2}y_{t-1}^2(v_t^2-1)+2cT^{-3/4}y_{t-1}(\varepsilon_tv_t-\sigma_{\varepsilon v})+(\varepsilon_t^2-\sigma_{\varepsilon}^2)\}\Bigr]^2 \\
		&=T^{-1}\sum_{t=1}^{T}(\varepsilon_t^2-\sigma_{\varepsilon}^2)^2 + o_p(1) \\& \stackrel{p}{\to}\sigma_\eta^2.
	\end{align*}
	The second term of \eqref{eqn:OLSvariance_expand} satisfies
	\begin{align}
		&T^{-1}\begin{pmatrix} \sum_{t=1}^{T}\widetilde{y_{t-1}}\xi_t^* & \sum_{t=1}^{T}\widetilde{y_{t-1}^2}\xi_t^* \\
		\end{pmatrix} \begin{pmatrix}
			\sum_{t=1}^{T}(\widetilde{y_{t-1}})^2 & \sum_{t=1}^{T}\widetilde{y_{t-1}}\widetilde{y_{t-1}^2} \\
			\sum_{t=1}^{T}\widetilde{y_{t-1}^2}\widetilde{y_{t-1}} & \sum_{t=1}^{T}(\widetilde{y_{t-1}^2})^2 \\ 
		\end{pmatrix}^{-1} \begin{pmatrix} \sum_{t=1}^{T}\widetilde{y_{t-1}}\xi_t^* \\ \sum_{t=1}^{T}\widetilde{y_{t-1}^2}\xi_t^* \\
		\end{pmatrix} \notag \\
		&=T^{-1}\begin{pmatrix} T^{-1}\sum_{t=1}^{T}\widetilde{y_{t-1}}(\varepsilon_t^2-\sigma_{\varepsilon}^2) + o_p(1) & T^{-3/2}\sum_{t=1}^{T}\widetilde{y_{t-1}^2}(\varepsilon_t^2-\sigma_{\varepsilon}^2)+o_p(1) \\
		\end{pmatrix} \notag \\ &\times \begin{pmatrix}
			T^{-2}\sum_{t=1}^{T}(\widetilde{y_{t-1}})^2 & T^{-5/2}\sum_{t=1}^{T}\widetilde{y_{t-1}}\widetilde{y_{t-1}^2} \\
			T^{-5/2}\sum_{t=1}^{T}\widetilde{y_{t-1}^2}\widetilde{y_{t-1}} & T^{-3}\sum_{t=1}^{T}(\widetilde{y_{t-1}^2})^2 \\ 
		\end{pmatrix}^{-1} \begin{pmatrix} T^{-1}\sum_{t=1}^{T}\widetilde{y_{t-1}}(\varepsilon_t^2-\sigma_{\varepsilon}^2) + o_p(1) \\ T^{-3/2}\sum_{t=1}^{T}\widetilde{y_{t-1}^2}(\varepsilon_t^2-\sigma_{\varepsilon}^2)+o_p(1) \\
	\end{pmatrix} \notag \\
	&=O_p(T^{-1}). \label{eqn:OLSvariance_negligible}
	\end{align}
	Hence, we obtain $\hat{\sigma}_{\widetilde{\xi^*}}^2(\rho_T) \stackrel{p}{\to}\sigma_\eta^2$, as desired.
	
	To prove part (d), let $\widetilde{Z_1} \coloneqq (\widetilde{z_1}(\rho_T), \widetilde{z_2}(\rho_T),\ldots,\widetilde{z_T}(\rho_T))'$. Then, $\widetilde{Z_2^*}(\rho_T)$ is expressed as
	\begin{align*}
		\widetilde{Z_2^*}(\rho_T) = \frac{1}{\sqrt{1-\hat{\psi}_T^2(\rho_T)}}\Bigl\{\widetilde{Z_2}(\rho_T)-\frac{\hat{\sigma}_{\eta,T}(\rho_T)\hat{\psi}_T(\rho_T)}{\hat{\sigma}_{\varepsilon,T}(\rho_T)}\widetilde{Z_1}(\rho_T)\Bigr\},
	\end{align*}
	which yields
	\begin{align}
		\hat{\sigma}_{\widetilde{\xi^{**}}}^2=T^{-1}\widetilde{Z_2^*}(\rho_T)'M\widetilde{Z_2^*}(\rho_T) =\frac{1}{1-\hat{\psi}_T^2(\rho_T)}\Bigl\{D_{1,T} - 2\frac{\hat{\sigma}_{\eta,T}(\rho_T)\hat{\psi}_T(\rho_T)}{\hat{\sigma}_{\varepsilon,T}(\rho_T)}D_{2,T} + \frac{\hat{\sigma}_{\eta,T}^2(\rho_T)\hat{\psi}_T^2(\rho_T)}{\hat{\sigma}_{\varepsilon,T}^2(\rho_T)}D_{3,T}\Bigr\}, \label{eqn:OLSvariance_mod_expand}
	\end{align}
	where $D_{1,T} \coloneqq T^{-1}\widetilde{Z_2}(\rho_T)'M\widetilde{Z_2}(\rho_T)$, $D_{2,T} \coloneqq T^{-1}\widetilde{Z_2}(\rho_T)'M\widetilde{Z_1}(\rho_T)$, and $D_{3,T} \coloneqq T^{-1}\widetilde{Z_1}(\rho_T)'M\widetilde{Z_1}(\rho_T)$. Since $D_{1,T}$ is $\hat{\sigma}_{\widetilde{\xi^*}}^2(\rho_T)$, we have already proven in part (c) that
	\begin{align}
		D_{1,T} \stackrel{p}{\to} \sigma_\eta^2. \label{eqn:D1}
	\end{align}
	In view of equation \eqref{model:linearized_augmented_matrix}, $D_{2,T}$ becomes
	\begin{align*}
		D_{2,T} &=T^{-1}\widetilde{\Xi^*}(\rho_T)'M\widetilde{Z_1}(\rho_T) \\
		&=T^{-1}\sum_{t=1}^{T}\xi_t^*\widetilde{z_t}(\rho_T) - T^{-1}\begin{pmatrix}
			\sum_{t=1}^{T}\widetilde{y_{t-1}}\xi_t^* & \sum_{t=1}^{T}\widetilde{y_{t-1}^2}\xi_t^* \\
		\end{pmatrix} \\
		& \ \ \ \times \begin{pmatrix}
			\sum_{t=1}^{T}(\widetilde{y_{t-1}})^2 & \sum_{t=1}^{T}\widetilde{y_{t-1}}\widetilde{y_{t-1}^2} \\
			\sum_{t=1}^{T}\widetilde{y_{t-1}^2}\widetilde{y_{t-1}} & \sum_{t=1}^{T}(\widetilde{y_{t-1}^2})^2
		\end{pmatrix}^{-1} \begin{pmatrix}
		\sum_{t=1}^{T}\widetilde{y_{t-1}}z_t(\rho_T) \\ \sum_{t=1}^{T}\widetilde{y_{t-1}^2}z_t(\rho_T) \\
	\end{pmatrix}.
	\end{align*}
	As for the first term of $D_{2,T}$, we have
	\begin{align*}
		T^{-1}\sum_{t=1}^{T}\xi_t^*\widetilde{z_t}(\rho_T) &= T^{-1}\sum_{t=1}^{T}\{c^2T^{-3/2}y_{t-1}^2(v_t^2-1) + 2cT^{-3/4}y_{t-1}(\varepsilon_tv_t-\sigma_{\varepsilon v})+(\varepsilon_t^2-\sigma_{\varepsilon}^2)\} \\ & \ \ \ \ \ \ \times \Bigl(z_t(\rho_T)-T^{-1}\sum_{t=1}^{T}z_t(\rho_T)\Bigr) \\
		&=T^{-1}\sum_{t=1}^{T}(\varepsilon_t^2-\sigma_{\varepsilon}^2)\varepsilon_t + o_p(1) \stackrel{p}{\to} \mathbb{E}[\varepsilon_t^3].
	\end{align*}
	We can also show that the second term of $D_{2,T}$ is $O_p(T^{-1})$ in the same way as we did in \eqref{eqn:OLSvariance_negligible}. Thus, we get
	\begin{align}
		D_{2,T} \stackrel{p}{\to} \mathbb{E}[\varepsilon_t^3]. \label{eqn:D2}
	\end{align}
	Lastly, $D_{3,T}$ becomes
	\begin{align}
		D_{3,T}&=T^{-1}\widetilde{Z_1}(\rho_T)'M\widetilde{Z_1}(\rho_T) \notag \\
		&=T^{-1}\sum_{t=1}^{T}\widetilde{z_t}^2(\rho_T) - T^{-1}\begin{pmatrix}
			\sum_{t=1}^{T}\widetilde{y_{t-1}}z_t(\rho_T) & \sum_{t=1}^{T}\widetilde{y_{t-1}^2}z_t(\rho_T) \\
		\end{pmatrix}' \notag \\ & \ \ \ \ \ \ \ \ \ \ \ \times \begin{pmatrix}
		\sum_{t=1}^{T}(\widetilde{y_{t-1}})^2 & \sum_{t=1}^{T}\widetilde{y_{t-1}}\widetilde{y_{t-1}^2} \\
		\sum_{t=1}^{T}\widetilde{y_{t-1}^2}\widetilde{y_{t-1}} & \sum_{t=1}^{T}(\widetilde{y_{t-1}^2})^2 \\
	\end{pmatrix}^{-1} \begin{pmatrix}
	\sum_{t=1}^{T}\widetilde{y_{t-1}}z_t(\rho_T) \\ \sum_{t=1}^{T}\widetilde{y_{t-1}^2}z_t(\rho_T) \\
	\end{pmatrix} \notag \\
	&=\hat{\sigma}_{\varepsilon,T}^2(\rho_T) - \Bigl(T^{-1}\sum_{t=1}^{T}z_t(\rho_T)\Bigr)^2 + O_p(T^{-1}) \stackrel{p}{\to} \sigma_{\varepsilon}^2. \label{eqn:D3}
	\end{align}
	Substituting \eqref{eqn:D1} through \eqref{eqn:D3} into \eqref{eqn:OLSvariance_mod_expand} and applying Lemma \ref{lemappB:consistency}, we deduce
	\begin{align*}
		\hat{\sigma}_{\widetilde{\xi^{**}}}^2 &\stackrel{p}{\to}\frac{1}{1-\psi^2}\Bigl\{\sigma_\eta^2 - 2\frac{\sigma_\eta\psi}{\sigma_{\varepsilon}}\mathbb{E}[\varepsilon_t^3] + \frac{\sigma_\eta^2\psi^2}{\sigma_{\varepsilon}^2}\sigma_{\varepsilon}^2\Bigr\} \\
		&=\frac{1}{1-\psi^2}(\sigma_\eta^2 - 2\sigma_\eta^2\psi^2 + \sigma_\eta^2\psi^2) = \sigma_\eta^2.
	\end{align*}
\end{proof}

\noindent \begin{proof}[Proof of Theorem \ref{thm:ln_dist}]
	First, note that
	\begin{align*}
		\mathrm{LN}_T(\rho_T)=\frac{T^{-3/2}\sum_{t=1}^{T}\widetilde{y_{t-1}^2}z_t^2(\rho_T)}{\hat{\sigma}_{\eta,T}(\rho_T)\Bigl\{T^{-3}\sum_{t=1}^{T}(\widetilde{y_{t-1}^2})^2\Bigr\}^{1/2}}.
	\end{align*}
	Then, using Lemmas \ref{lemappB:consistency} and \ref{lemappB:wc_components} and the CMT, we deduce
	\begin{align*}
		\mathrm{LN}_T(\rho_T)&\Rightarrow \frac{\sigma_\eta\sigma_{\varepsilon}^2\int_{0}^{1}\widetilde{J_{a,2}}(r)dW_\eta(r) + c^2\sigma_{\varepsilon}^4\int_{0}^{1}(\widetilde{J_{a,2}})^2(r)dr + 2c\sigma_{\varepsilon}^4q\int_{0}^{1}\widetilde{J_{a,1}}(r)\widetilde{J_{a,2}}(r)dr}{\sigma_\eta\Bigl\{\sigma_{\varepsilon}^4\int_{0}^{1}(\widetilde{J_{a,2}})^2(r)dr\Bigr\}^{1/2}} \\
		&=\frac{\int_{0}^{1}\widetilde{J_{a,2}}(r)dW_\eta(r)}{\bigl\{\int_{0}^{1}\bigl(\widetilde{J_{a,2}}\bigr)^2(r)dr\bigr\}^{1/2}} +  \frac{\sigma_\varepsilon^2}{\sigma_\eta}\Biggl[\frac{c^2\int_{0}^{1}\bigl(\widetilde{J_{a,2}}\bigr)^2(r)dr + 2cq\int_{0}^{1}\widetilde{J_{a,2}}(r)\widetilde{J_{a,1}}(r)dr}{\bigl\{\int_{0}^{1}\bigl(\widetilde{J_{a,2}}\bigr)^2(r)dr\bigr\}^{1/2}}\Biggr].
	\end{align*}
\end{proof}

\noindent \begin{proof}[Proof of Theorem \ref{thm:t_wald_dist}]
	First, by Lemma \ref{lemappB:wc_components}(c) and the CMT, the denominator of $t_{\hat{\omega}_T^2}(\rho_T)$ divided by $T^{3/2}$ becomes
	\begin{align*}
		\hat{\sigma}_{\widetilde{\xi^*}}(\rho_T)T^{-3/2}(\widetilde{X_2}'M_1\widetilde{X_2})^{1/2} &=\hat{\sigma}_{\widetilde{\xi^*}}(\rho_T) \{T^{-3}(M_1\widetilde{X_2})'(M_1\widetilde{X_2})\}^{1/2} \\
		&= \hat{\sigma}_{\widetilde{\xi^*}}(\rho_T)\biggl\{T^{-3}\sum_{t=1}^{T}\Bigl(\widetilde{y_{t-1}^2}-\frac{\sum_{t=1}^{T}\widetilde{y_{t-1}}\widetilde{y_{t-1}^2}}{\sum_{t=1}^{T}(\widetilde{y_{t-1}})^2}\widetilde{y_{t-1}}\Bigr)^2\biggr\}^{1/2} \\
		&=\hat{\sigma}_{\widetilde{\xi^*}}(\rho_T)\biggl\{\int_{0}^{1}\Bigl(\widetilde{Y_{2,T}}(r)-\frac{\int_{0}^{1}\widetilde{Y_{1,T}}(s)\widetilde{Y_{2,T}}(s)ds}{\int_{0}^{1}(\widetilde{Y_{1,T}})^2(s)ds}\widetilde{Y_{1,T}}(r)\Bigr)^2\biggr\}^{1/2} \\
		&\Rightarrow \sigma_\eta \biggl\{\sigma_{\varepsilon}^4\int_{0}^{1}\Bigl(\widetilde{J_{a,2}}(r)-\frac{\int_{0}^{1}\widetilde{J_{a,1}}(s)\widetilde{J_{a,2}}(s)ds}{\int_{0}^{1}(\widetilde{J_{a,1}})^2(s)ds}\widetilde{J_{a,1}}(r)\Bigr)^2dr\biggr\}^{1/2} \\
		&= \sigma_\eta \sigma_{\varepsilon}^2\Bigl[\int_{0}^{1}Q_a^2(r)dr\Bigr]^{1/2},
	\end{align*}
	where $\widetilde{Y_{1,T}}(r) \coloneqq Y_T(r)-\int_{0}^{1}Y_T(s)ds$ and $\widetilde{Y_{2,T}}(r) \coloneqq Y_T^2(r) - \int_{0}^{1}Y_T^2(s)ds$. Next, applying Lemma \ref{lemappB:wc_components}, the numerator of $t_{\hat{\omega}_T^2}(\rho_T)$ divided by $T^{3/2}$ is seen to satisfy
	\begin{align*}
		T^{-3/2}\widetilde{X_2}'M_1\widetilde{Z_2}(\rho_T) &= T^{-3/2}\sum_{t=1}^{T}\widetilde{y_{t-1}^2}z_t^2(\rho_T) - \frac{T^{-5/2}\sum_{t=1}^{T}\widetilde{y_{t-1}}\widetilde{y_{t-1}^2}T^{-1}\sum_{t=1}^{T}\widetilde{y_{t-1}}z_t^2(\rho_T)}{T^{-2}\sum_{t=1}^{T}(\widetilde{y_{t-1}})^2} \\
		&\Rightarrow \sigma_\eta\sigma_{\varepsilon}^2\int_{0}^{1}\widetilde{J_{a,2}}(r)dW_\eta(r) + c^2\sigma_{\varepsilon}^4\int_{0}^{1}(\widetilde{J_{a,2}})^2(r)dr + 2c\sigma_{\varepsilon}^4q\int_{0}^{1}\widetilde{J_{a,1}}(r)\widetilde{J_{a,2}}(r)dr \\
		&- \sigma_{\varepsilon}^3\int_{0}^{1}\widetilde{J_{a,1}}(r)\widetilde{J_{a,2}}(r)dr \\ &\times \frac{\sigma_\eta\sigma_{\varepsilon}\int_{0}^{1}\widetilde{J_{a,1}}(r)dW_\eta(r) + c^2\sigma_{\varepsilon}^3\int_{0}^{1}\widetilde{J_{a,1}}(r)\widetilde{J_{a,2}}(r)dr + 2c\sigma_{\varepsilon}^3q\int_{0}^{1}(\widetilde{J_{a,1}})^2(r)dr}{\sigma_{\varepsilon}^2\int_{0}^{1}(\widetilde{J_{a,1}})^2(r)dr} \\
		&=\sigma_\eta \sigma_{\varepsilon}^2 \int_{0}^{1}Q_a(r)dW_\eta(r)+c^2\sigma_{\varepsilon}^4\int_{0}^{1}Q_a^2(r)dr.
	\end{align*}
	Combining the above results gives
	\begin{align*}
		t_{\hat{\omega}_T^2}(\rho_T) \Rightarrow \frac{\int_{0}^{1}Q_a(r)dW_\eta(r)}{\bigl[\int_{0}^{1}Q_a^2(r)dr\bigr]^{1/2}} + \frac{c^2\sigma_\varepsilon^2}{\sigma_\eta}\Biggl[\int_{0}^{1}Q_a^2(r)dr\Biggr]^{1/2}.
	\end{align*}

	To derive the asymptotic distribution of $W_T(\rho_T)$, note that
	\begin{align*}
		W_T(\rho_T) &=\hat{\sigma}_{\widetilde{\xi^*}}^{-2}(\widetilde{X}'\widetilde{Z_2}(\rho_T))'(\widetilde{X}'\widetilde{X})^{-1}(\widetilde{X}'\widetilde{Z_2}(\rho_T)) \\
		&=\hat{\sigma}_{\widetilde{\xi^*}}^{-2}\begin{pmatrix} T^{-1}\sum_{t=1}^{T}\widetilde{y_{t-1}}z_t^2(\rho_T) \\ T^{-3/2}\sum_{t=1}^{T}\widetilde{y_{t-1}^2}z_t^2(\rho_T) \\
		\end{pmatrix}' \begin{pmatrix}
			T^{-2}\sum_{t=1}^{T}(\widetilde{y_{t-1}})^2 & T^{-5/2}\sum_{t=1}^{T}\widetilde{y_{t-1}}\widetilde{y_{t-1}^2} \\
			T^{-5/2}\sum_{t=1}^{T}\widetilde{y_{t-1}^2}\widetilde{y_{t-1}} & T^{-3}\sum_{t=1}^{T}(\widetilde{y_{t-1}^2})^2 \\ 
		\end{pmatrix}^{-1} \\
	& \ \ \ \ \ \ \ \ \ \ \ \times \begin{pmatrix} T^{-1}\sum_{t=1}^{T}\widetilde{y_{t-1}}z_t^2(\rho_T) \\ T^{-3/2}\sum_{t=1}^{T}\widetilde{y_{t-1}^2}z_t^2(\rho_T) \\
		\end{pmatrix}.
	\end{align*}
	Then, applying Lemma \ref{lemappB:wc_components} and the CMT, we get
	\begin{align*}
		W_T(\rho_T)&\Rightarrow \sigma_\eta^{-2}\begin{pmatrix} \sigma_\eta\sigma_{\varepsilon}\int_{0}^{1}\widetilde{J_{a,1}}(r)dW_\eta(r) + c^2\sigma_{\varepsilon}^3\int_{0}^{1}\widetilde{J_{a,1}}(r)\widetilde{J_{a,2}}(r)dr + 2c\sigma_{\varepsilon}^3q\int_{0}^{1}(\widetilde{J_{a,1}})^2(r)dr \\ \sigma_\eta\sigma_{\varepsilon}^2\int_{0}^{1}\widetilde{J_{a,2}}(r)dW_\eta(r) + c^2\sigma_{\varepsilon}^4\int_{0}^{1}(\widetilde{J_{a,2}})^2(r)dr + 2c\sigma_{\varepsilon}^4q\int_{0}^{1}\widetilde{J_{a,1}}(r)\widetilde{J_{a,2}}(r)dr \\
		\end{pmatrix}' \\
	&\times \begin{pmatrix}
		\sigma_{\varepsilon}^2\int_{0}^{1}\bigl(\widetilde{J_{a,1}}\bigr)^2(r)dr & \sigma_{\varepsilon}^3\int_{0}^{1}\widetilde{J_{a,1}}(r)\widetilde{J_{a,2}}(r)dr \\ 
		\sigma_{\varepsilon}^3\int_{0}^{1}\widetilde{J_{a,2}}(r)\widetilde{J_{a,1}}(r)dr &
		\sigma_{\varepsilon}^4\int_{0}^{1}\bigl(\widetilde{J_{a,2}}\bigr)^2(r)dr \\
	\end{pmatrix}^{-1} \\
	&\times \begin{pmatrix} \sigma_\eta\sigma_{\varepsilon}\int_{0}^{1}\widetilde{J_{a,1}}(r)dW_\eta(r) + c^2\sigma_{\varepsilon}^3\int_{0}^{1}\widetilde{J_{a,1}}(r)\widetilde{J_{a,2}}(r)dr + 2c\sigma_{\varepsilon}^3q\int_{0}^{1}(\widetilde{J_{a,1}})^2(r)dr \\ \sigma_\eta\sigma_{\varepsilon}^2\int_{0}^{1}\widetilde{J_{a,2}}(r)dW_\eta(r) + c^2\sigma_{\varepsilon}^4\int_{0}^{1}(\widetilde{J_{a,2}})^2(r)dr + 2c\sigma_{\varepsilon}^4q\int_{0}^{1}\widetilde{J_{a,1}}(r)\widetilde{J_{a,2}}(r)dr \\
	\end{pmatrix} \\
	&=\Bigg\{\begin{pmatrix}
		\int_{0}^{1}\widetilde{J_{a,1}}(r)dW_\eta(r) \\ \int_{0}^{1}\widetilde{J_{a,2}}(r)dW_\eta(r) \end{pmatrix} + \frac{\sigma_\varepsilon^2}{\sigma_\eta} \begin{pmatrix} c^2\int_{0}^{1}\widetilde{J_{a,1}}(r)\widetilde{J_{a,2}}(r)dr+2cq\int_{0}^{1}\bigl(\widetilde{J_{a,1}}\bigr)^2(r)dr \\  c^2\int_{0}^{1}\bigl(\widetilde{J_{a,2}}\bigr)^2(r)dr+2cq\int_{0}^{1}\widetilde{J_{a,1}}(r)\widetilde{J_{a,2}}(r)dr
	\end{pmatrix}\Biggr\}' \notag \\ &\times \begin{pmatrix}
		\int_{0}^{1}\bigl(\widetilde{J_{a,1}}\bigr)^2(r)dr & \int_{0}^{1}\widetilde{J_{a,1}}(r)\widetilde{J_{a,2}}(r)dr \\ \int_{0}^{1}\widetilde{J_{a,2}}(r)\widetilde{J_{a,1}}(r)dr & \int_{0}^{1}\bigl(\widetilde{J_{a,2}}\bigr)^2(r)dr
	\end{pmatrix}^{-1} \notag \\ &\times \Bigg\{\begin{pmatrix}
		\int_{0}^{1}\widetilde{J_{a,1}}(r)dW_\eta(r) \\ \int_{0}^{1}\widetilde{J_{a,2}}(r)dW_\eta(r) \end{pmatrix} + \frac{\sigma_\varepsilon^2}{\sigma_\eta} \begin{pmatrix} c^2\int_{0}^{1}\widetilde{J_{a,1}}(r)\widetilde{J_{a,2}}(r)dr+2cq\int_{0}^{1}\bigl(\widetilde{J_{a,1}}\bigr)^2(r)dr \\  c^2\int_{0}^{1}\bigl(\widetilde{J_{a,2}}\bigr)^2(r)dr+2cq\int_{0}^{1}\widetilde{J_{a,1}}(r)\widetilde{J_{a,2}}(r)dr
	\end{pmatrix}\Biggr\},
	\end{align*}
	completing the proof.
\end{proof}

To prove Theorem \ref{thm:mod_ln_t_Wald_dist}, we use the following lemma.
\begin{lemappB}
	Consider model \eqref{model:near_unity_rca} under Assumptions \ref{asm:local} and \ref{asm:localized_corr}. Then, we have
\begin{itemize}
	\item[(a)] 
	\begin{align*}
		T^{-3/2}\sum_{t=1}^{T}\widetilde{y_{t-1}^2}z_t^{2*}(\rho_T) \Rightarrow& \sigma_\eta\sigma_{\varepsilon}^2\int_{0}^{1}\widetilde{J_{a,2}}(r)dW_1(r) \\ &+ (1-\psi^2)^{-1/2}\Bigl\{ c^2\sigma_{\varepsilon}^4\int_{0}^{1}(\widetilde{J_{a,2}})^2(r)dr
		+   2c\sigma_{\varepsilon}^4q\int_{0}^{1}\widetilde{J_{a,1}}(r)\widetilde{J_{a,2}}(r)dr\Bigr\},
	\end{align*}
	
	\item[(b)] \begin{align*}
		T^{-1}\sum_{t=1}^{T}\widetilde{y_{t-1}}z_t^{2*}(\rho_T) \Rightarrow &\sigma_\eta\sigma_{\varepsilon}\int_{0}^{1}\widetilde{J_{a,1}}(r)dW_1(r) \\ &+ (1-\psi^2)^{-1/2}\Bigl\{ c^2\sigma_{\varepsilon}^3\int_{0}^{1}\widetilde{J_{a,1}}(r)\widetilde{J_{a,2}}(r)dr 
		+ 2c\sigma_{\varepsilon}^3q\int_{0}^{1}(\widetilde{J_{a,1}})^2(r)dr\Bigr\}.
	\end{align*}
\end{itemize} \label{lemappB:wc_components_mod}
\end{lemappB}

\noindent \begin{proof}
	To prove part (a), note that
	\begin{align*}
		T^{-3/2}\sum_{t=1}^{T}\widetilde{y_{t-1}^2}z_t^{2*}(\rho_T) = \frac{1}{\sqrt{1-\hat{\psi}_T^2(\rho_T)}}\Bigl\{E_{1,T} - \frac{\hat{\sigma}_{\eta,T}(\rho_T)\hat{\psi}_T(\rho_T)}{\hat{\sigma}_{\varepsilon,T}(\rho_T)}E_{2,T}\Bigr\},
	\end{align*}
	where $E_{1,T} \coloneqq T^{-3/2}\sum_{t=1}^{T}\widetilde{y_{t-1}^2}z_t^2(\rho_T)$ and $E_{2,T}\coloneqq T^{-3/2}\sum_{t=1}^{T}\widetilde{y_{t-1}^2}z_t(\rho_T)$. By Lemma \ref{lemappB:wc_components}(a), $E_{1,T}$ satisfies
	\begin{align*}
		E_{1,T} \Rightarrow \sigma_\eta\sigma_{\varepsilon}^2\int_{0}^{1}\widetilde{J_{a,2}}(r)dW_\eta(r) + c^2\sigma_{\varepsilon}^4\int_{0}^{1}(\widetilde{J_{a,2}})^2(r)dr + 2c\sigma_{\varepsilon}^4q\int_{0}^{1}\widetilde{J_{a,1}}(r)\widetilde{J_{a,2}}(r)dr.
	\end{align*}
	As for $E_{2,T}$, a straightforward calculation yields
	\begin{align*}
		E_{2,T} &= T^{-3/2}\sum_{t=1}^{T}\widetilde{y_{t-1}^2}(cT^{-3/4}y_{t-1}v_t + \varepsilon_t) \\
		&=T^{-3/2}\sum_{t=1}^{T}\widetilde{y_{t-1}^2}\varepsilon_t + O_p(T^{-1/4}) \Rightarrow \sigma_{\varepsilon}^3\int_{0}^{1}\widetilde{J_{a,2}}(r)dW_\varepsilon(r).
	\end{align*}
	Hence, we obtain
	\begin{align*}
		T^{-3/2}\sum_{t=1}^{T}\widetilde{y_{t-1}^2}z_t^{2*}(\rho_T)  &\Rightarrow \frac{1}{\sqrt{1-\psi^2}}\biggl\{\sigma_\eta\sigma_{\varepsilon}^2\int_{0}^{1}\widetilde{J_{a,2}}(r)dW_\eta(r) + c^2\sigma_{\varepsilon}^4\int_{0}^{1}(\widetilde{J_{a,2}})^2(r)dr \\ & \ \ \ \ \ \ \ + 2c\sigma_{\varepsilon}^4q\int_{0}^{1}\widetilde{J_{a,1}}(r)\widetilde{J_{a,2}}(r)dr - \frac{\sigma_\eta \psi}{\sigma_\varepsilon}\sigma_{\varepsilon}^3\int_{0}^{1}\widetilde{J_{a,2}}(r)dW_\varepsilon(r)\biggr\} \\
		&=\sigma_\eta\sigma_{\varepsilon}^2\int_{0}^{1}\widetilde{J_{a,2}}(r)d\Bigl(\frac{W_\eta(r)-\psi W_\varepsilon(r)}{\sqrt{1-\psi^2}}\Bigr) \\
		& \ \ \ \ \ \ +\frac{1}{\sqrt{1-\psi^2}}\Bigl(c^2\sigma_{\varepsilon}^4\int_{0}^{1}(\widetilde{J_{a,2}})^2(r)dr+2c\sigma_{\varepsilon}^4q\int_{0}^{1}\widetilde{J_{a,1}}(r)\widetilde{J_{a,2}}(r)dr\Bigr) \\
		&\stackrel{d}{=}\sigma_\eta\sigma_{\varepsilon}^2\int_{0}^{1}\widetilde{J_{a,2}}(r)dW_1(r) \\
		& \ \ \ \ \ \ +\frac{1}{\sqrt{1-\psi^2}}\Bigl(c^2\sigma_{\varepsilon}^4\int_{0}^{1}(\widetilde{J_{a,2}})^2(r)dr+2c\sigma_{\varepsilon}^4q\int_{0}^{1}\widetilde{J_{a,1}}(r)\widetilde{J_{a,2}}(r)dr\Bigr),
	\end{align*}
	in view of \eqref{eqn:w_eta_varepsilon_1}. This proves part (a). The proof of part (b) is similar and thus is omitted.
\end{proof}

\noindent \begin{proof}[Proof of Theorem \ref{thm:mod_ln_t_Wald_dist}]
	The proof for $\mathrm{LN}_T^*(\rho_T)$ is essentially the same as that of Theorem \ref{thm:ln_dist} except that we consider $\sum_{t=1}^{T}\widetilde{y_{t-1}^2}z_t^{2*}(\rho_T)$ in the numerator of the test statistic. Dividing both the numerator and denominator by $T^{3/2}$ and applying Lemma \ref{lemappB:wc_components_mod}(a) leads to the desired result. The proof for the augmented tests goes along the same lines as those of Theorem \ref{thm:t_wald_dist} if we replace $z_t^2(\rho_T)$ with $z_t^{2*}(\rho_T)$ and apply Lemmas \ref{lemappB:wc_components}(d) and \ref{lemappB:wc_components_mod}.
\end{proof}

\section*{Appendix C: Proofs of Results in Section 3}
\setcounter{equation}{0}
\renewcommand{\theequation}{C.\arabic{equation}}
In this appendix, we prove the asymptotic results mentioned in Section 3: namely, the asymptotic distribution of $\hat{\rho}_T$ and the consistency of $\hat{\sigma}_{\varepsilon,T}^2(\hat{\rho}_T)$, $\hat{\sigma}_{\eta,T}^2(\hat{\rho}_T)$ and $\hat{\psi}_T(\hat{\rho}_T)$. 

\begin{lemappC}
	Consider model \eqref{model:near_unity_rca} under Assumptions \ref{asm:local} and \ref{asm:localized_corr}. Then, we have
	\begin{align*}
		T(\hat{\rho}_T - \rho_T) \Rightarrow \frac{\int_{0}^{1}J_a(r)W_\varepsilon(r)}{\int_{0}^{1}J_a^2(r)dr}.
	\end{align*} \label{lemappC:rho_hat}
\end{lemappC}

\noindent \begin{proof}
	From the definition of $\hat{\rho}_T$, we have
	\begin{align*}
		T(\hat{\rho}_T-\rho_T) &= \frac{T^{-1}\sum_{t=1}^{T}y_{t-1}(cT^{-3/4}y_{t-1}v_t+\varepsilon_t)}{T^{-2}\sum_{t=1}^{T}y_{t-1}^2} \\
		&=\frac{cT^{-1/4}\int_{0}^{1}Y_T^2(r)dW_{v,T}(r)+\sigma_{\varepsilon}\int_{0}^{1}Y_T(r)dW_{\varepsilon,T}(r)}{\int_{0}^{1}Y_T^2(r)dr} \Rightarrow \frac{\int_{0}^{1}J_a(r)dW_\varepsilon(r)}{\int_{0}^{1}J_a^2(r)}.
	\end{align*} 
\end{proof}

\begin{lemappC}
	Consider model \eqref{model:near_unity_rca} under Assumptions \ref{asm:local} and \ref{asm:localized_corr}. Then, we have
	\begin{itemize}
		\item[(a)] $\hat{\sigma}_{\varepsilon,T}^2(\hat{\rho}_T) \stackrel{p}{\to} \sigma_{\varepsilon}^2$,
		
		\item[(b)] $\hat{\sigma}_{\eta,T}^2(\hat{\rho}_T) \stackrel{p}{\to} \sigma_{\eta}^2$,
		
		\item[(c)] $\hat{\psi}_T(\hat{\rho}_T) \stackrel{p}{\to} \psi$.
	\end{itemize}
\end{lemappC}

\noindent \begin{proof}
	As for part (a), $\hat{\sigma}_{\varepsilon,T}^2(\hat{\rho}_T)$ satisfies
	\begin{align*}
		\hat{\sigma}_{\varepsilon,T}^2(\hat{\rho}_T) &= T^{-1}\sum_{t=1}^{T}z_t^2(\hat{\rho}_T) \\
		&=T^{-1}\sum_{t=1}^{T}\{z_t(\rho_T)-(\hat{\rho}_T-\rho_T)y_{t-1}\}^2 \\
		&=T^{-1}\sum_{t=1}^{T}z_t^2(\rho_T) - 2T(\hat{\rho}_T-\rho_T)T^{-2}\sum_{t=1}^{T}y_{t-1}z_t(\rho_T) + T^{2}(\hat{\rho}_T-\rho_T)^2T^{-3}\sum_{t=1}^{T}y_{t-1}^2,
	\end{align*}
	for which we have $T^{-1}\sum_{t=1}^{T}z_t^2(\rho_T)=\hat{\sigma}_{\varepsilon,T}^2(\rho_T)$, $T^2(\hat{\rho}_T-\rho_T)^2T^{-3}\sum_{t=1}^{T}y_{t-1}^2=O_p(T^{-1})$ by Lemma \ref{lemappC:rho_hat} and the CMT, and
	\begin{align*}
		T^{-2}\sum_{t=1}^{T}y_{t-1}z_t(\rho_T) &= T^{-2}\sum_{t=1}^{T}y_{t-1}(cT^{-3/4}y_{t-1}v_t+\varepsilon_t) \\
		&=cT^{-5/4}\int_{0}^{1}Y_T^2(r)dW_{T,v}(r) + T^{-1}\sigma_{\varepsilon}\int_{0}^{1}Y_T(r)dW_{\varepsilon,T}(r) = O_p(T^{-1}).
	\end{align*}
	Therefore
	\begin{align*}
		\hat{\sigma}_{\varepsilon,T}^2(\hat{\rho}_T) = \hat{\sigma}_{\varepsilon,T}^2(\rho_T) + o_p(1) \stackrel{p}{\to}\sigma_{\varepsilon}^2,
	\end{align*}
	given Lemma \ref{lemappB:consistency}(a).
	
	To prove part (b), write $\hat{\sigma}_{\eta,T}^2(\hat{\rho}_T)$ as
	\begin{align}
		\hat{\sigma}_{\eta,T}^2(\hat{\rho}_T) = T^{-1}\sum_{t=1}^{T}\{z_t^2(\hat{\rho}_T)-\hat{\sigma}_{\varepsilon,T}^2(\hat{\rho}_T)\}^2 = T^{-1}\sum_{t=1}^{T}z_t^4(\hat{\rho}_T) - \hat{\sigma}_{\varepsilon,T}^4(\hat{\rho}_T). \label{eqn:kappa_hat_extend}
	\end{align}
	The first term is
	\begin{align*}
		T^{-1}\sum_{t=1}^{T}z_t^4(\hat{\rho}_T) &= T^{-1}\sum_{t=1}^{T}\{z_t(\rho_T) - (\hat{\rho}_T-\rho_T)y_{t-1}\}^4 \\
		&=T^{-1}\sum_{t=1}^{T}z_t^4(\rho_T) - 4T(\hat{\rho}_T-\rho_T)T^{-2}\sum_{t=1}^{T}z_t^3(\rho_T)y_{t-1} + 6T^2(\hat{\rho}_T-\rho_T)^2T^{-3}\sum_{t=1}^{T}z_t^2(\rho_T)y_{t-1}^2 \\
		&-4T^3(\hat{\rho}_T-\rho_T)^3T^{-4}\sum_{t=1}^{T}z_t(\rho_T)y_{t-1}^3 + T^4(\hat{\rho}_T-\rho_T)^4T^{-5}\sum_{t=1}^{T}y_{t-1}^4.
	\end{align*}
	Straightforward calculations reveal
	\begin{align*}
		&T^{-2}\sum_{t=1}^{T}z_t^3(\rho_T)y_{t-1}=O_p(T^{-1/2}), \quad T^{-3}\sum_{t=1}^{T}z_t^2(\rho_T)y_{t-1}^2=O_p(T^{-1}), \\
		&T^{-4}\sum_{t=1}^{T}z_t(\rho_T)y_{t-1}^3=O_p(T^{-2}), \quad T^{-5}\sum_{t=1}^{T}y_{t-1}^4=O_p(T^{-2}),
	\end{align*}
	which gives
	\begin{align*}
		T^{-1}\sum_{t=1}^{T}z_t^4(\hat{\rho}_T)=T^{-1}\sum_{t=1}^{T}z_t^4(\rho_T)+O_p(T^{-1/2}).
	\end{align*}
	Substituting this and $\hat{\sigma}_{\varepsilon,T}^2(\hat{\rho}_T) = \hat{\sigma}_{\varepsilon,T}^2(\rho_T)+o_p(1)$ into \eqref{eqn:kappa_hat_extend}, we arrive at
	\begin{align*}
		\hat{\sigma}_{\eta,T}^2(\hat{\rho}_T) &= T^{-1}\sum_{t=1}^{T}z_t^4(\rho_T) - \hat{\sigma}_{\varepsilon,T}^4(\rho_T) + o_p(1) \\
		&=\hat{\sigma}_{\eta,T}^2(\rho_T) + o_p(1) \stackrel{p}{\to} \sigma_{\eta}^2,
	\end{align*}
	by Lemma \ref{lemappB:consistency}(b).
	
	To prove part (c), it suffices to show $T^{-1}\sum_{t=1}^{T}z_t(\hat{\rho}_T)\{z_t^2(\hat{\rho}_T)-\hat{\sigma}_{\varepsilon,T}^2(\hat{\rho}_T)\} \stackrel{p}{\to} \mathbb{E}[\varepsilon_t^3]$, given that $\hat{\sigma}_{\varepsilon,T}^2(\hat{\rho}_T) \stackrel{p}{\to} \sigma_{\varepsilon}^2$ and $\hat{\sigma}_{\eta,T}^2(\hat{\rho}_T) \stackrel{p}{\to}\sigma_{\eta}^2$. Now, we have
	\begin{align*}
		T^{-1}\sum_{t=1}^{T}&z_t(\hat{\rho}_T)\{z_t^2(\hat{\rho}_T)-\hat{\sigma}_{\varepsilon,T}^2(\hat{\rho}_T)\} \\
		&=T^{-1}\sum_{t=1}^{T}\{z_t(\rho_T)-(\hat{\rho}_T-\rho_T)y_{t-1}\}\{z_t^2(\rho_T) - 2(\hat{\rho}_T-\rho_T)y_{t-1}z_t(\rho_T)+(\hat{\rho}_T-\rho_T)^2y_{t-1}^2 -\hat{\sigma}_{\varepsilon,T}^2(\hat{\rho}_T)\} \\
		&=T^{-1}\sum_{t=1}^{T}z_t(\rho_T)\{z_t^2(\rho_T)-\hat{\sigma}_{\varepsilon,T}^2(\hat{\rho}_T)\} - 3T(\hat{\rho}_T-\rho_T)T^{-2}\sum_{t=1}^{T}y_{t-1}z_t^2(\rho_T) \\ &+ 3T^2(\hat{\rho}_T-\rho_T)^2T^{-3}\sum_{t=1}^{T}y_{t-1}^2z_t(\rho_T)
		-T^3(\hat{\rho}_T-\rho_T)^3T^{-4}\sum_{t=1}^{T}y_{t-1}^3 - \hat{\sigma}_{\varepsilon,T}^2(\hat{\rho}_T)T(\hat{\rho}_T-\rho_T)T^{-2}\sum_{t=1}^{T}y_{t-1} \\
		&=T^{-1}\sum_{t=1}^{T}z_t(\rho_T)\{z_t^2(\rho_T)-\hat{\sigma}_{\varepsilon,T}^2(\rho_T)\} + o_p(1) \stackrel{p}{\to}\mathbb{E}[\varepsilon_t^3],
	\end{align*}
	in view of the last line of the proof of Lemma \ref{lemappB:consistency}(c).
\end{proof}

\clearpage
\begin{table} \caption{Significance levels of the equal-tailed confidence interval for $a$}
	\centering
	\begin{threeparttable}
	\renewcommand{\arraystretch}{1.8}	\begin{tabular}{|c|@{\hskip 13pt}c@{\hskip 13pt}|@{\hskip 9pt}|c|@{\hskip 13pt}c@{\hskip 13pt}|}
		\hline $|\psi|\in$&$\alpha_1$&$|\psi|\in$&$\alpha_1$ \\ 
		\hline $[0,0.05)$&0.09&$(0.55,0.6]$&0.42 \\
		\hline $[0.05,0.1)$&0.17&$(0.6,0.65]$&0.38 \\
		\hline $[0.1,0.15)$&0.23&$(0.65,0.7]$&0.35 \\
		\hline $[0.15,0.2)$&0.31&$(0.7,0.75]$&0.31 \\
		\hline $[0.2,0.25)$&0.38&$(0.75,0.8]$&0.26 \\
		\hline $[0.25,0.3)$&0.45&$(0.8,0.85]$&0.22 \\
		\hline $[0.3,0.4]$&0.5&$(0.85,0.9]$&0.17 \\
		\hline $(0.4,0.45]$&0.48&$(0.9,0.95]$&0.11 \\
		\hline $(0.45,0.5]$&0.46&$(0.95,1)$&0.05 \\
		\hline $(0.5,0.55]$&0.44&\multicolumn{2}{|c|}{} \\
		\hline 
		\end{tabular}
		\begin{tablenotes}
			\footnotesize
			\item[a.] Entries in the second and fourth columns are the significance level of the equal-tailed confidence interval for $\rho_T$ when the significance levels of the Bonferroni-Wald test ($\tilde{\alpha}$) and individual modified Wald test ($\alpha_2$) are 0.05.
			\item[b.] To determine the $\alpha_1$ value for the interval $(0.95,1)$, we actually computed type 1 errors over $(0.95,0.99]$.
		\end{tablenotes}
	\end{threeparttable}	\label{tab:significance_ci}
\end{table} 

\clearpage
\begin{table} \caption{Data description of the series used in Section 5}
	\centering
	\begin{threeparttable}
		\renewcommand{\arraystretch}{1.8}	\begin{tabular}{|c|c|c|c|}
			\hline Series&Frequency&Sample period&$T$ \\
			\hline CPI (1982-1984=100)&Monthly&Jan. 1913-Dec. 2019&1284 \\
			\hline Real GDP (2012 chained)&Quarterly&Q1 1947-Q4 2019&292 \\
			\hline Industrial production (2017=100)&Monthly&Jan. 1919-Dec. 2019&1212 \\
			\hline M2&Weekly&Nov. 3, 1980-Dec. 30, 2019&2044 \\
			\hline S\&P 500&Daily&Dec. 3, 2012-Dec. 31, 2019&1782 \\
			\hline 3 month T-bill rate&Daily&Dec. 3, 2012-Dec. 31, 2019&1770 \\
			\hline Unemployment rate&Monthly&Jan. 1948-Dec. 2019&864 \\
			\hline 
		\end{tabular}
	\end{threeparttable}	\label{tab:empirical_description}
\end{table} 

\clearpage
\begin{table} \caption{Estimation and testing results from the empirical analysis in Section 5}
	\centering
	\begin{threeparttable}
		\renewcommand{\arraystretch}{1.8}	\begin{tabular}{|c|c|c|c|c|c|c|c|}
			\hline Series&$T$&$\hat{\rho}_T$&$\hat{\psi}(\hat{\rho_T})$&$\hat{\omega}^2_{\mathrm{HT}}$&Bonf-Wald&LN ($\widetilde{G}_{T,1}$)&HT ($\Theta_{T,R}$) \\
			\hline CPI&1284&0.999&0.211&$4.12\times10^{-4}$&Yes&Yes&- \\
			\hline GDP&292&0.994&-0.065&$1.16\times10^{-3}$&-&-&- \\
			\hline Industrial production&1212&0.998&0.054&$3.98\times10^{-3}$&Yes&Yes&- \\
			\hline M2&2044&0.995&0.066&$3.12\times10^{-4}$&Yes&-&- \\
			\hline S\&P 500&1782&0.984&-0.258&$9.38\times10^{-3}$&Yes&Yes&- \\
			\hline T-bill rate&1770&0.999&0.108&$-4.24\times10^{-5}$&-&-&- \\
			\hline Unemployment rate&864&0.993&0.142&$4.03\times10^{-4}$&Yes&-&- \\
			\hline 
		\end{tabular}
		\begin{tablenotes}
			\footnotesize
			\item For entries in the last three columns, ``Yes" (``-") signifies the rejection (nonrejection) of the null $H_0:\omega^2=0$ by the corresponding test with 5\% significance level.
		\end{tablenotes}
	\end{threeparttable}	\label{tab:empirical_results}
\end{table} 

\clearpage
\begin{figure}
	\includegraphics[width=14cm]{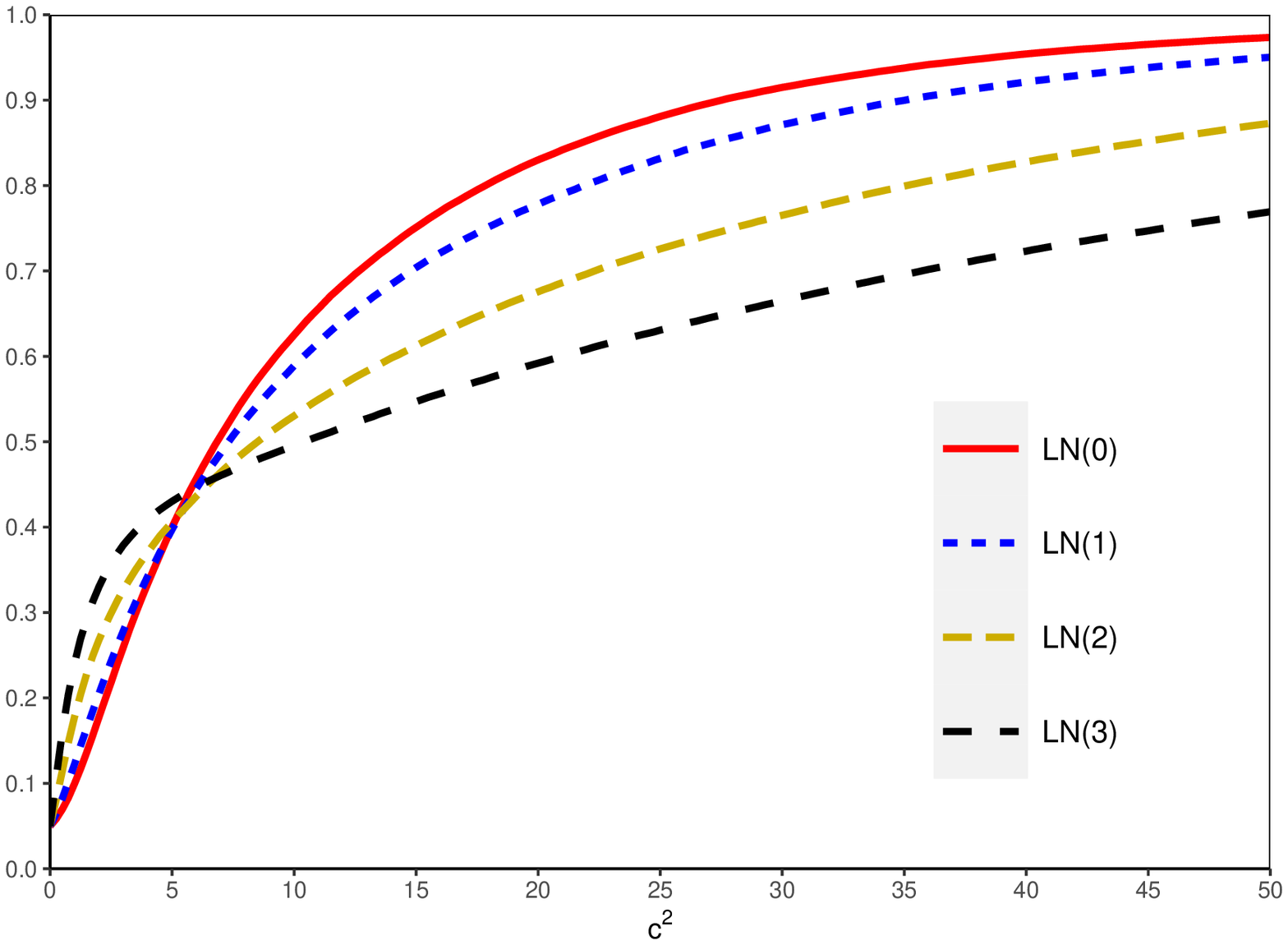}
	\caption{Asymptotic power functions of the Lee-Nagakura test for $q=0,1,2,3$}
	\medskip
	\centering
	\begin{minipage}{10cm}
		{(The values in the parentheses denote the values of $q$.) \par}
	\end{minipage}
	\label{fig:power_asym_Lee}
\end{figure}

\clearpage
\begin{sidewaysfigure}
	\centering
	\begin{subfigure}{0.47\textwidth}
	\includegraphics[width=\textwidth]{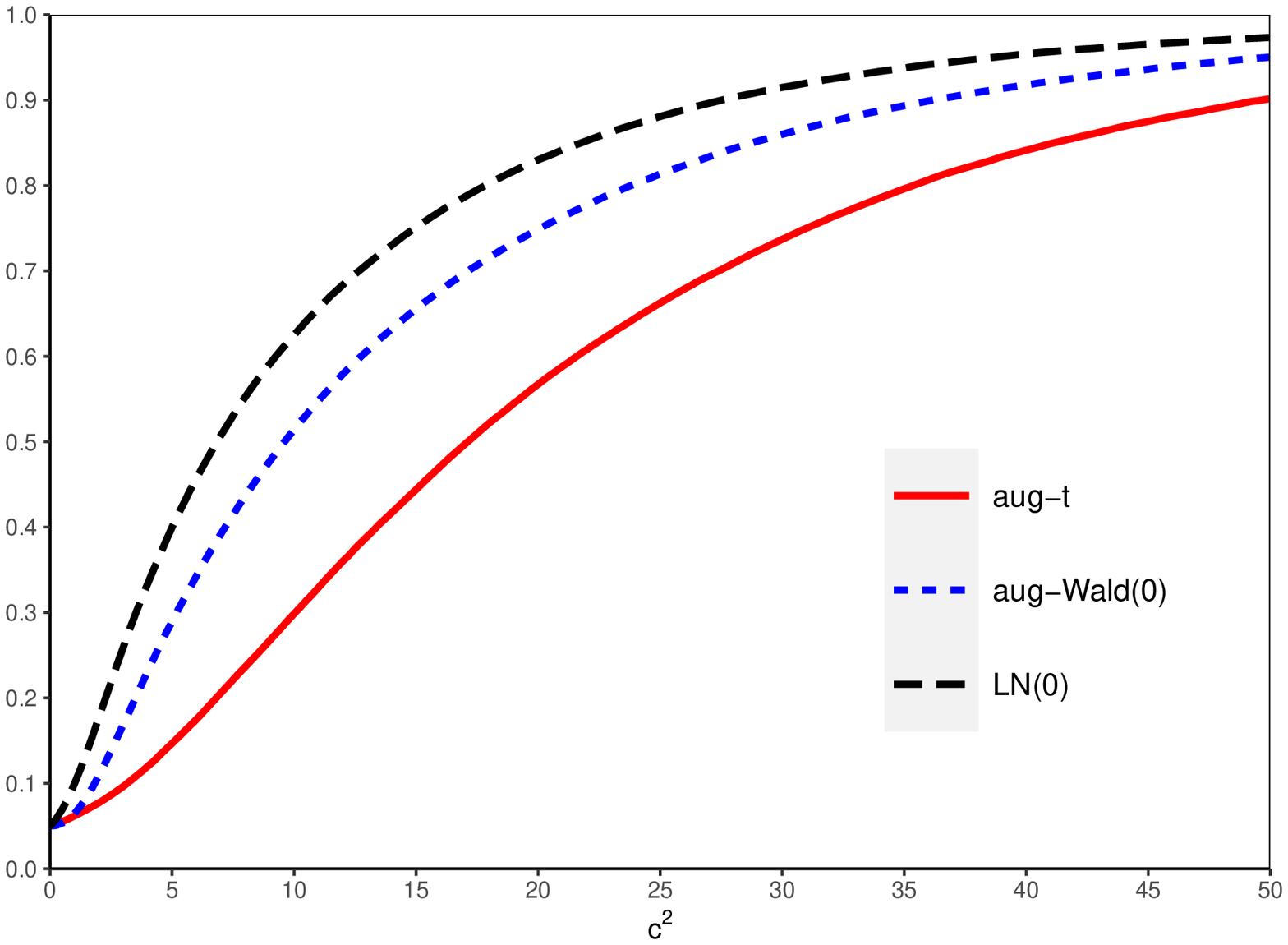}
	\caption{$q=0$}
	\end{subfigure} \quad
	\begin{subfigure}{0.47\textwidth}
	\includegraphics[width=\textwidth]{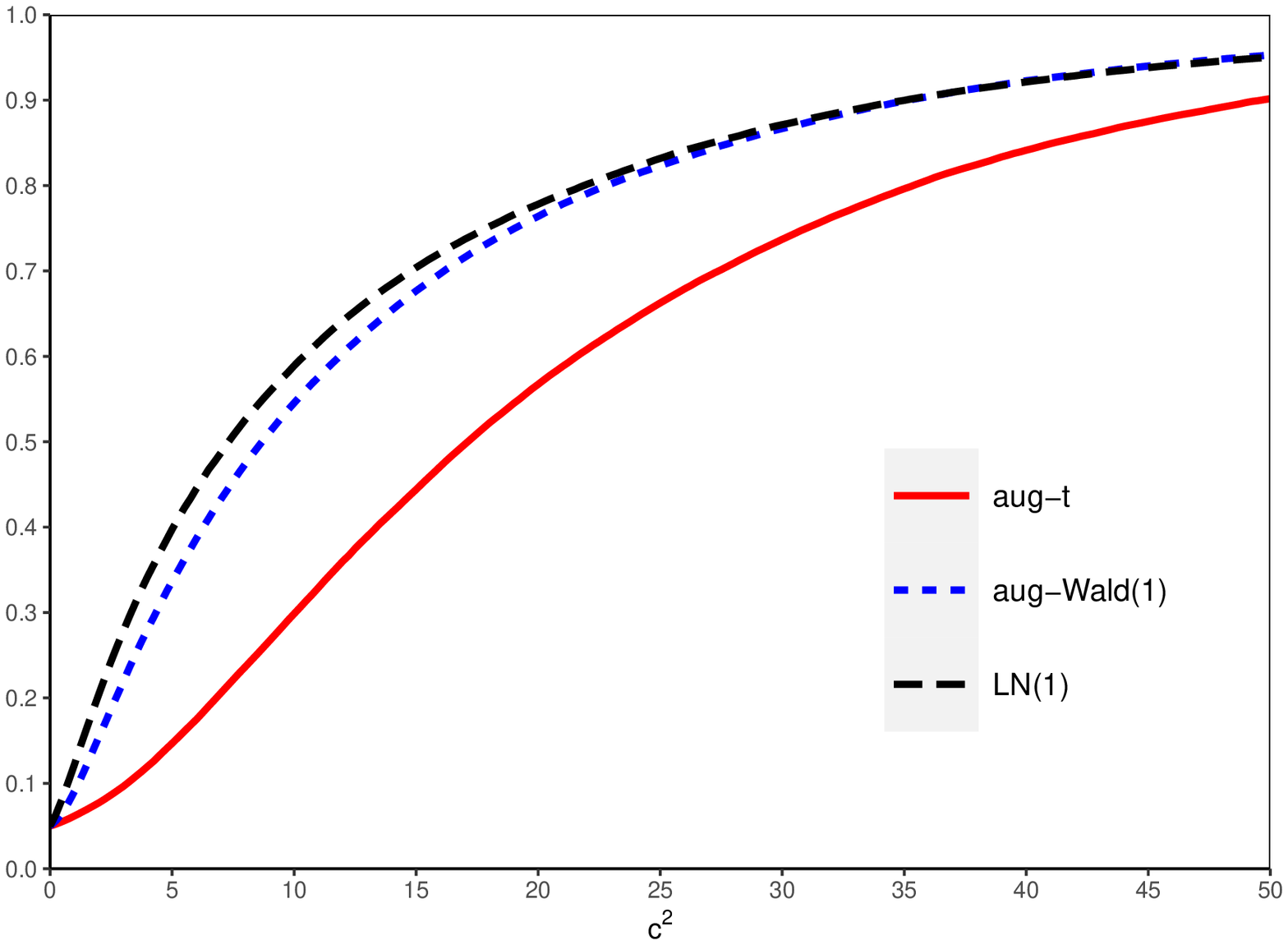}
	\caption{$q=1$}
	\end{subfigure} \quad
	\begin{subfigure}{0.47\textwidth}
	\includegraphics[width=\textwidth]{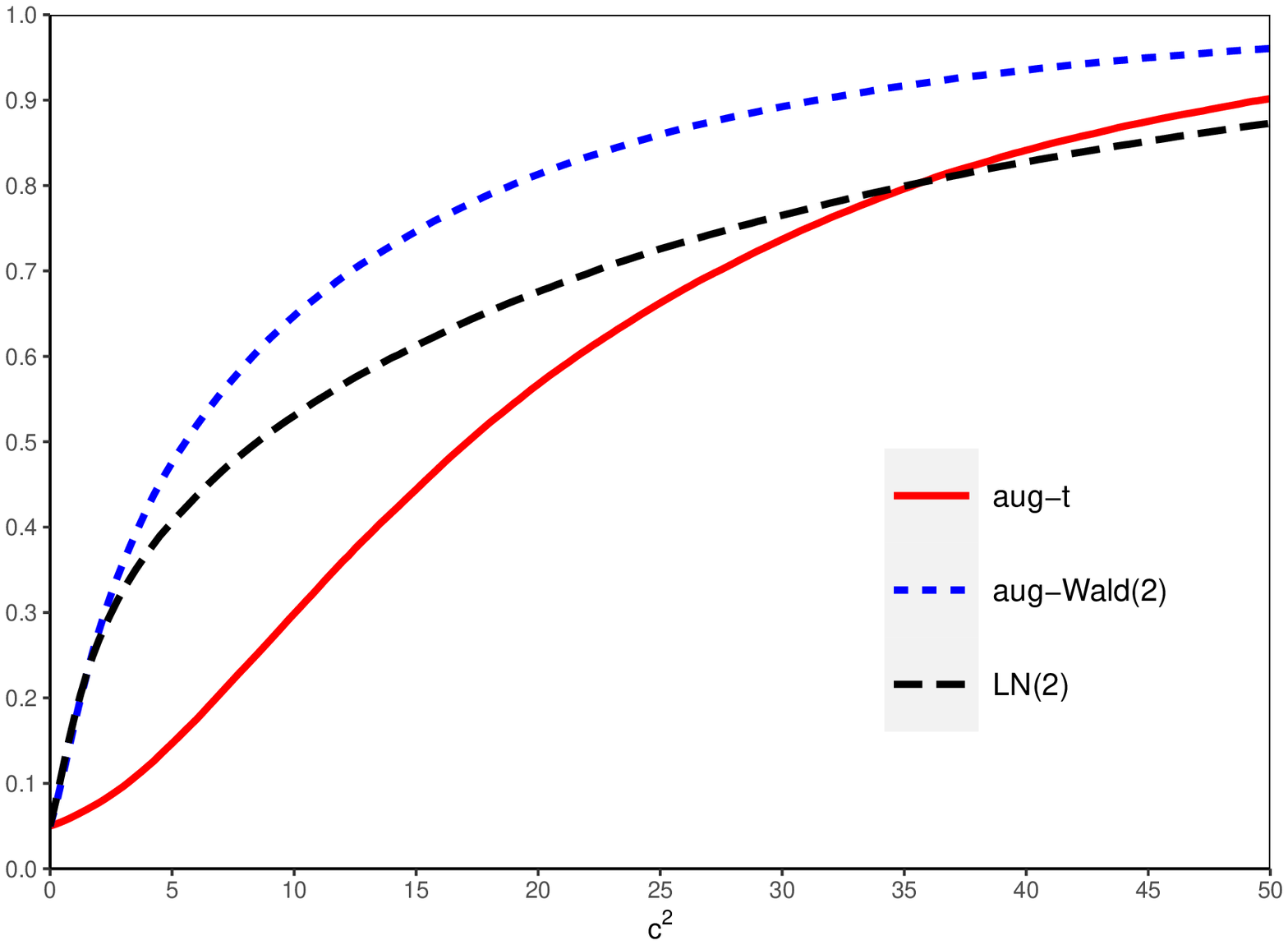}
	\caption{$q=2$}
	\end{subfigure} \quad
	\begin{subfigure}{0.47\textwidth}
	\includegraphics[width=\textwidth]{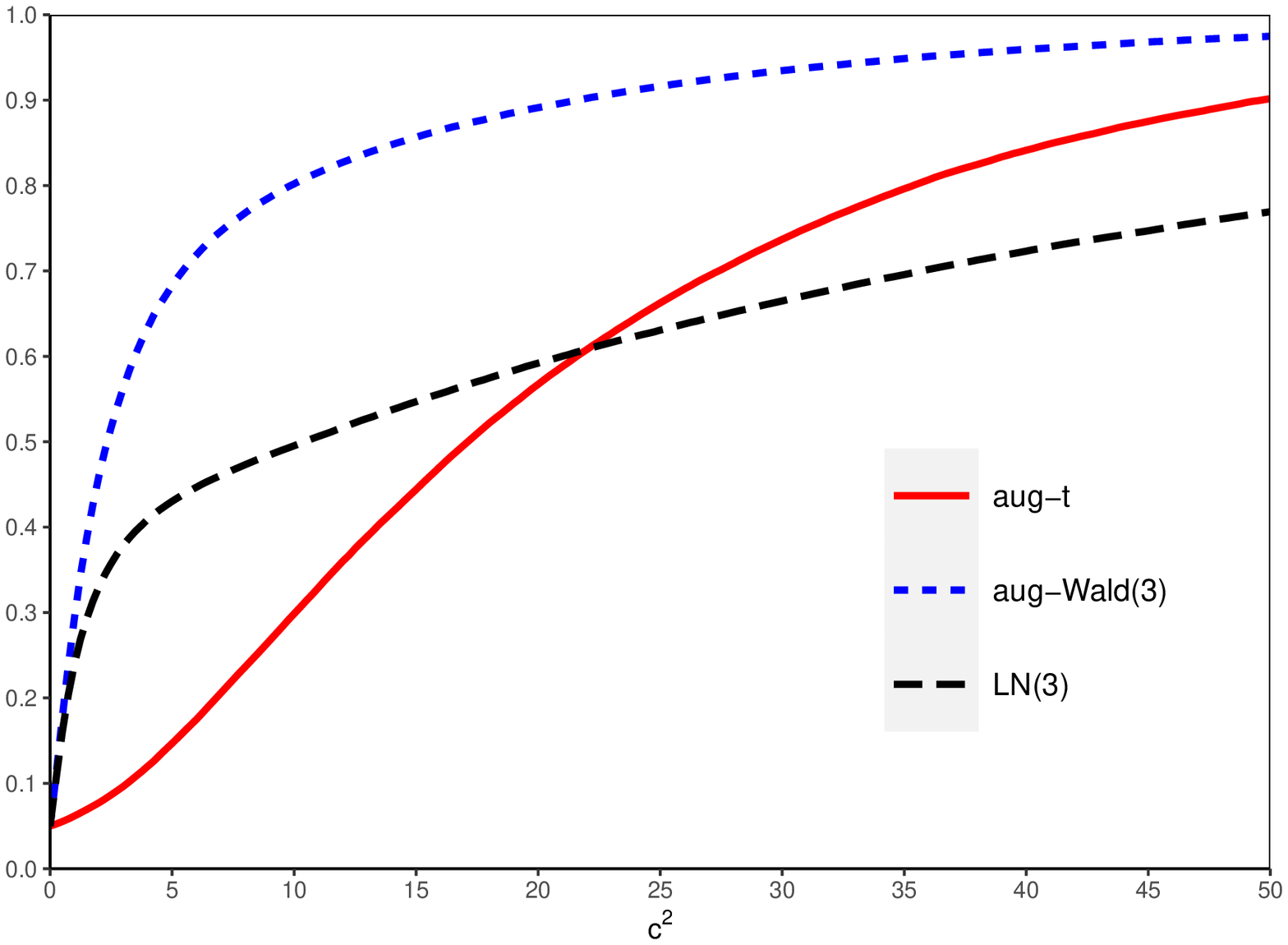}
	\caption{$q=3$}
	\end{subfigure} \quad
\caption{Asymptotic power functions under $a=0$}
\label{fig:power_ln_t_wald}
\end{sidewaysfigure}

\clearpage
\begin{sidewaysfigure}
	\centering
	\begin{subfigure}{0.47\textwidth}
		\includegraphics[width=\textwidth]{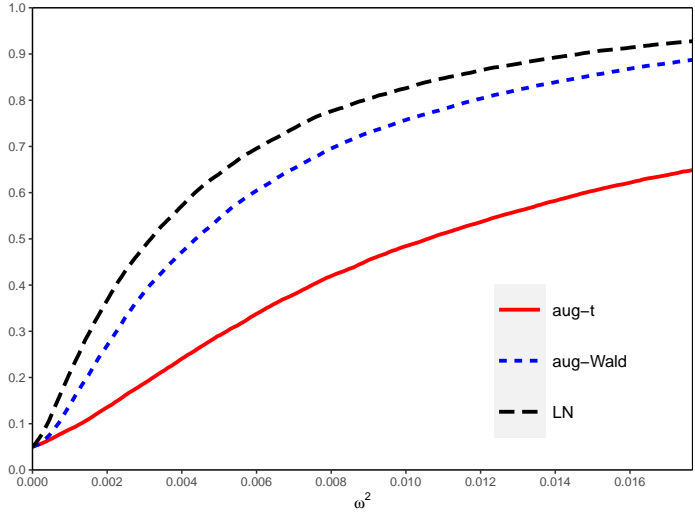}
		\caption{$\mathrm{Corr}(\varepsilon_t,v_t)=0$}
	\end{subfigure} \quad
	\begin{subfigure}{0.47\textwidth}
		\includegraphics[width=\textwidth]{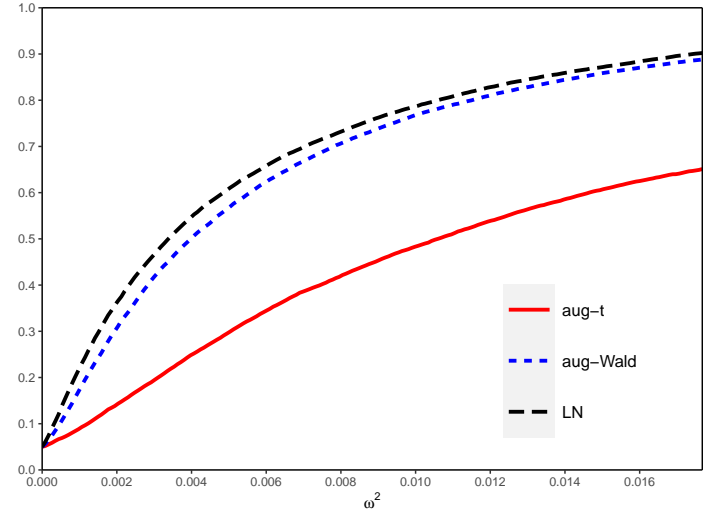}
		\caption{$\mathrm{Corr}(\varepsilon_t,v_t)=0.25$}
	\end{subfigure} \quad
	\begin{subfigure}{0.47\textwidth}
		\includegraphics[width=\textwidth]{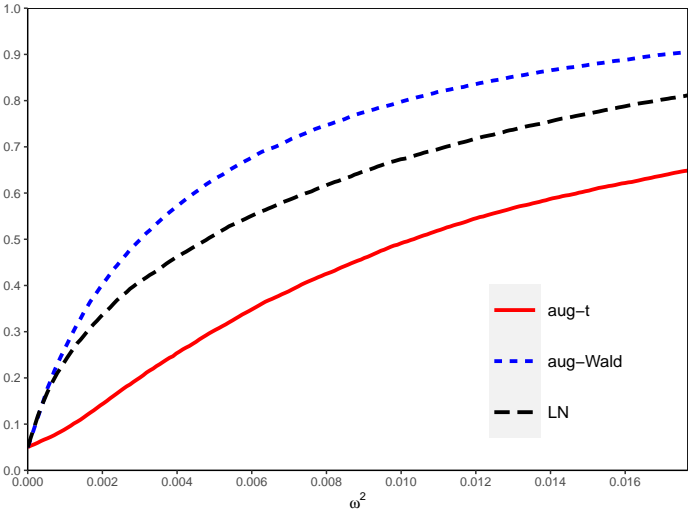}
		\caption{$\mathrm{Corr}(\varepsilon_t,v_t)=0.5$}
	\end{subfigure} \quad
	\begin{subfigure}{0.47\textwidth}
		\includegraphics[width=\textwidth]{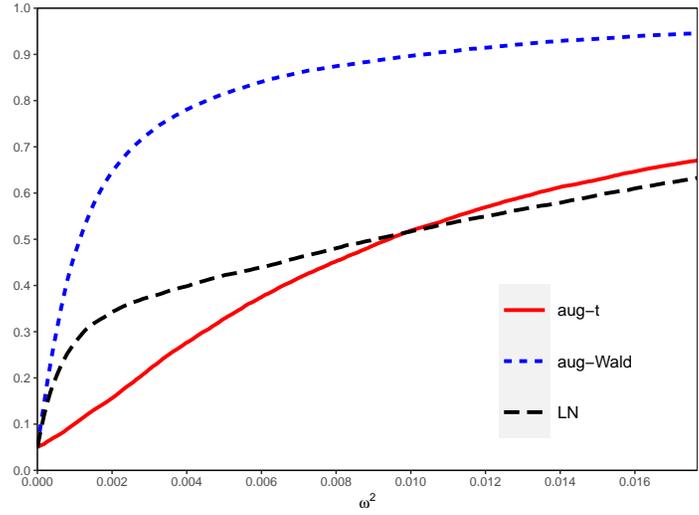}
		\caption{$\mathrm{Corr}(\varepsilon_t,v_t)=0.75$}
	\end{subfigure} \quad
	\caption{Size-adjusted power functions with $T=200$ and $\rho_T=1$}
	\label{fig:power_ln_t_wald_200}
\end{sidewaysfigure}

\clearpage
\begin{sidewaysfigure}
	\centering
	\begin{subfigure}{0.47\textwidth}
		\includegraphics[width=\textwidth]{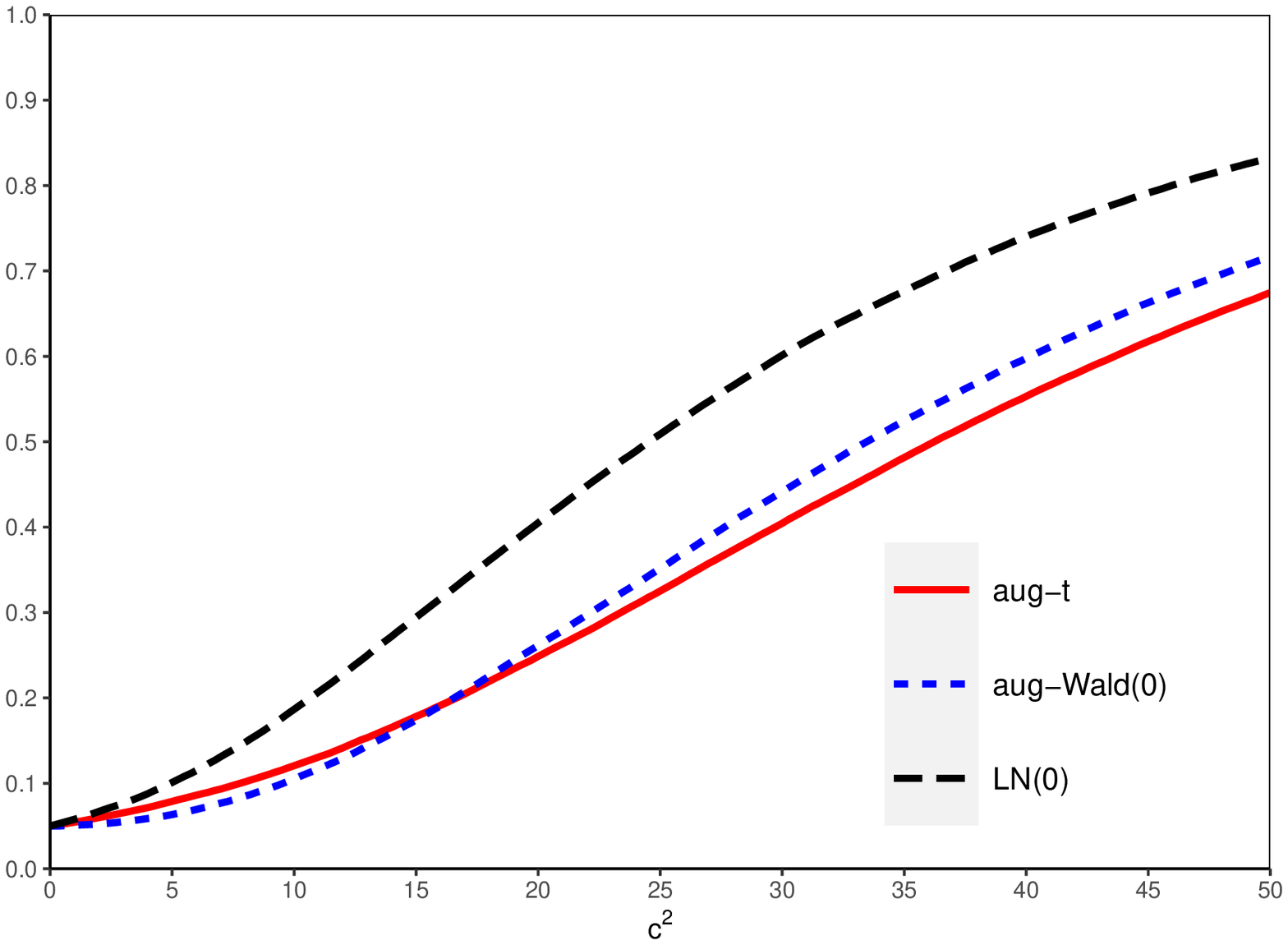}
		\caption{$q=0$}
	\end{subfigure} \quad
	\begin{subfigure}{0.47\textwidth}
		\includegraphics[width=\textwidth]{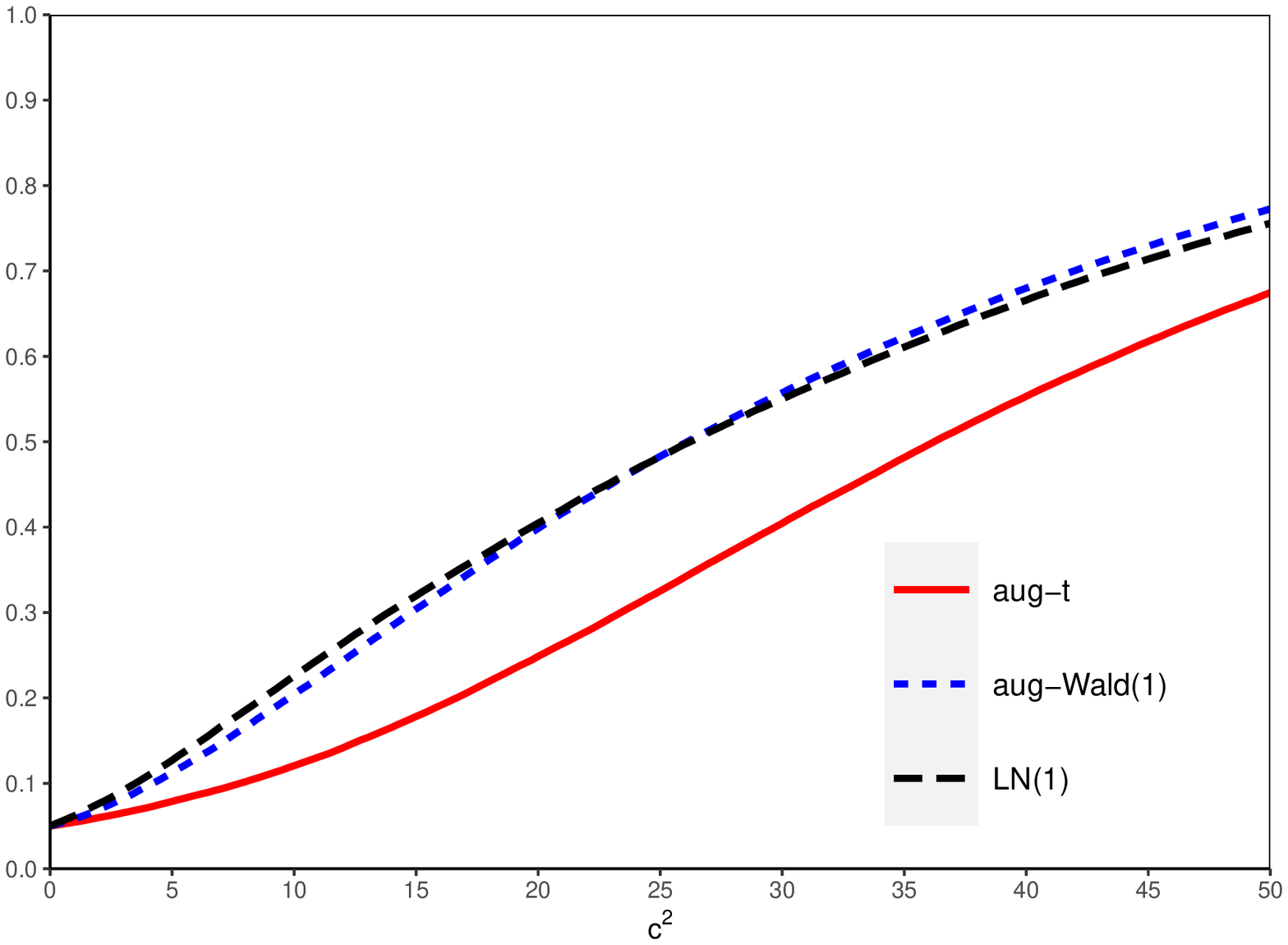}
		\caption{$q=1$}
	\end{subfigure} \quad
	\begin{subfigure}{0.47\textwidth}
		\includegraphics[width=\textwidth]{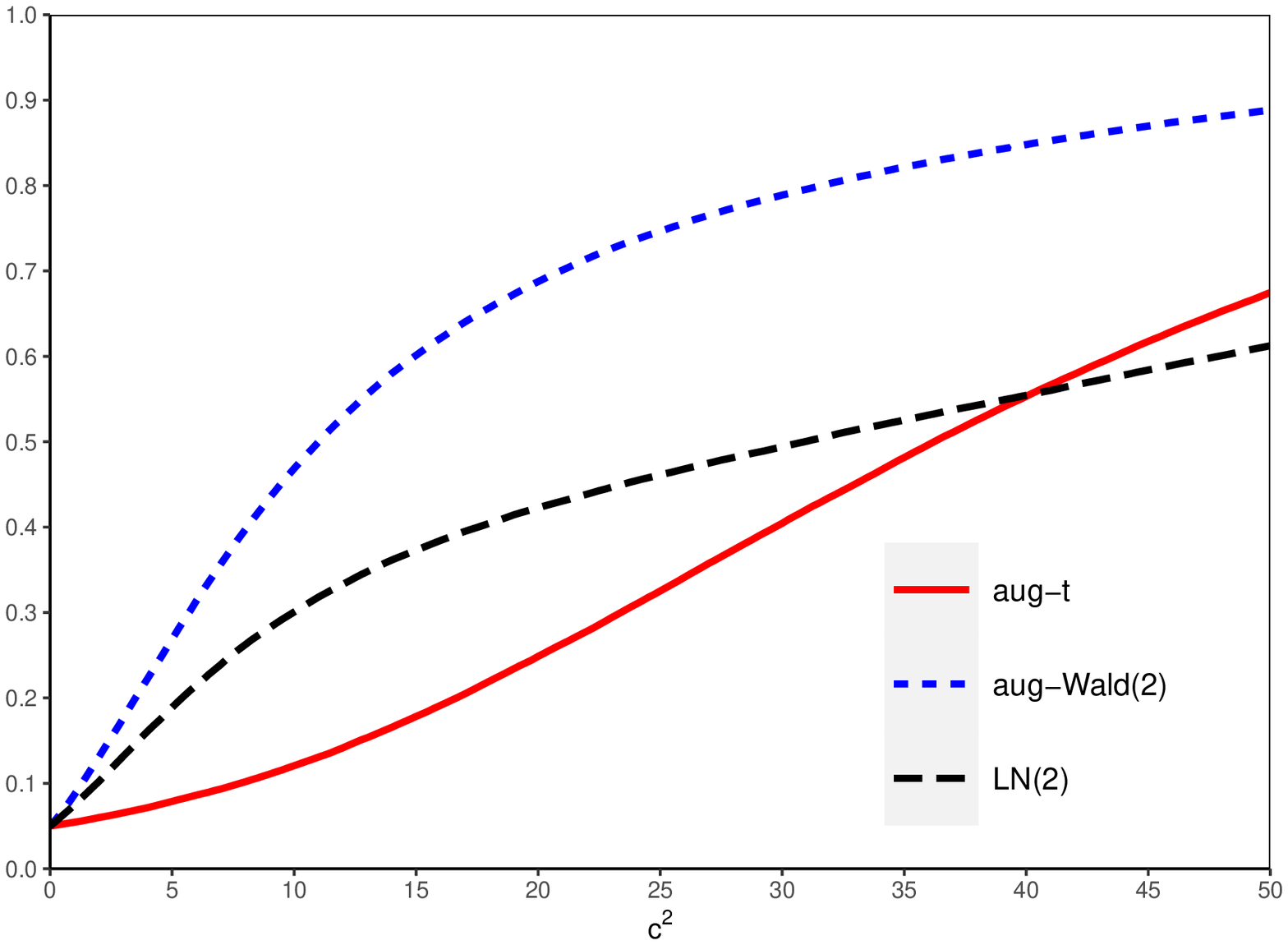}
		\caption{$q=2$}
	\end{subfigure} \quad
	\begin{subfigure}{0.47\textwidth}
		\includegraphics[width=\textwidth]{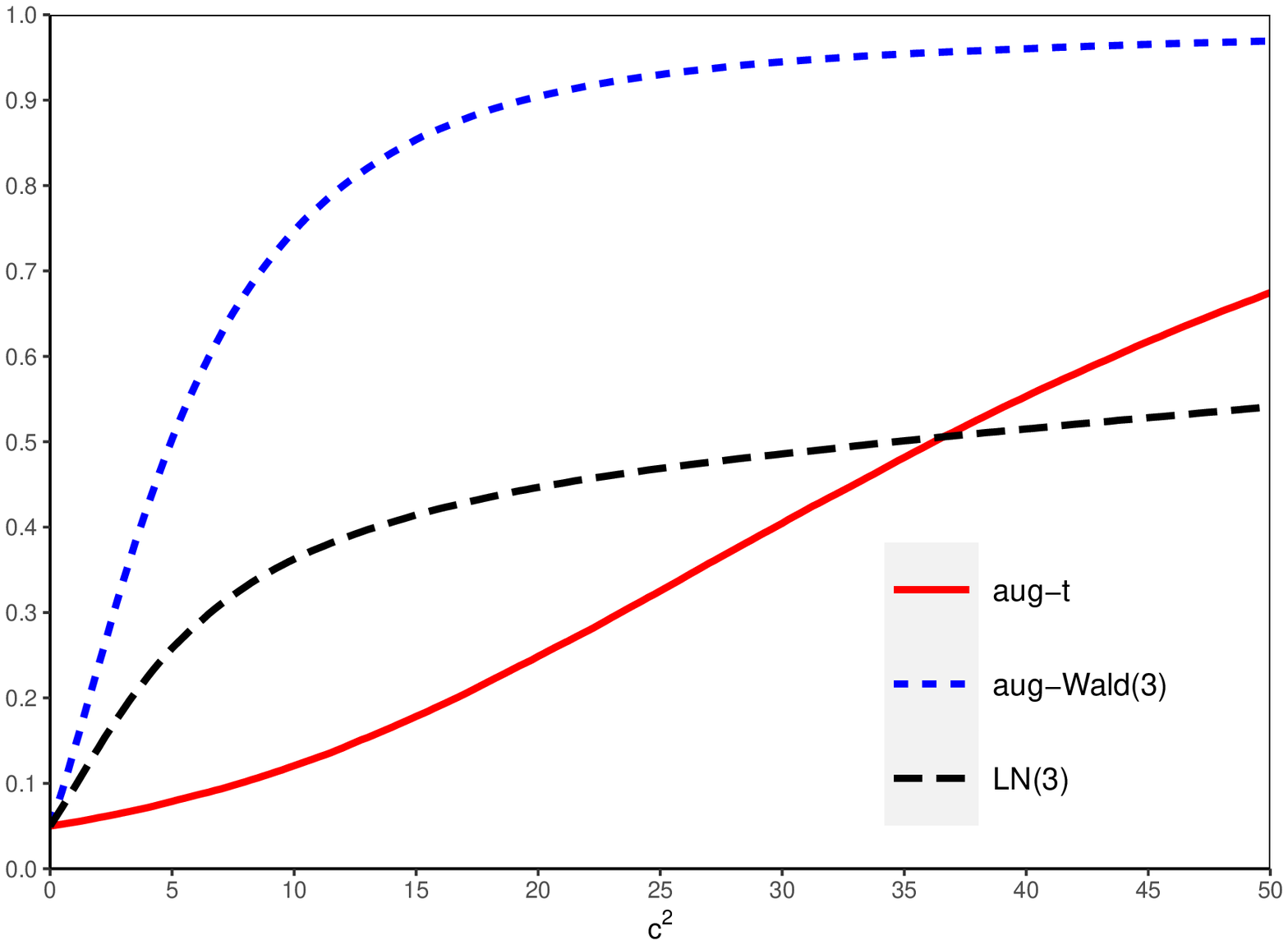}
		\caption{$q=3$}
	\end{subfigure} \quad
	\caption{Asymptotic power functions under $a=-5$}
	\label{fig:power_ln_t_wald_general_5}
\end{sidewaysfigure}

\clearpage
\begin{sidewaysfigure}
	\centering
	\begin{subfigure}{0.47\textwidth}
		\includegraphics[width=\textwidth]{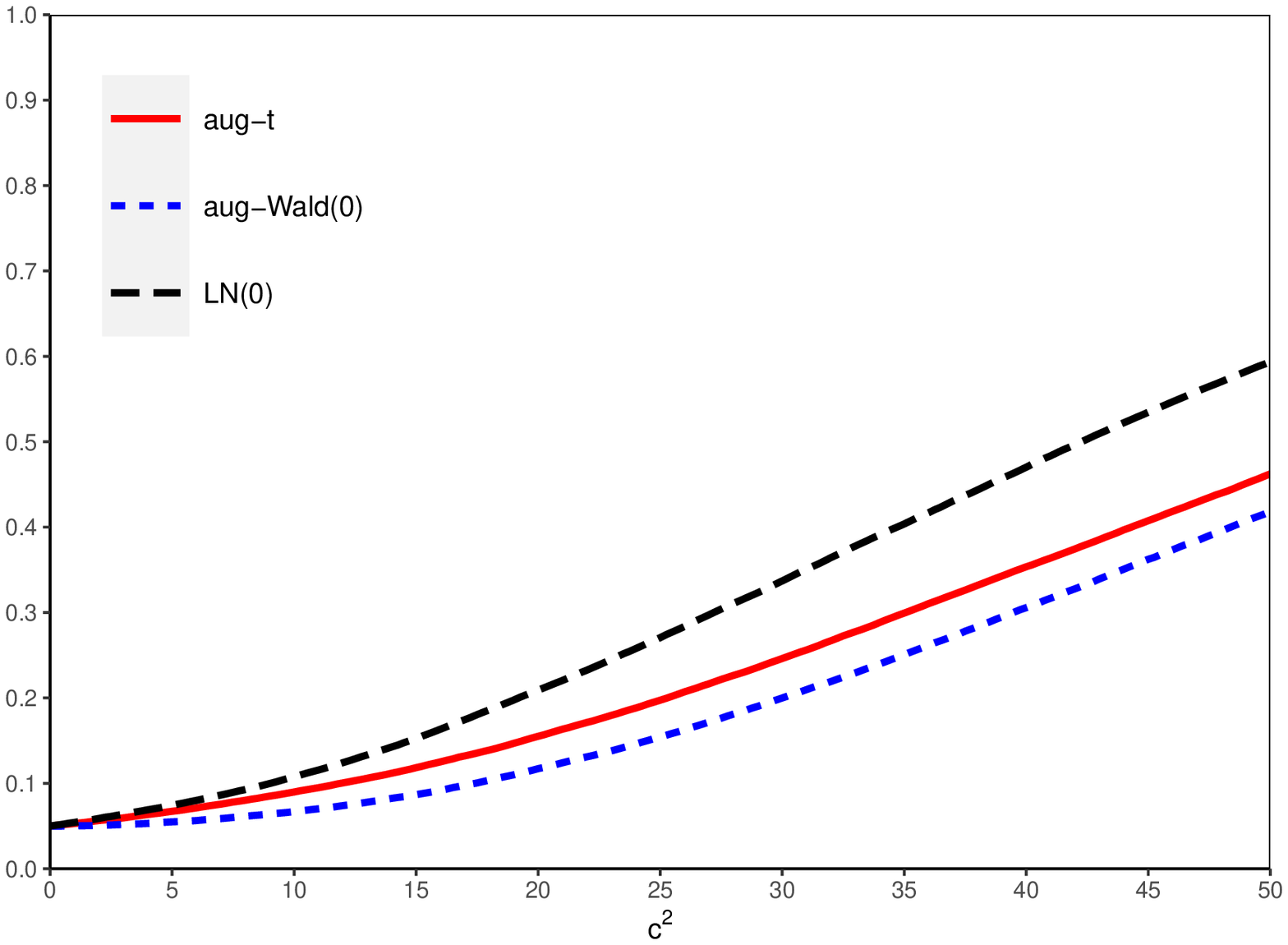}
		\caption{$q=0$}
	\end{subfigure} \quad
	\begin{subfigure}{0.47\textwidth}
		\includegraphics[width=\textwidth]{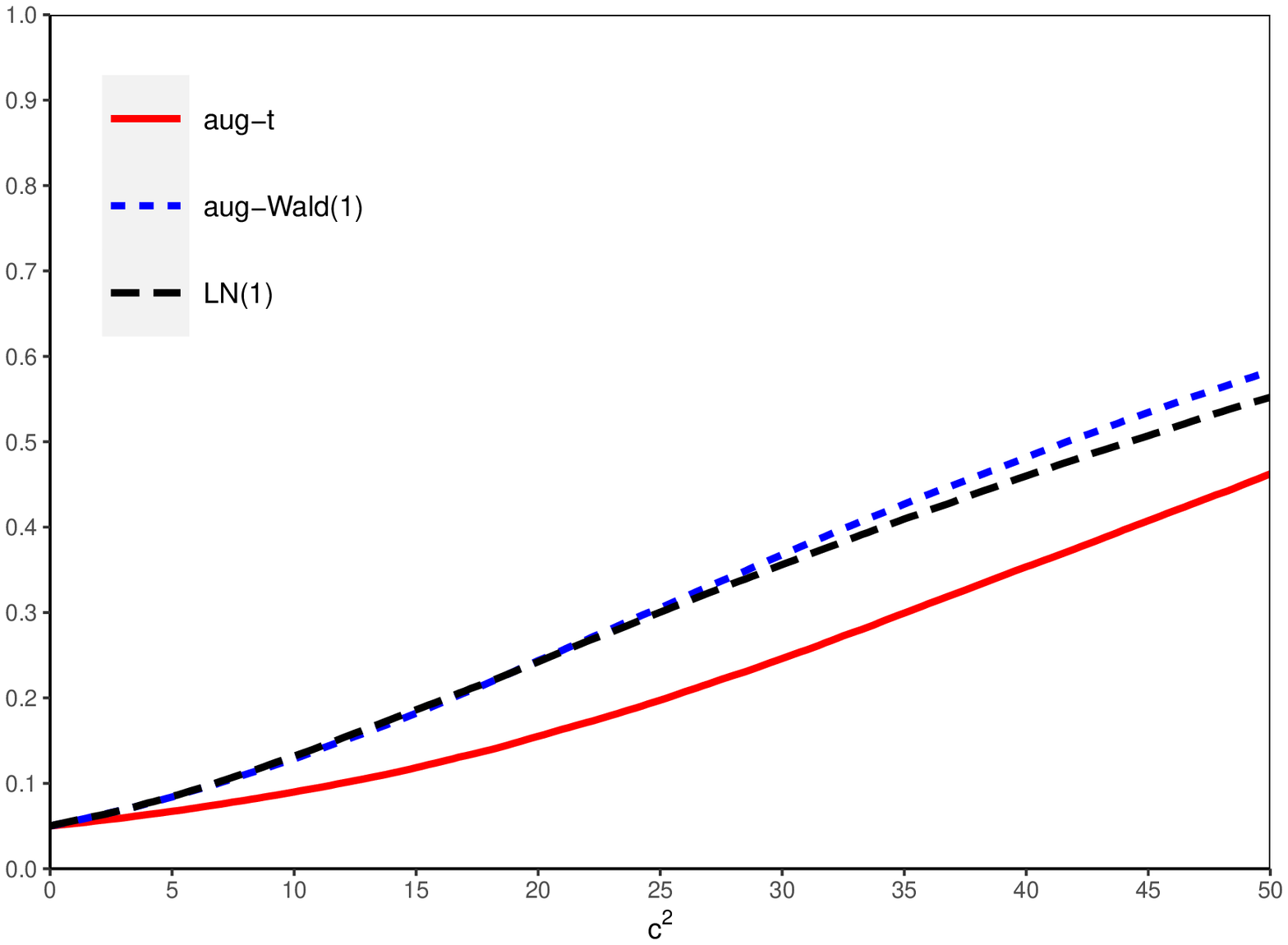}
		\caption{$q=1$}
	\end{subfigure} \quad
	\begin{subfigure}{0.47\textwidth}
		\includegraphics[width=\textwidth]{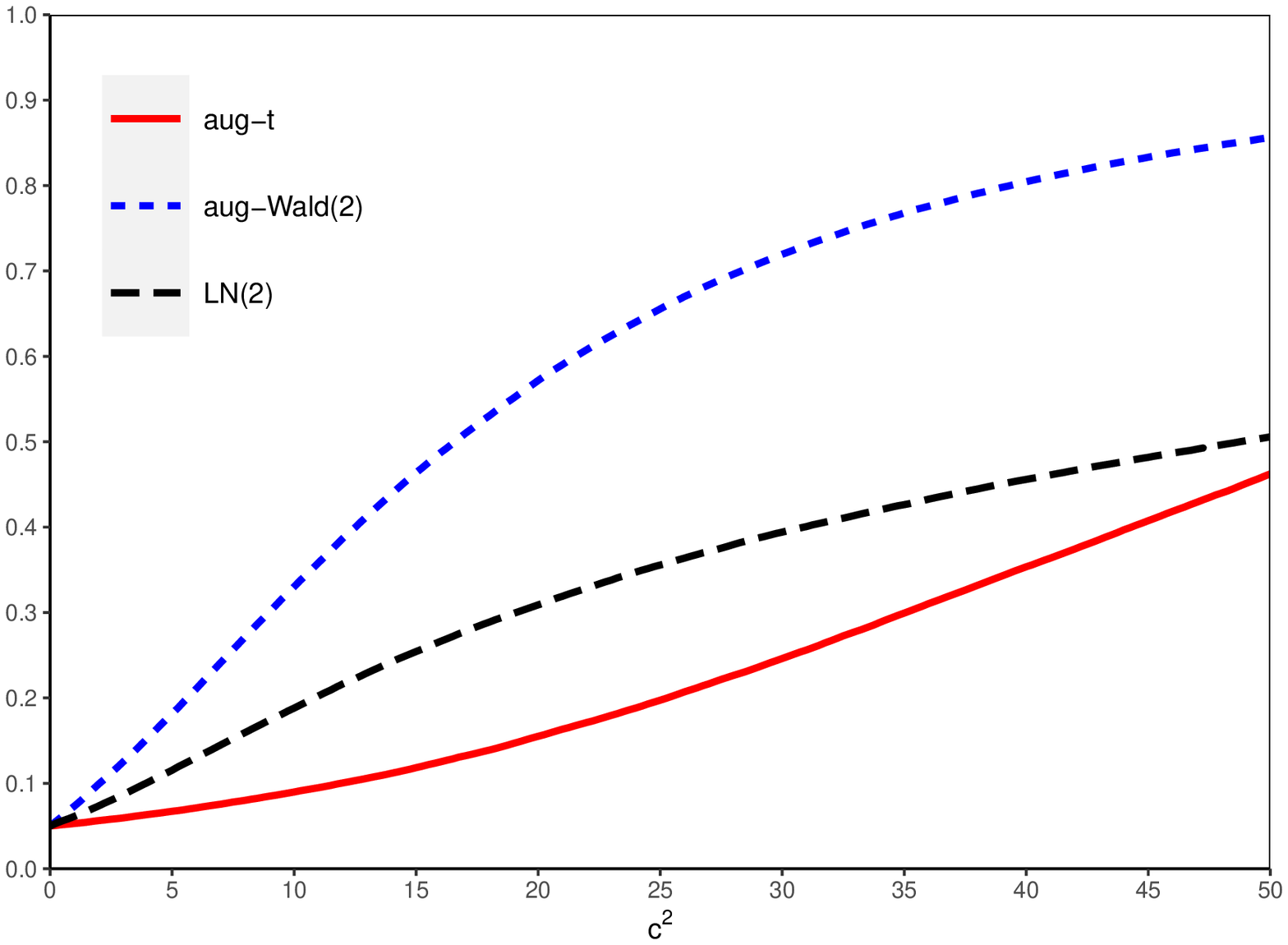}
		\caption{$q=2$}
	\end{subfigure} \quad
	\begin{subfigure}{0.47\textwidth}
		\includegraphics[width=\textwidth]{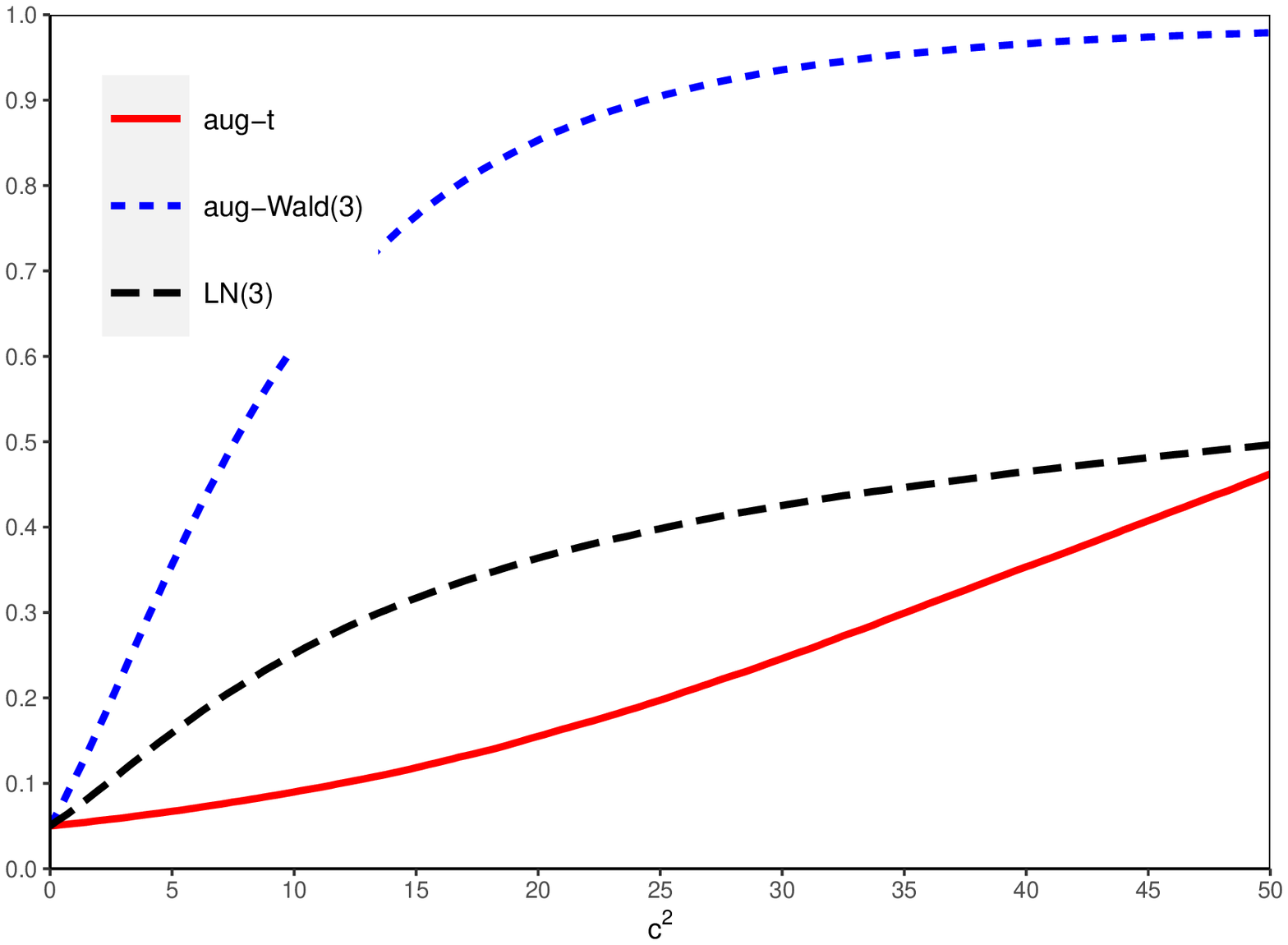}
		\caption{$q=3$}
	\end{subfigure} \quad
	\caption{Asymptotic power functions under $a=-10$}
	\label{fig:power_ln_t_wald_general_10}
\end{sidewaysfigure}

\clearpage
\begin{sidewaysfigure}
	\centering
	\begin{subfigure}{0.47\textwidth}
		\includegraphics[width=\textwidth]{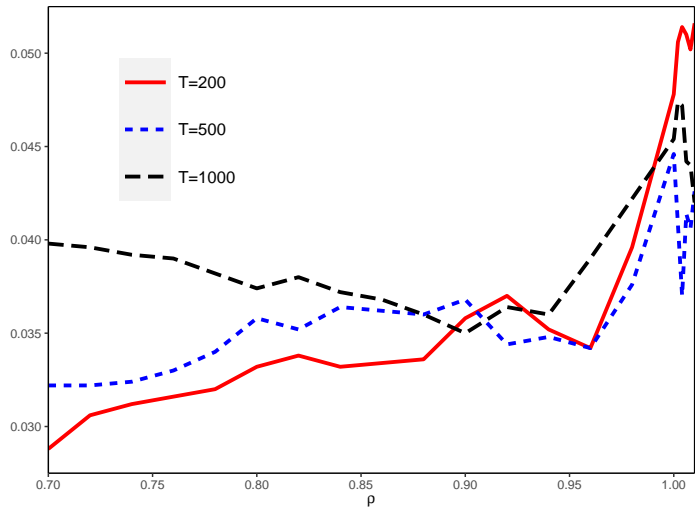}
		\caption{$\varepsilon_t \sim \mathrm{i.i.d} \ N(0,1)$}
	\end{subfigure} \quad
	\begin{subfigure}{0.47\textwidth}
		\includegraphics[width=\textwidth]{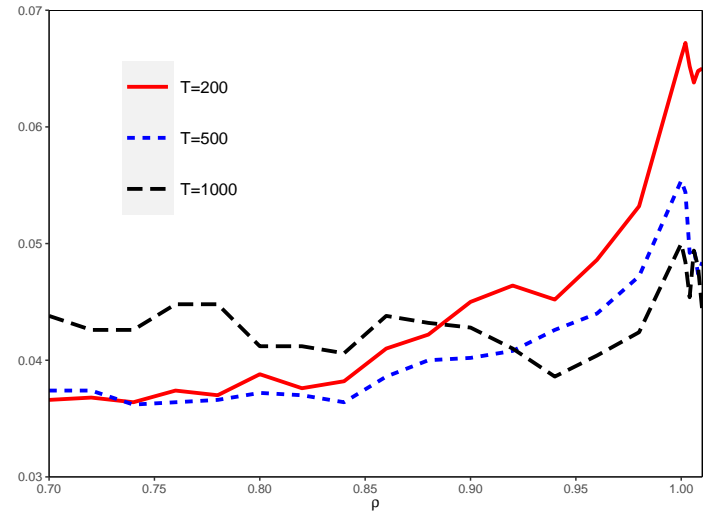}
		\caption{$\varepsilon_t \sim \mathrm{i.i.d} \ (\chi^2(10)-10)/\sqrt(20)$}
	\end{subfigure} \quad
	\begin{subfigure}{0.47\textwidth}
		\includegraphics[width=\textwidth]{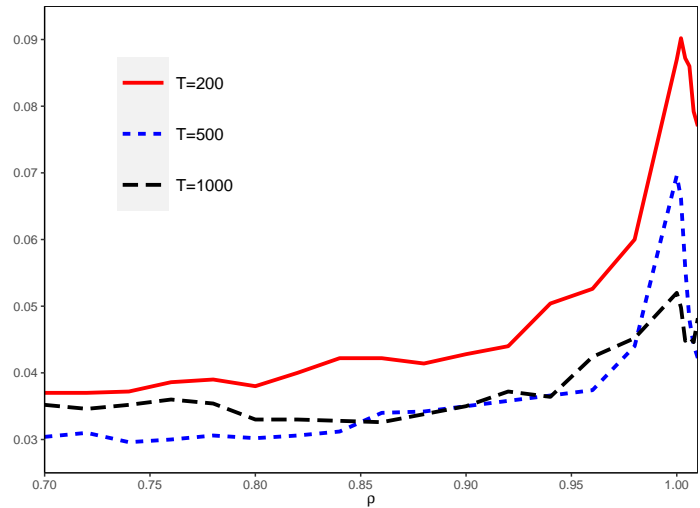}
		\caption{$\varepsilon_t \sim \mathrm{i.i.d} \ (\chi^2(1)-1)/\sqrt(2)$}
	\end{subfigure} \quad
	\caption{Finite-sample rejection rates of the Bonferroni-Wald test under the null with significance level 0.05}
	\label{fig:sizes_fs}
\end{sidewaysfigure}

\clearpage
\begin{sidewaysfigure}
	\centering
	\begin{subfigure}{0.47\textwidth}
		\includegraphics[width=\textwidth]{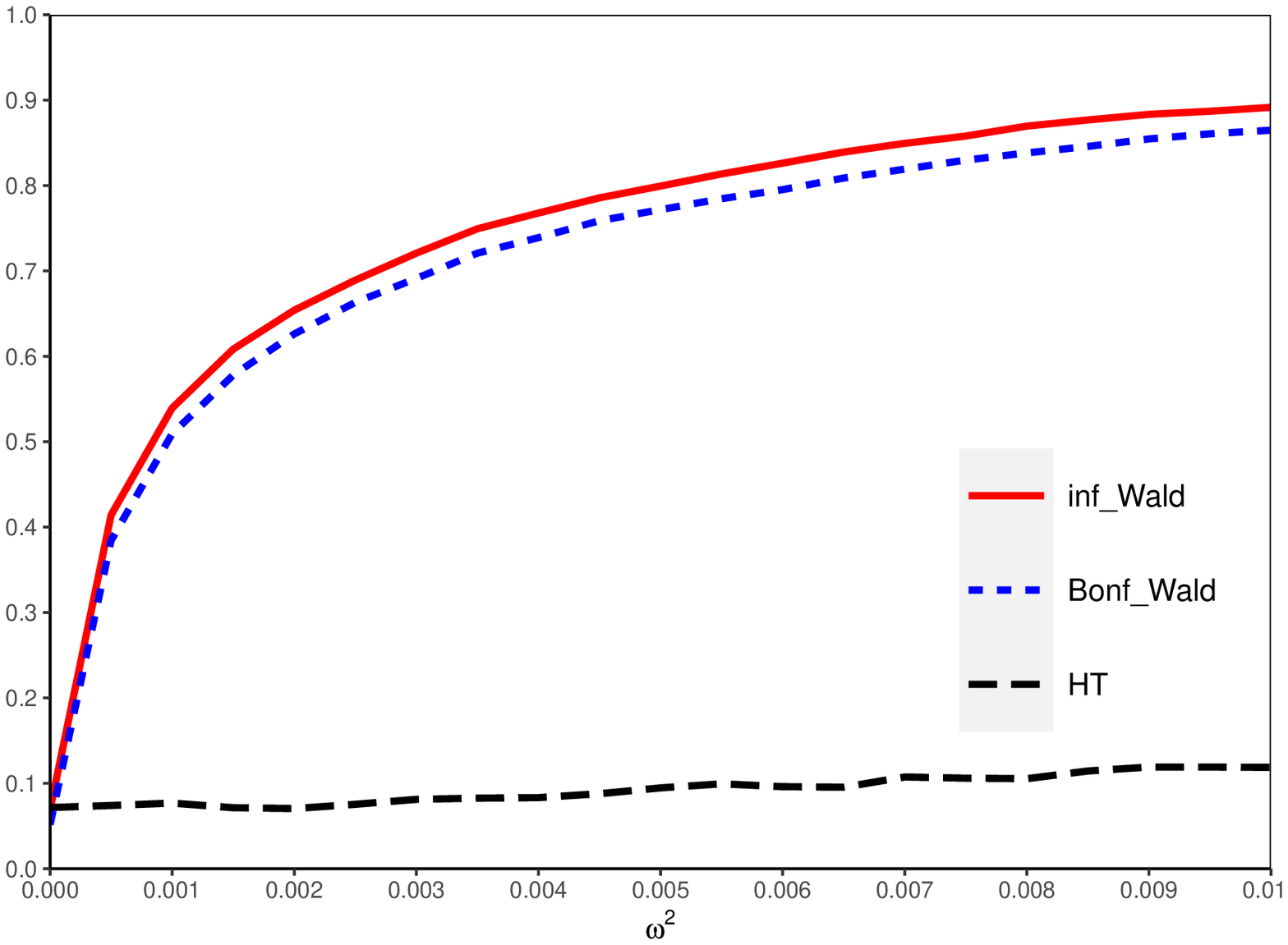}
		\caption{$\mathrm{Corr}(\varepsilon_t,v_t)=0$}
	\end{subfigure} \quad
	\begin{subfigure}{0.47\textwidth}
		\includegraphics[width=\textwidth]{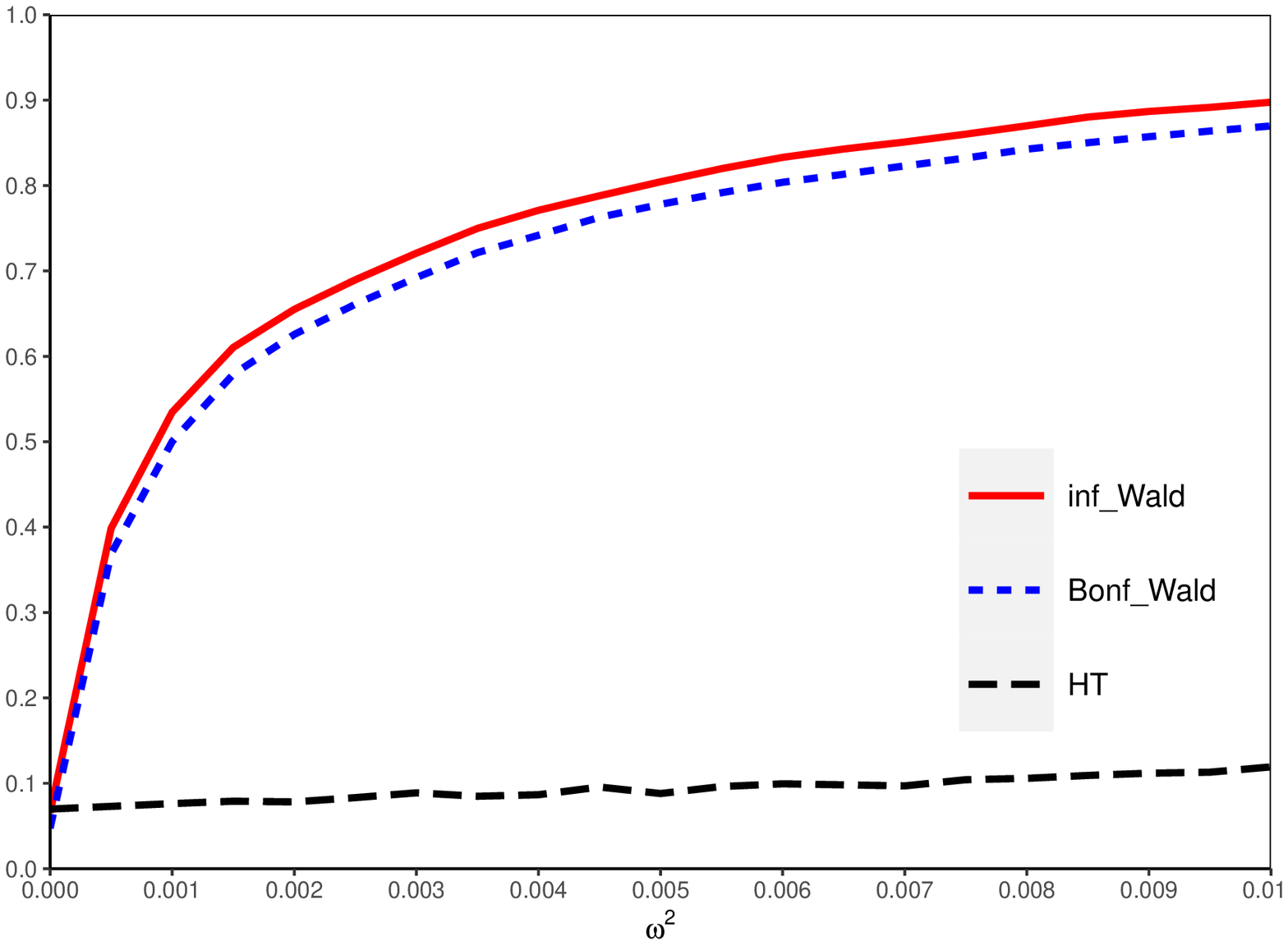}
		\caption{$\mathrm{Corr}(\varepsilon_t,v_t)=0.25$}
	\end{subfigure} \quad
	\begin{subfigure}{0.47\textwidth}
		\includegraphics[width=\textwidth]{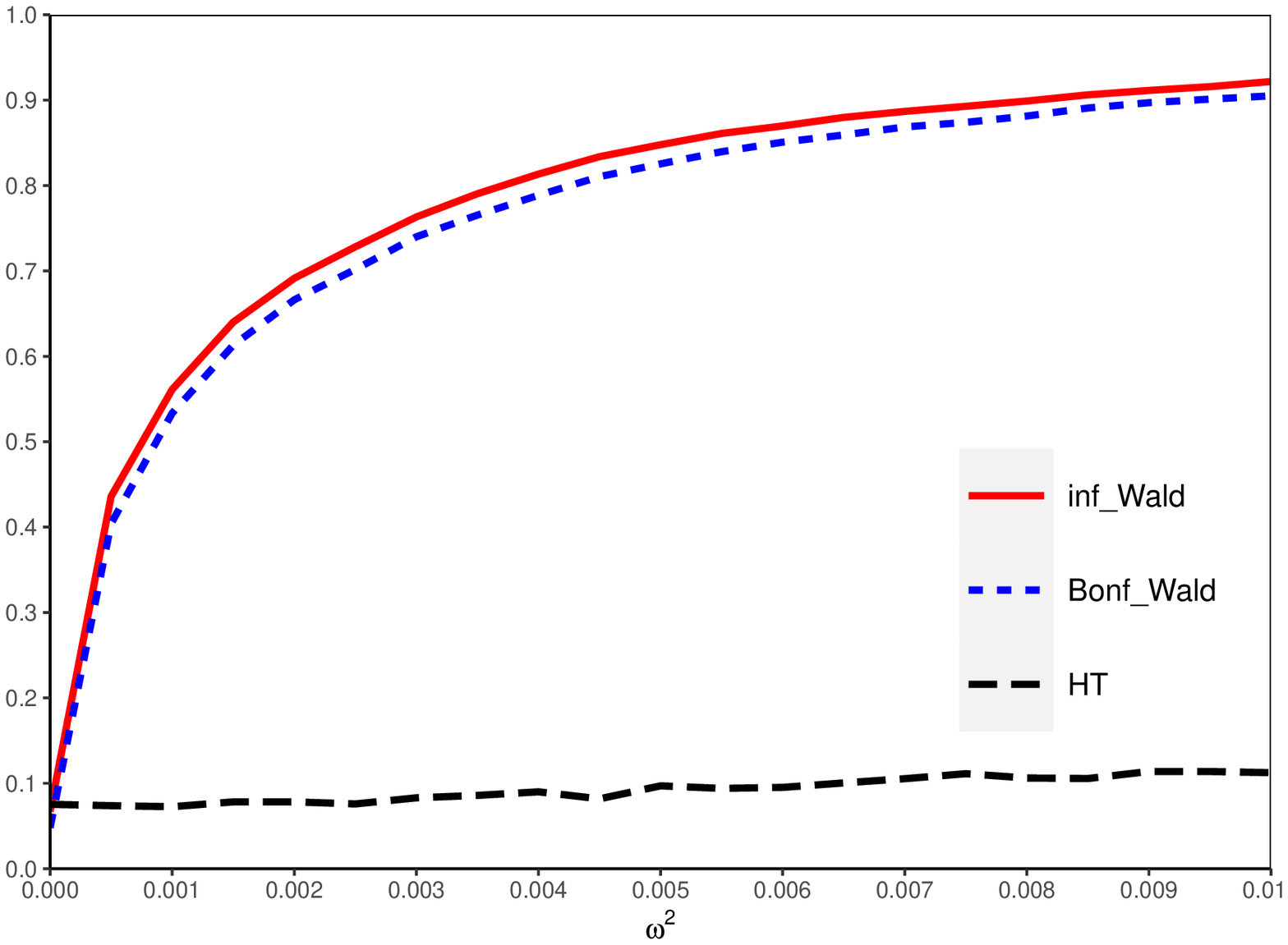}
		\caption{$\mathrm{Corr}(\varepsilon_t,v_t)=0.5$}
	\end{subfigure} \quad
	\begin{subfigure}{0.47\textwidth}
		\includegraphics[width=\textwidth]{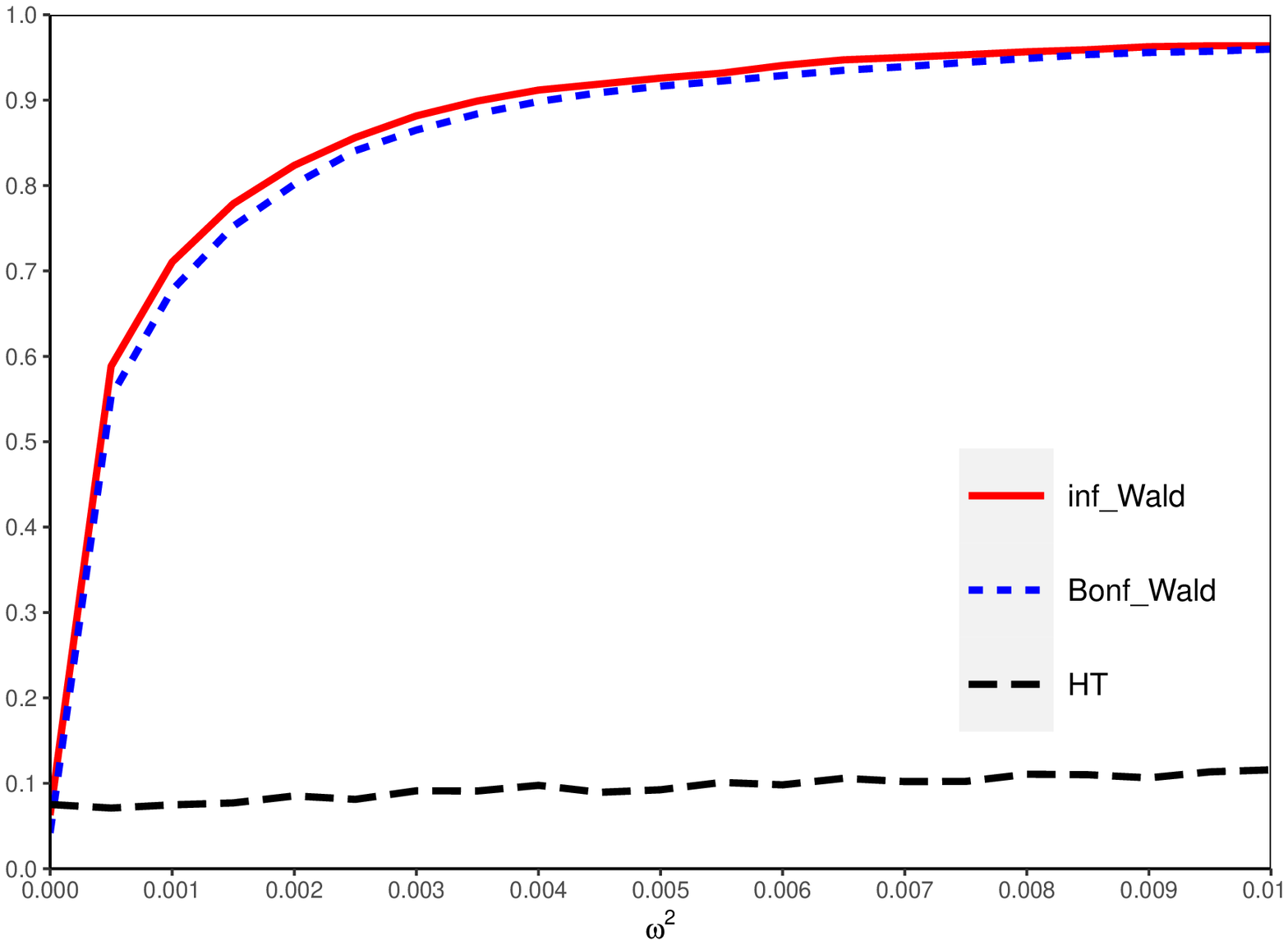}
		\caption{$\mathrm{Corr}(\varepsilon_t,v_t)=0.75$}
	\end{subfigure} \quad
	\caption{Finite-sample power functions under $\rho=1.01$}
	\label{fig:powers_fs_101}
\end{sidewaysfigure}

\clearpage
\begin{sidewaysfigure}
	\centering
	\begin{subfigure}{0.47\textwidth}
		\includegraphics[width=\textwidth]{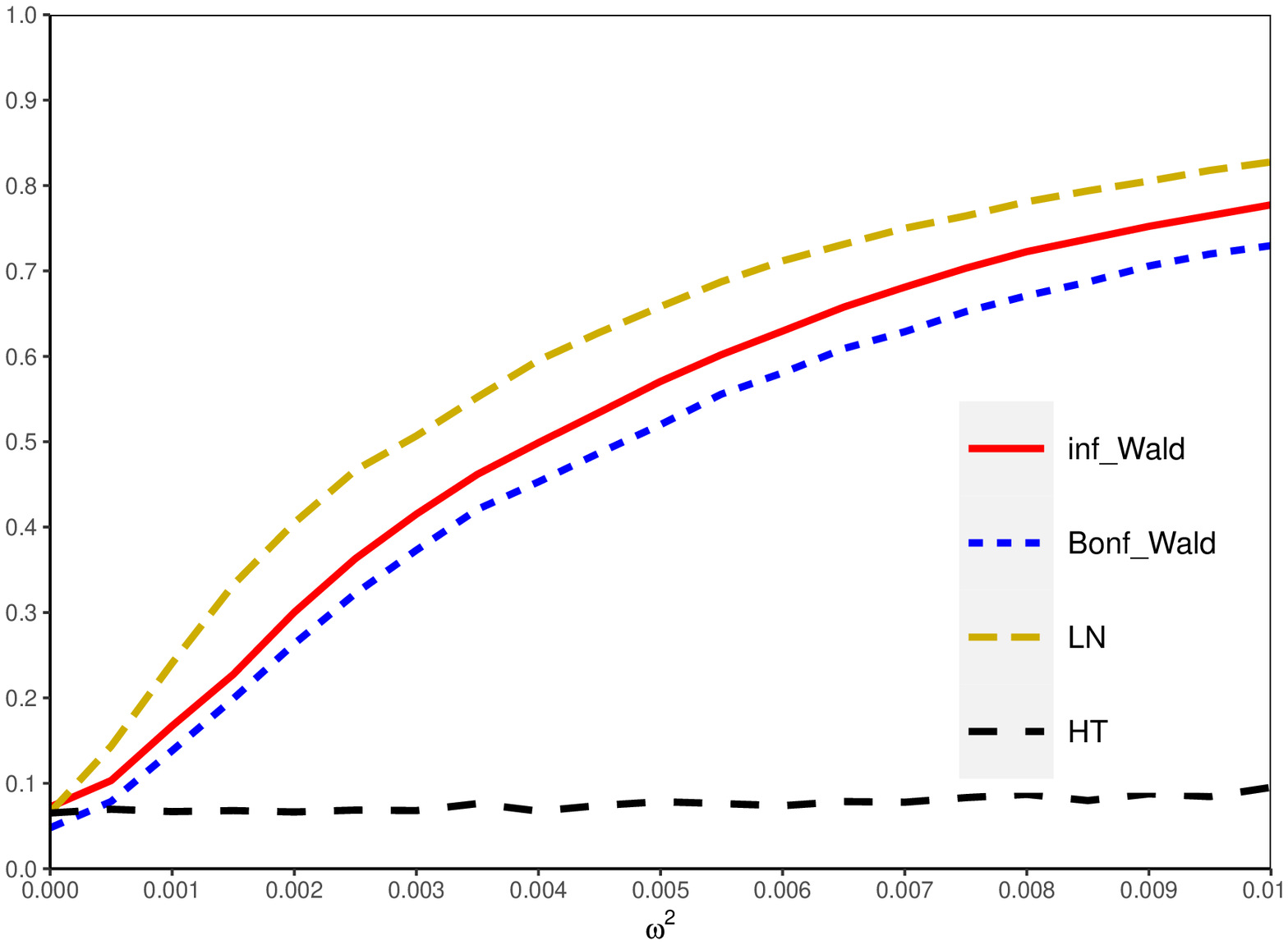}
		\caption{$\mathrm{Corr}(\varepsilon_t,v_t)=0$}
	\end{subfigure} \quad
	\begin{subfigure}{0.47\textwidth}
		\includegraphics[width=\textwidth]{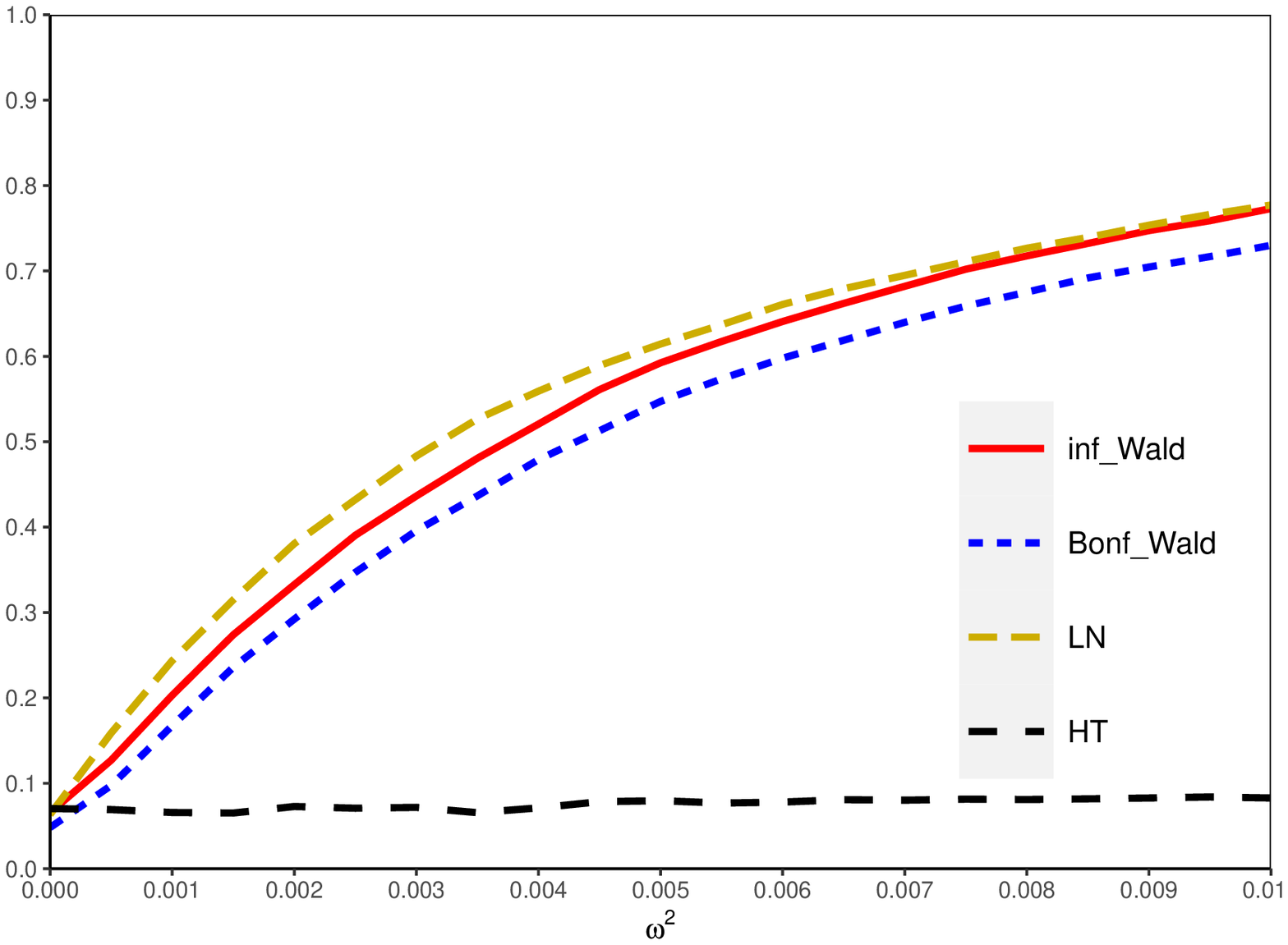}
		\caption{$\mathrm{Corr}(\varepsilon_t,v_t)=0.25$}
	\end{subfigure} \quad
	\begin{subfigure}{0.47\textwidth}
		\includegraphics[width=\textwidth]{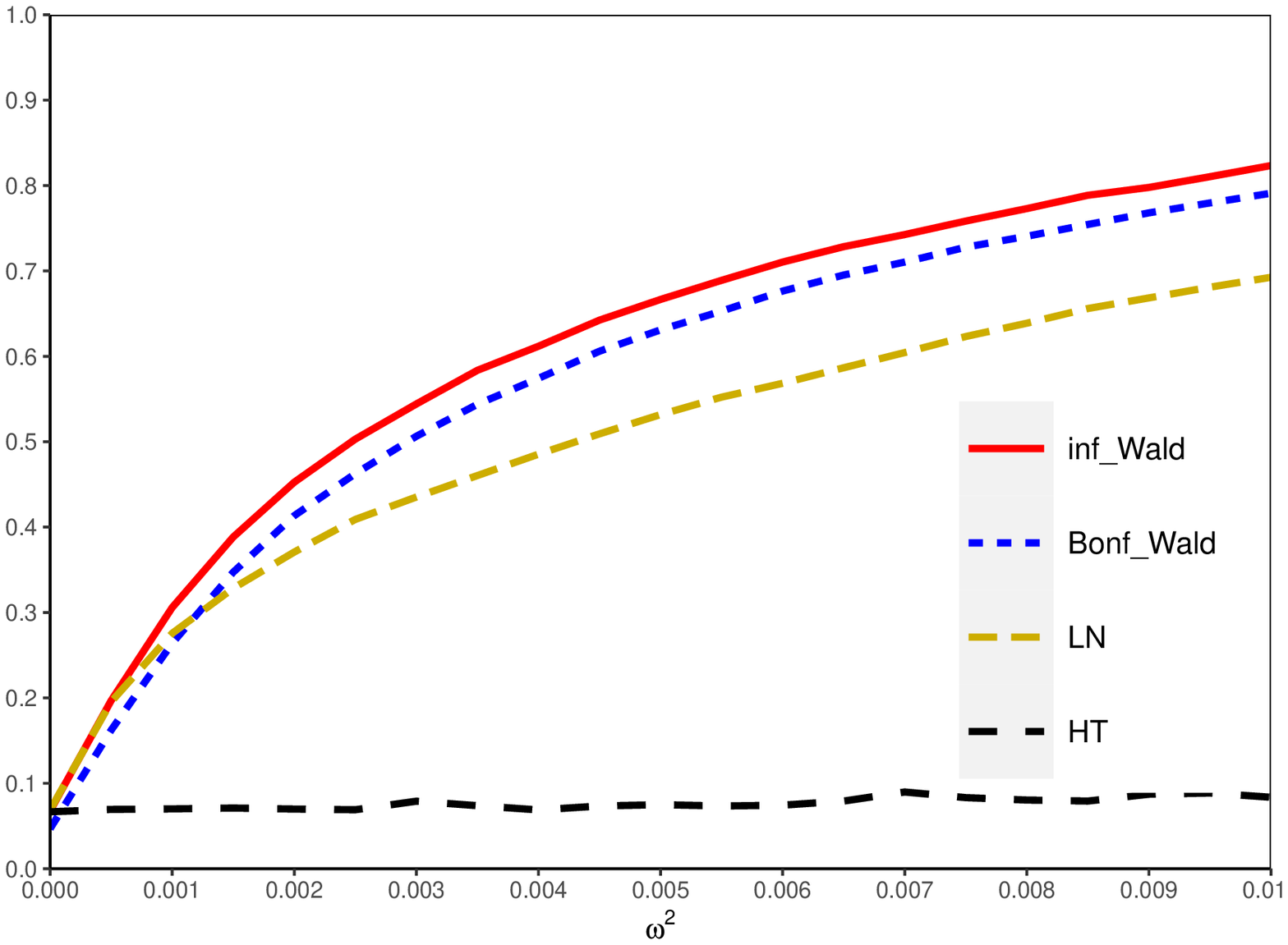}
		\caption{$\mathrm{Corr}(\varepsilon_t,v_t)=0.5$}
	\end{subfigure} \quad
	\begin{subfigure}{0.47\textwidth}
		\includegraphics[width=\textwidth]{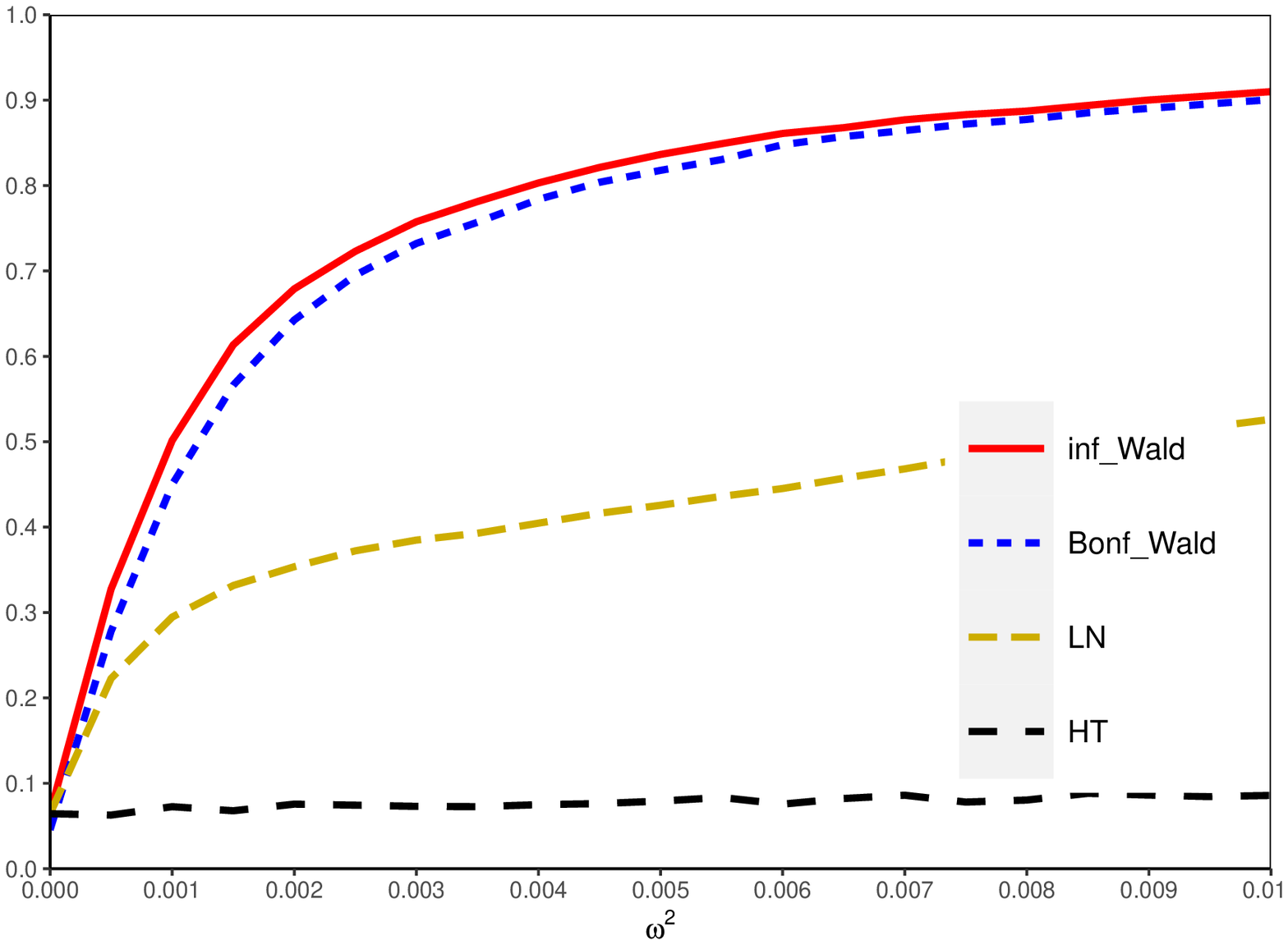}
		\caption{$\mathrm{Corr}(\varepsilon_t,v_t)=0.75$}
	\end{subfigure} \quad
	\caption{Finite-sample power functions under $\rho=1$}
	\label{fig:powers_fs_1}
\end{sidewaysfigure}

\clearpage
\begin{sidewaysfigure}
	\centering
	\begin{subfigure}{0.47\textwidth}
		\includegraphics[width=\textwidth]{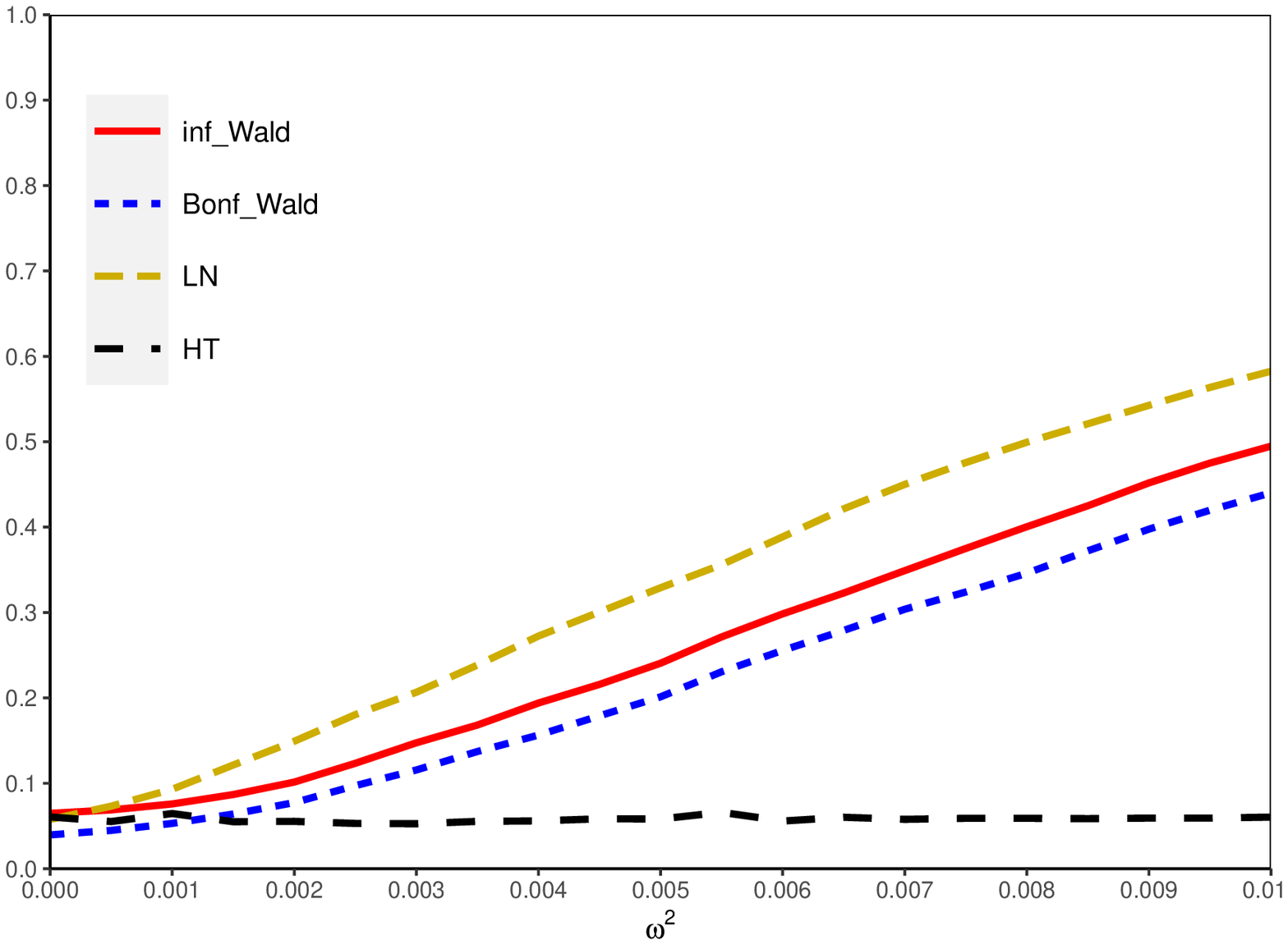}
		\caption{$\mathrm{Corr}(\varepsilon_t,v_t)=0$}
	\end{subfigure} \quad
	\begin{subfigure}{0.47\textwidth}
		\includegraphics[width=\textwidth]{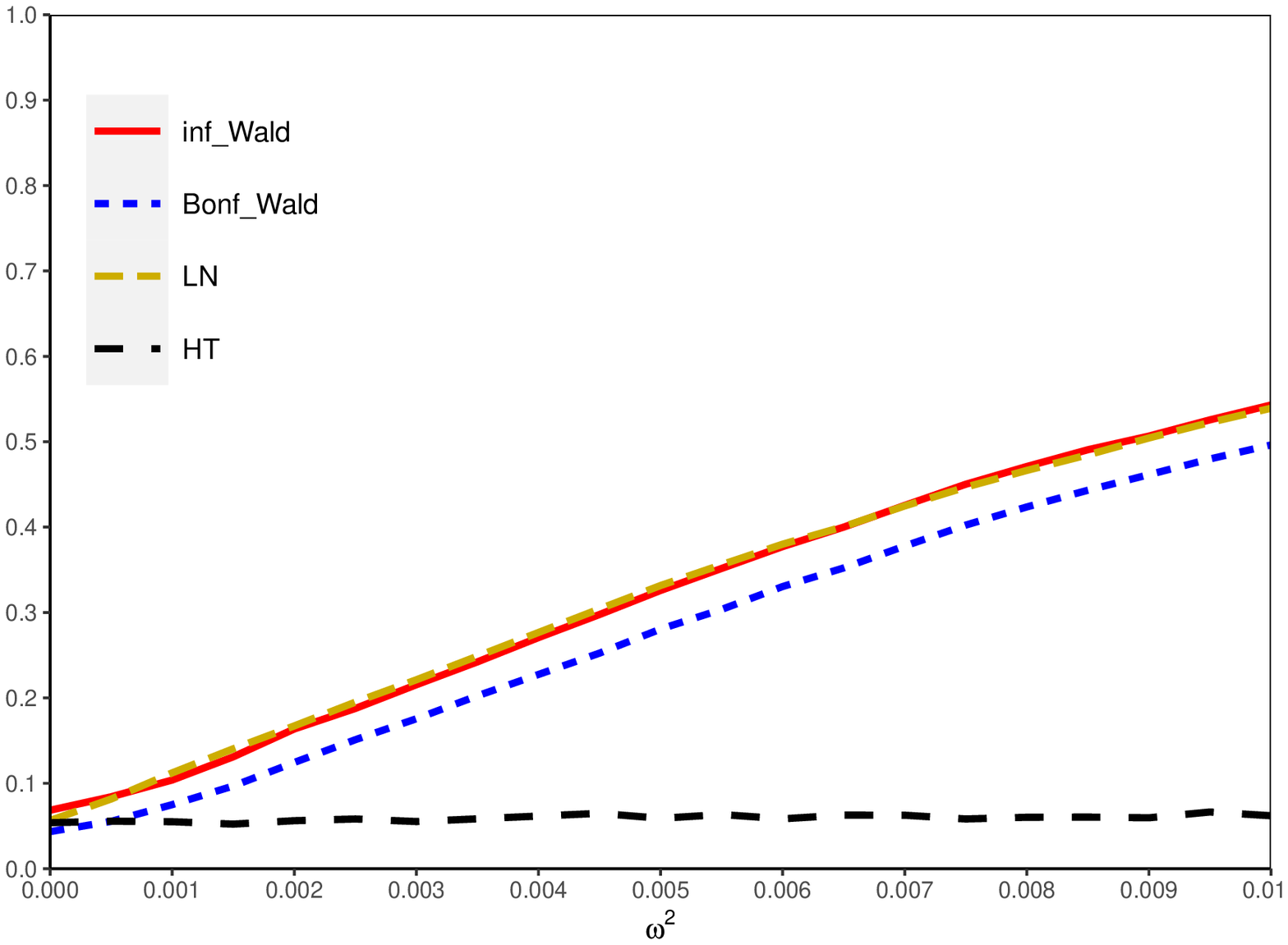}
		\caption{$\mathrm{Corr}(\varepsilon_t,v_t)=0.25$}
	\end{subfigure} \quad
	\begin{subfigure}{0.47\textwidth}
		\includegraphics[width=\textwidth]{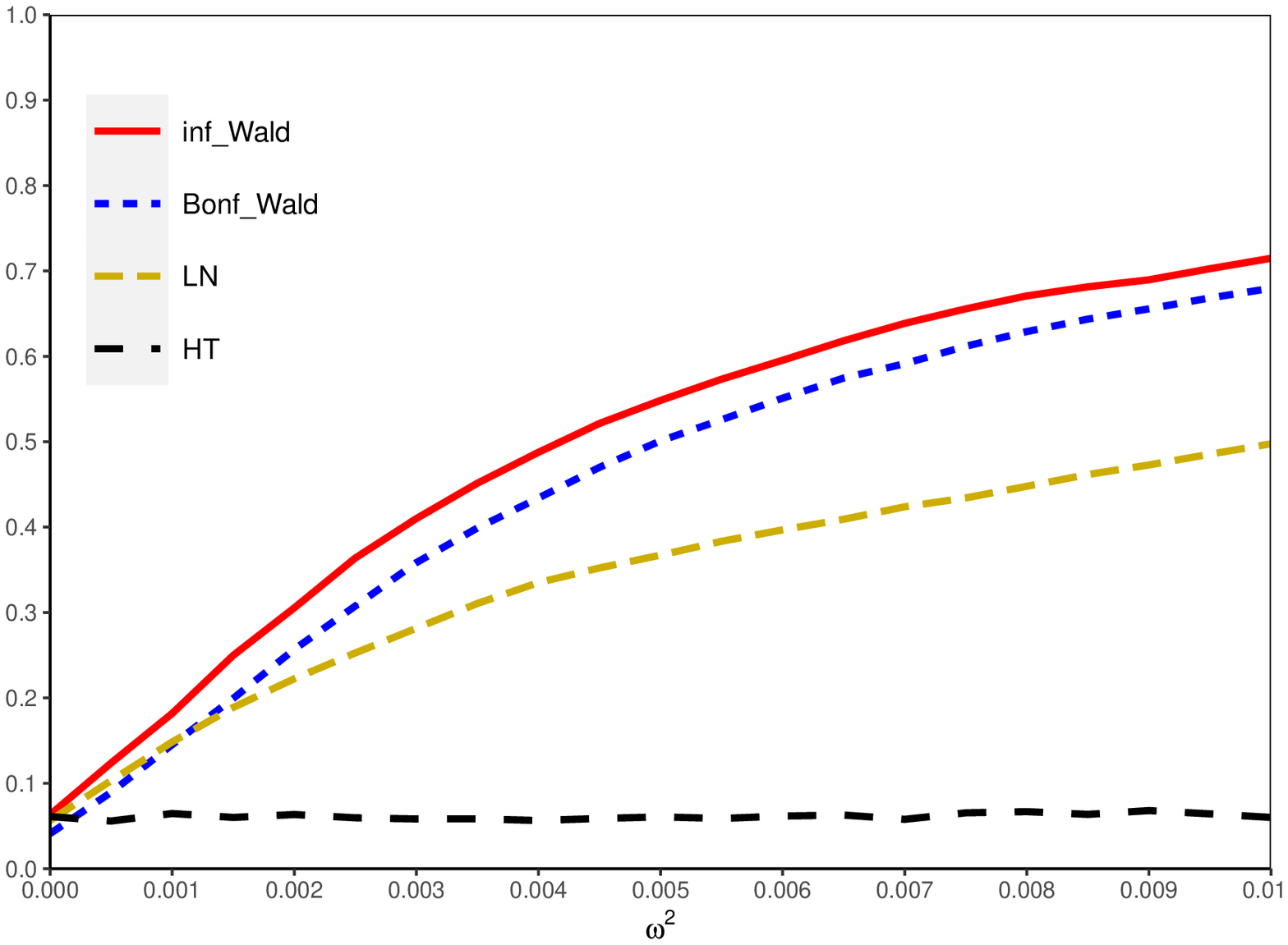}
		\caption{$\mathrm{Corr}(\varepsilon_t,v_t)=0.5$}
	\end{subfigure} \quad
	\begin{subfigure}{0.47\textwidth}
		\includegraphics[width=\textwidth]{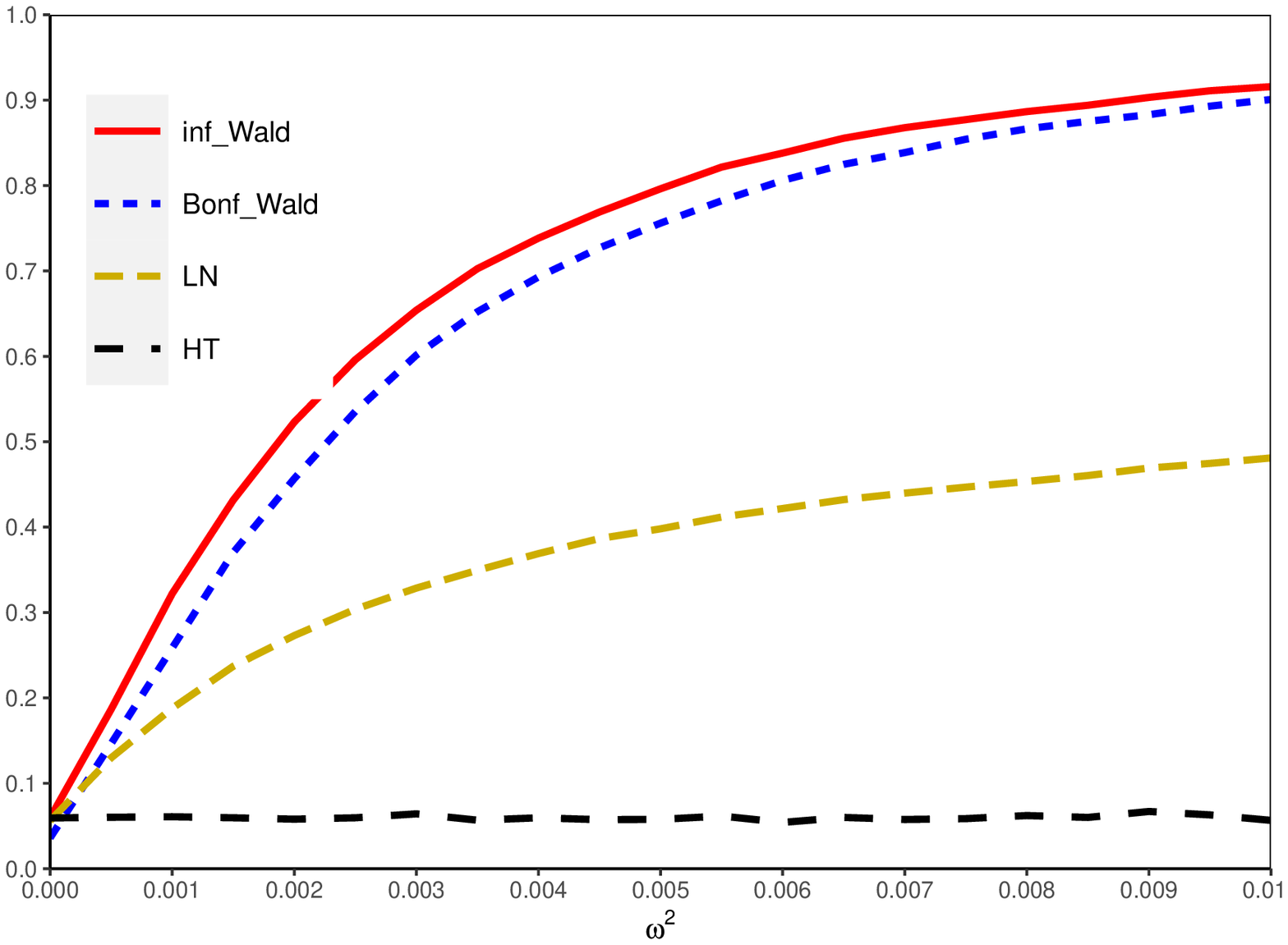}
		\caption{$\mathrm{Corr}(\varepsilon_t,v_t)=0.75$}
	\end{subfigure} \quad
	\caption{Finite-sample power functions under $\rho=0.98$}
	\label{fig:powers_fs_098}
\end{sidewaysfigure}

\clearpage
\begin{sidewaysfigure}
	\centering
	\begin{subfigure}{0.47\textwidth}
		\includegraphics[width=\textwidth]{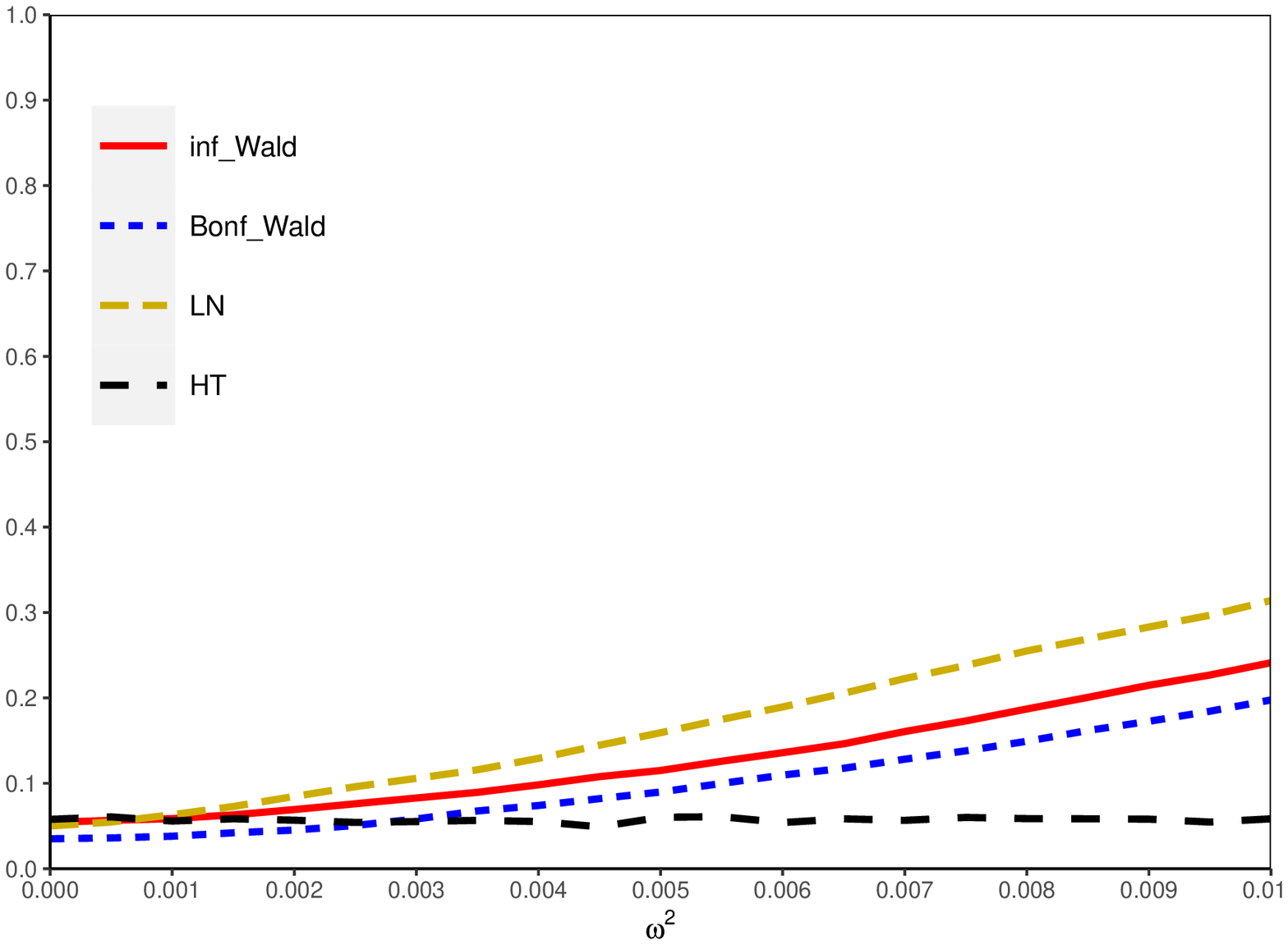}
		\caption{$\mathrm{Corr}(\varepsilon_t,v_t)=0$}
	\end{subfigure} \quad
	\begin{subfigure}{0.47\textwidth}
		\includegraphics[width=\textwidth]{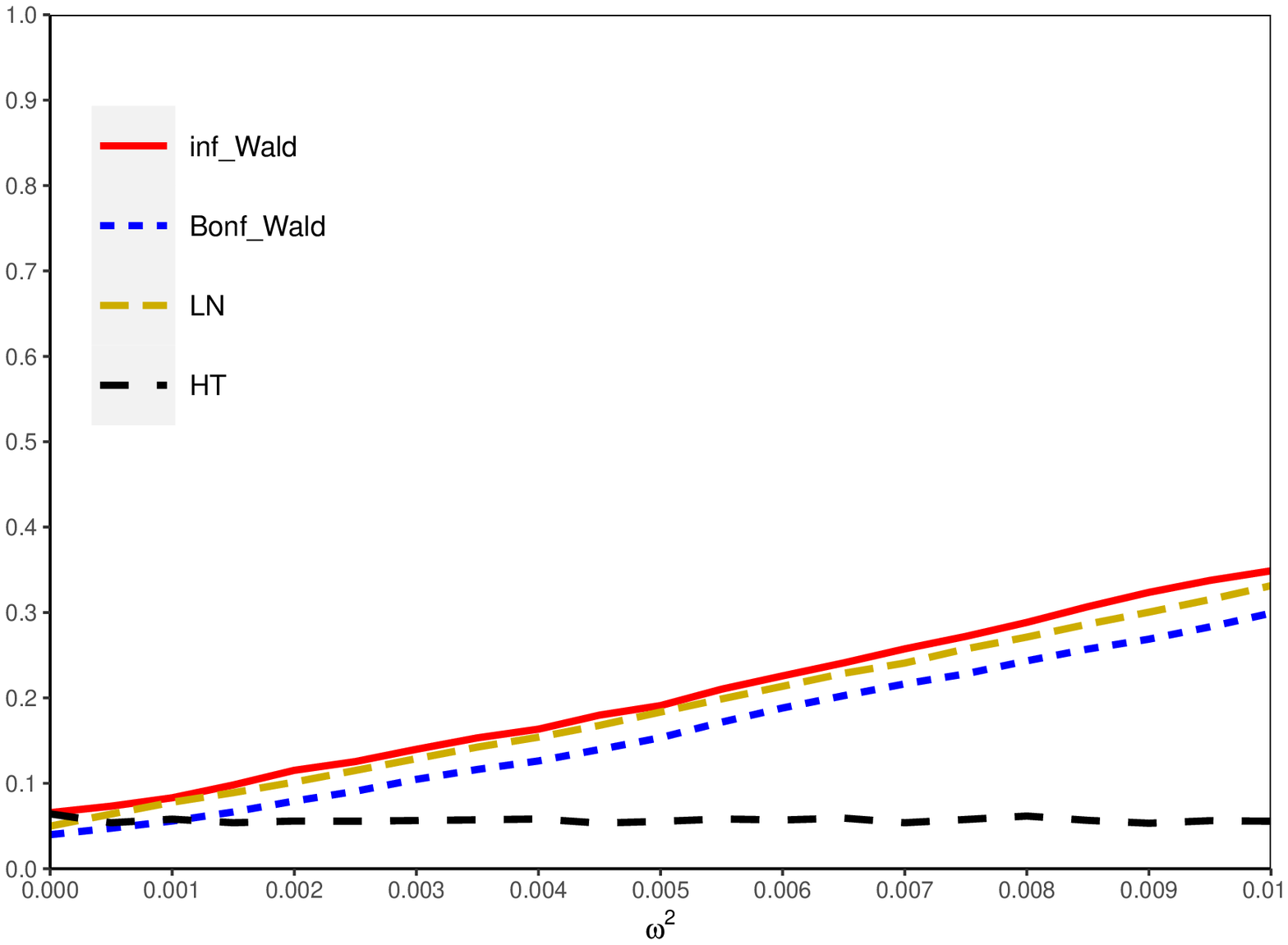}
		\caption{$\mathrm{Corr}(\varepsilon_t,v_t)=0.25$}
	\end{subfigure} \quad
	\begin{subfigure}{0.47\textwidth}
		\includegraphics[width=\textwidth]{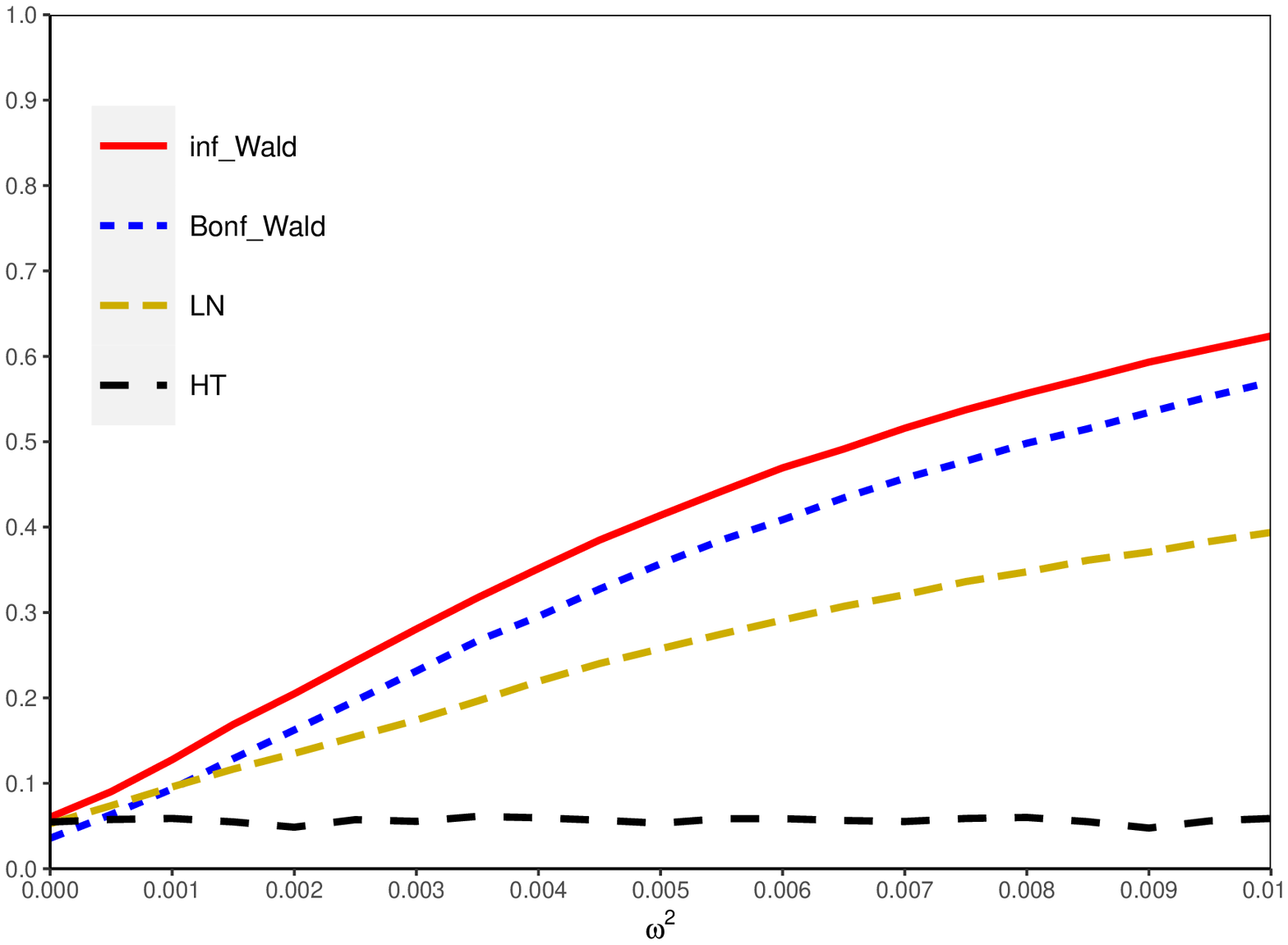}
		\caption{$\mathrm{Corr}(\varepsilon_t,v_t)=0.5$}
	\end{subfigure} \quad
	\begin{subfigure}{0.47\textwidth}
		\includegraphics[width=\textwidth]{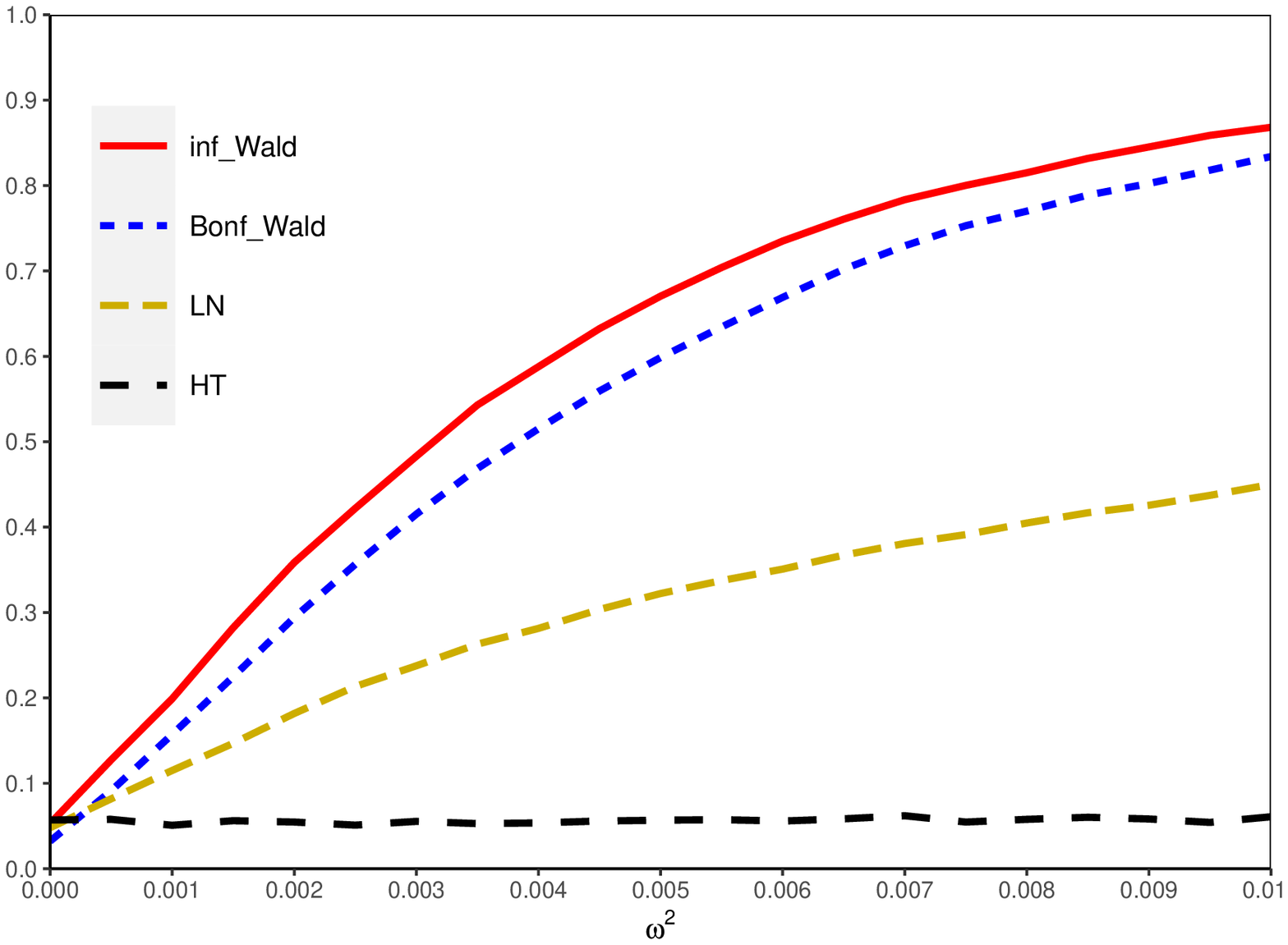}
		\caption{$\mathrm{Corr}(\varepsilon_t,v_t)=0.75$}
	\end{subfigure} \quad
	\caption{Finite-sample power functions under $\rho=0.95$}
	\label{fig:powers_fs_095}
\end{sidewaysfigure}

\end{document}